\documentclass[11pt,fleqn,a4paper]{article} 

\usepackage{amsmath,amssymb,amsthm,amsfonts,cancel} 
\usepackage[mathscr]{eucal}

\usepackage{hyperref}
\hypersetup{colorlinks, linkcolor=blue, citecolor=blue, urlcolor=blue}

\allowdisplaybreaks

\newcommand{\p}{\partial}

\newcommand{\sgn}{\mathop{\rm sgn}\nolimits}
\newcommand{\const}{{\rm const}}
\newcommand{\lsemioplus}{\mathbin{\mbox{$\lefteqn{\hspace{.77ex}\rule{.4pt}{1.2ex}}{\in}$}}}

\newlength{\mylength}
\newcommand{\solution}{\hspace*{-\mylength}\bullet\quad}

\newtheorem{theorem}{Theorem}
\newtheorem{lemma}[theorem]{Lemma}
\newtheorem{corollary}[theorem]{Corollary}
\newtheorem{proposition}[theorem]{Proposition}
{\theoremstyle{definition}

\newtheorem{remark}[theorem]{Remark}
}

\makeatletter
\long\def\@makecaption#1#2{%
  \vskip\abovecaptionskip\footnotesize
  \sbox\@tempboxa{#1. #2}%
  \ifdim \wd\@tempboxa >\hsize
    #1. #2\par
  \else
    \global \@minipagefalse
    \hb@xt@\hsize{\hfil\box\@tempboxa\hfil}%
  \fi
  \vskip\belowcaptionskip}
\makeatother

\newcommand{\todo}[1][\null]{\ensuremath{\clubsuit}}

\newcommand{\noprint}[1]{}

\begin{document}

\par\noindent {\LARGE\bf
Surprising symmetry properties and exact solutions\\
of Kolmogorov backward equations\\ with power diffusivity
\par}

\vspace{5.5mm}\par\noindent{\large
\large Serhii D. Koval$^{\dag\ddag}$, Elsa Dos Santos Cardoso-Bihlo$^\dag$ and Roman O. Popovych$^{\S\ddag}$
}

\vspace{5.5mm}\par\noindent{\it\small
$^\dag$Department of Mathematics and Statistics, Memorial University of Newfoundland,\\
$\phantom{^\ddag}$\,St.\ John's (NL) A1C 5S7, Canada
\par}

\vspace{2mm}\par\noindent{\it\small
$^\S$\,Mathematical Institute, Silesian University in Opava, Na Rybn\'\i{}\v{c}ku 1, 746 01 Opava, Czech Republic
\par}
\vspace{2mm}\par\noindent{\it\small	
$^\ddag$\,Institute of Mathematics of NAS of Ukraine, 3 Tereshchenkivska Str., 01024 Kyiv, Ukraine
\par}

\vspace{4mm}\par\noindent
E-mails:
skoval@mun.ca,
ecardosobihl@mun.ca,
rop@imath.kiev.ua

\vspace{8mm}\par\noindent\hspace*{10mm}\parbox{140mm}{\small
Using the original advanced version of the direct method, we efficiently compute
the equivalence groupoids and equivalence groups
of two peculiar classes of Kolmogorov backward equations with power diffusivity
and solve the problems of their complete group classifications.
The results on the equivalence groups are double-checked with the algebraic method.
Within these classes, the remarkable Fokker--Planck and the fine Kolmogorov backward equations are distinguished
by their exceptional symmetry properties.
We extend the known results on these two equations to their counterparts
with respect to a nontrivial discrete equivalence transformation.
Additionally, we carry out Lie reductions of the equations under consideration
up to the point equivalence, exhaustively study their hidden Lie symmetries
and generate wider families of their new exact solutions
via acting by their recursion operators on constructed Lie-invariant solutions.
This analysis reveals eight powers of the space variable
with exponents $-1$, 0, 1, 2, 3, 4, 5 and~6 as values of the diffusion coefficient
that are prominent due to symmetry properties of the corresponding equations.
}\par\vspace{4mm}

\noprint{
Keywords:
(1+2)-dimensional ultraparabolic Kolmogorov backward equations;
group classification;
equivalence groupoid;
equivalence group;
point-symmetry group;
exact solutions;
Lie reductions;
Darboux transformation

MSC: 35B06, 35A30, 35C05, 35K10, 35K70, 35C06
35-XX   Partial differential equations
  35A30   Geometric theory, characteristics, transformations [See also 58J70, 58J72]
  35B06   Symmetries, invariants, etc.
 35Cxx  Representations of solutions
  35C05   Solutions in closed form
  35C06   Self-similar solutions
 35Kxx	Parabolic equations and parabolic systems
  35K10   Second-order parabolic equations
  35K70   Ultraparabolic equations, pseudoparabolic equations, etc.
}

\section{Introduction}

The Kolmogorov backward equations originated from the foundational work by Kolmogorov~\cite{kolm1931a},
where he extended the theory of discrete-time Markov processes to its continuous counterpart.
Among the latter, Kolmogorov distinguished diffusion processes and, for describing their evolution,
derived a pair of classes of equations,
usually referred to as Kolmogorov forward (or Fokker--Planck) and Kolmogorov backward equations.
Since these equations capture fundamental features of diffusion processes,
they became the cornerstones of many branches of physics, natural sciences and even economics.
For instance, in physics these equations are instrumental in modeling lasers
and other nonlinear systems far from thermal equilibrium~\cite[Chapter~12]{risk1989A}.
Specifically for the Kolmogorov backward equations,
they are applied in evolutionary biology, ecology and environmental science
to model and predict the dynamics of ecosystems and populations.\looseness=-1

In the present paper, we carry out extended symmetry analysis  of
the ultraparabolic (1+2)-dimensional Kolmogorov backward equations of the form
\begin{gather}\label{eq:FPsubclass}
\mathcal F_{\alpha\beta}\colon\ u_t+xu_y=|x-\alpha|^\beta u_{xx},
\end{gather}
which were formally introduced in~\cite{zhan2020a}.%
\footnote{%
The formally adjoint to the Kolmogorov forward equation $u_t+xu_y=(|x-\alpha|^\beta u)_{xx}$
is the (true) Kolmogorov backward equation $-u_t-xu_y=|x-\alpha|^\beta u_{xx}$.
For convenience, we simultaneously change the signs of~$t$ and~$y$,
obtaining the equation~\eqref{eq:FPsubclass}.
This change is not essential from the point of view of symmetry analysis of differential equations,
and we preserve the name ``Kolmogorov backward equation'' for the equations of the form~\eqref{eq:FPsubclass}.
}
Here $\alpha$ and $\beta$ are arbitrary real parameters.
The more general Kolmogorov backward equations $u_t+xu_y=a(x)u_{xx}+b(x)u_x$,
appeared in~\cite[Eq.~(30)]{rote83a} as models for describing the evolution of cell populations,
and then only the values $a=\const$ and $b=0$ were selected as the simplest ones \cite[Eq.~(32)]{rote83a}.
Equations of the form~\eqref{eq:FPsubclass} with~$\alpha=0$ also arise
in the course of considering path-dependent options in a constant elasticity variance environment
\cite[Section~5.1]{baru2001a}.
It is also shown in~\cite{baru2001a} that
the Kolmogorov equations modeling the price of the geometric and arithmetic average strike call option
are reduced by point transformations to the above equations, where in addition $\beta=0$ and $\beta=2$, respectively.

Furthermore, the equations of the form~\eqref{eq:FPsubclass} are important beyond their practical applications.
They also constitute a family of interesting objects for theoretical studying
within the field of symmetry analysis of differential equations.
The study of their symmetries provides powerful insights
into the dynamics of real-world phenomena these equations model.
Moreover, it offers the potential for constructing their exact solutions.

Switching to the established rigorous framework of group analysis of differential equations,
the main object of the present paper is the class of differential equations of the form~\eqref{eq:FPsubclass}
denoted by~$\mathcal F$ with the tuple of arbitrary elements $\theta:=(\alpha,\beta)$,
whose components as constants run through the solution set of the auxiliary system of differential equations
$\alpha_t=\alpha_x=\alpha_y=0$, $\beta_t=\beta_x=\beta_y=0$.
See \cite{bihl2012b,bihlo2016a,opan2022a,popo2010a,vane2020a} and references therein
for the definition of classes of differential equations,
a theoretical background on their transformational properties and related computational techniques.
In contrast to the preliminary consideration of the class~$\mathcal F$ in~\cite{zhan2020a},
we compute the (usual) equivalence group~$G^\sim_{\mathcal F}$
and the equivalence groupoid~$\mathcal G^\sim_{\mathcal F}$ of this class
and thus find out unobvious discrete equivalences between equations from~$\mathcal F$.
Moreover, we easily show that any equation from the class~$\mathcal F$
is $G^\sim_{\mathcal F}$-equivalent to the equation
\begin{gather}\label{eq:FPGaugedClass}
\mathcal F'_\beta\colon\ \ u_t+xu_y = |x|^\beta u_{xx}
\end{gather}
with the same~$\beta$.
The equations~$\mathcal F'_\beta$ constitute the class~$\mathcal F'$,
where~$\beta$ remains the sole arbitrary element.
The above similarity between equations from the classes $\mathcal F$ and $\mathcal F'$
means that these classes are \emph{weakly similar} \cite[Definition~3]{opan2022a}.
This is why their group classifications modulo their equivalence groupoids are in fact the same,
but this is not the case for the analogous classifications modulo their equivalence groups.
To relate the latter classifications, we perform more thorough analysis.
As a result, we solve the complete group classification problems for the classes~$\mathcal F$ and~$\mathcal F'$
modulo both of the above kinds of equivalences.
Developing results of~\cite{kova2023a,kova2024a,popo2024b},
we describe other symmetry properties of equations from~$\mathcal F$
and construct wide families of their exact solutions in closed form.
Even though the arbitrary elements of the classes $\mathcal F'$ and $\mathcal F$
are just constant parameters,
symmetry analysis of these classes and their single elements is highly nontrivial.

More specifically, we begin solving the equivalence problem for the classes~$\mathcal F$ and $\mathcal F'$
by describing properties of their equivalence groupoids~$\mathcal G^\sim_{\mathcal F}$ and~$\mathcal G^\sim_{\mathcal F'}$
in Section~\ref{sec:EquivGroupoids} using the advanced version of the direct method.
A disadvantage of the direct method is
that its straightforward application to the classes~$\mathcal F$ and $\mathcal F'$
leads to tedious and cumbersome computations.
An additional complication is provided by the fact that these classes are not normalized in the usual sense.
Nevertheless, the problem complexity can be reduced in the following way.
We first embed the classes~$\mathcal F$ and $\mathcal F'$
in the class~$\bar{\mathcal F}$ of the linear (1+2)-dimensional second-order ultraparabolic equations of the form
\begin{gather}\label{eq:FPclass}
\begin{split}
&u_t+B(t,x,y)u_y=A^2(t,x,y)u_{xx}+A^1(t,x,y)u_x+A^0(t,x,y)u+C(t,x,y) \\
&\text{with}\quad A^2\neq 0,\quad B_x\neq 0,
\end{split}
\end{gather}
which is normalized and whose equivalence pseudogroup $G^{\sim}_{\bar{\mathcal F}}$
and thus equivalence groupoid $\mathcal G^\sim_{\bar{\mathcal F}}$ are known~\cite{kova2023a,kova2024b}.
The tuple $\bar\theta:=(B,A^2,A^1,A^0,C)$ of arbitrary elements of the class~$\bar{\mathcal F}$
runs through the solution set of the system of the inequalities $A^2\neq0$ and $B_x\neq0$ with no restrictions on~$A^0$, $A^1$ and~$C$.
The class embeddings lead to the corresponding groupoid embeddings
${\mathcal G^\sim_{\mathcal F}\hookrightarrow\mathcal G^\sim_{\bar{\mathcal F}}}$  and
${\mathcal G^\sim_{\mathcal F'}\hookrightarrow\mathcal G^\sim_{\bar{\mathcal F}}}$.
We use the description of~$\mathcal G^\sim_{\bar{\mathcal F}}$ to derive basic properties
of the equivalence groupoids~$\mathcal G^\sim_{\mathcal F}$ and~$\mathcal G^\sim_{\mathcal F'}$
in Section~\ref{sec:EquivGroupoids},
which allows us to simplify the computation of the equivalence groups~$G^\sim_{\mathcal F}$ and~$G^\sim_{\mathcal F'}$
and the equivalence algebras~$\mathfrak g^\sim_{\mathcal F}$ and~$\mathfrak g^\sim_{\mathcal F'}$
of the classes~$\mathcal F$ and $\mathcal F'$
in Sections~\ref{sec:EquivGroupClassF} and~\ref{sec:EquivGroupClassF'}, respectively.
We show that the equivalence group $G^\sim_{\mathcal F'}$ contains the distinguished discrete equivalence transformation
\begin{gather}\label{eq:EssEquivTransOfF'}
\mathscr J'\colon\quad
\tilde t=y\sgn x,\quad
\tilde x=\frac1x,\quad
\tilde y=t\sgn x,\quad
\tilde u=\frac ux, \quad
\tilde\beta=5-\beta,
\end{gather}
which turns out to be the only point equivalence transformation that is essential for carrying out
the group classification of the class~$\mathcal F'$.
Due to this transformation, we can impose the gauge $\beta\leqslant5/2$.

\begin{remark}\label{rem:SimplerEquivTrans}
Instead of the class~$\mathcal F'$, we can consider the subset of its equations,
where the constant arbitrary element~$\beta$ takes only rational values with odd denominators,
and replace such equations by their counterparts without the absolute value sign,
\begin{gather}\label{eq:FPWithRationalBeta}
u_t+xu_y=x^\beta u_{xx}.
\end{gather}
Then the equivalence transformation~\eqref{eq:EssEquivTransOfF'}
is replaced by the simpler transformation
\begin{gather}\label{eq:EssEquivTransMod}
\tilde t=y,\quad
\tilde x=\frac1x,\quad
\tilde y=t,\quad
\tilde u=\frac ux, \quad
\tilde\beta=5-\beta
\end{gather}
between equations of the form~\eqref{eq:FPWithRationalBeta}.
For the specific values $\beta=-1,1,3,5$, when they are considered separately from the generic values of~$\beta$,
we present exact solutions of the corresponding equation~\eqref{eq:FPWithRationalBeta}
but not~\eqref{eq:FPGaugedClass}.
\end{remark}

In Section~\ref{sec:AlgMet}, we also compute the equivalence groups~$G^\sim_{\mathcal F}$ and~$G^\sim_{\mathcal F'}$
using the automorphism-based version of the algebraic method,
which was originally proposed by Hydon \cite{hydo1998a,hydo1998b,hydo2000b,hydo2000A}
for finding the point symmetry groups of single systems of differential equations
and extended to finding the equivalence groups of classes of such systems in~\cite{bihl2015a}.
This gives the first example of such computation in the literature.
Since the dimensions of the corresponding equivalence algebras~$\mathfrak g^\sim_{\mathcal F}$ and~$\mathfrak g^\sim_{\mathcal F'}$
are relatively high, the straightforward construction of their automorphism groups is complicated,
perhaps even impossible.
To overcome this challenge, we first obtain a sufficient number of megaideals
of the algebras~$\mathfrak g^\sim_{\mathcal F}$ and~$\mathfrak g^\sim_{\mathcal F'}$,
which allowed us, after fixing bases consistent with the constructed megaideal hierarchies,
to preliminarily constraint automorphism matrices for proceeding with their computation in an efficient way.
In total, we observe that the construction of the groups~$G^\sim_{\mathcal F}$ and~$G^\sim_{\mathcal F'}$
by the algebraic method is more involved than that by the advanced version of the direct method.
Nevertheless, this serves as an excellent cross-verification of the correctness of the obtained results.
In addition, to the best of our knowledge, there has been no comparison in the literature of the capabilities
of these two methods when applied to a particular class of differential equations.

The results of the group classifications of the classes~$\mathcal F'$ and~$\mathcal F$
are presented in Section~\ref{sec:GroupClassificationClassF'}.
We prove that there are only two $G^\sim_{\mathcal F'}$-inequivalent cases of essential Lie symmetry extensions
in the class~$\mathcal F'$, with $\beta=0$ and~$\beta=2$.
The corresponding essential Lie invariance algebras~$\mathfrak g^{\rm ess}_0$ and~$\mathfrak g^{\rm ess}_2$
are eight- and five-dimensional, respectively.
The essential Lie invariance algebra~$\mathfrak g^{\rm ess}_\beta$ of the equation $\mathcal F'_\beta$
in the generic case,
where $\beta\in\mathbb R\setminus\{0,2,3,5\}$, is four-dimensional and is isomorphic to the algebra $A_{3.4}^a\oplus A_1$
with $a=(\beta-2)/(\beta-3)$ in the notation of~\cite{popo2003a}.
Among these Lie algebras with $\beta\in(-\infty,5/2]\setminus\{0,2\}$, there are no pairwise isomorphic ones.
Hence the corresponding equations are definitely $G^\sim_{\mathcal F'}$-inequivalent.
The above implies that
$\mathcal G^\sim_{\mathcal F}$-equivalent cases of essential Lie symmetry extensions
within the class~$\mathcal F$ are exhausted by the equations~$\mathcal F_{00}$ and~$\mathcal F_{02}$.
Due to the inconsistency of the groups~$G^\sim_{\mathcal F'}$ and~$G^\sim_{\mathcal F}$,
the classification list for the class~$\mathcal F$ up to the $G^\sim_{\mathcal F}$-equivalence
should also include the cases $\beta=5$ and~$\beta=3$ with $\alpha=0$,
which are $G^\sim_{\mathcal F'}$-equivalent to the cases $\beta=0$ and~$\beta=2$, respectively.

The equations~$\mathcal F'_0$ and~$\mathcal F'_2$ are also singled out within the class~$\bar{\mathcal F}$
by their exceptional Lie symmetry properties.
The dimension of~$\mathfrak g^{\rm ess}_0$ is maximum for the essential Lie invariance algebras of equations
from this class.
Moreover, the equation~$\mathcal F'_0$ is the unique equation
with this property in~$\bar{\mathcal F}$ modulo the point equivalence.
The algebra~$\mathfrak g^{\rm ess}_2$ is five-dimensional and nonsolvable
and up to the point equivalence, the equation~$\mathcal F'_2$ is the unique equation in the class~$\bar{\mathcal F}$
that satisfies this property.
This is why the equations~$\mathcal F'_0$ and~$\mathcal F'_2$ were referred to
as the {\it remarkable Fokker--Planck equation}
and the {\it fine Kolmogorov backward equation},
and their extended symmetry analysis was carried out in~\cite{kova2023a} and~\cite{kova2024a}, respectively.

In Sections~\ref{sec:RemarkableFPandCounterpart} and~\ref{sec:FineFPandCounterpart},
we outline and complete the results of~\cite{kova2023a,kova2024a,popo2024b}
on the equations~$\mathcal F'_0$ and~$\mathcal F'_2$,
including the solution generation for these equations
using their recursion operators related to their Lie symmetries.
We also generate new solution families of the equation~$\mathcal F'_2$
using Darboux transformations of its codimension-one submodels,
which are derived in Section~\ref{sec:DarbouxTrans}.
All the known results on the equations~$\mathcal F'_0$ and~$\mathcal F'_2$
are extended to their $G^\sim_{\mathcal F'}$-counterparts~$\mathcal F'_5$ and $\mathcal F'_3$, respectively.
In particular, this provides
the point symmetry pseudogroups and wide families of exact solutions of~$\mathcal F'_5$ and~$\mathcal F'_3$.
To be precise, we actually consider the equations of the form~\eqref{eq:FPWithRationalBeta}
with $\beta=5$ and $\beta=3$ instead of~$\mathcal F'_5$ and~$\mathcal F'_3$
and use the transformation~\eqref{eq:EssEquivTransMod} but not~\eqref{eq:EssEquivTransOfF'}.

\looseness=-1
Section~\ref{sec:GenericCase} is devoted to the extended symmetry analysis
of the equations $\mathcal F'_\beta$ with $\beta\in\mathbb R\setminus\{0,2,3,5\}$,
which constitute the case of zero essential Lie symmetry extension within the class~$\mathcal F'$.
An original combination of the direct and algebraic methods is applied in Section~\ref{sec:GenCaseSymGroup}
for constructing the point symmetry pseudogroups of these equations.
Jointly with the obtained descriptions of the analogous pseudogroups for the other values of~$\beta$,
$\beta\in\{0,2,3,5\}$,
it provides the exhaustive description of the fundamental groupoid~$\mathcal G^{\rm f}_{\mathcal F'}$
of the class~$\mathcal F'$.
Hence, combining it with results of Section~\ref{sec:EquivGroupoids}, we obtain the complete descriptions
of the equivalence groupoids~$\mathcal G^\sim_{\mathcal F'}$ and~$\mathcal G^\sim_{\mathcal F}$,
which are presented in Section~\ref{sec:Conclusion}.
After listing inequivalent subalgebras of the essential Lie invariance algebra~\smash{$\mathfrak g^{\rm ess}_\beta$}
of the equation~$\mathcal F'_\beta$ for each value of~$\beta$ from $\mathbb R\setminus\{0,2,3,5\}$
in Section~\ref{sec:SubalgGbeta},
we classify its Lie reductions of codimension one and two
in Sections~\ref{sec:FPGenBetaLieRedCoD1} and~\ref{sec:FPGenBetaLieRedCoD2}
up to the $G^\sim_{\mathcal F'}$-equivalence.
In the latter section, we also show that Lie reductions of codimension three
do not lead to new solutions of~$\mathcal F'_\beta$.
Involving the $G^\sim_{\mathcal F'}$-equivalence allows us to reduce the number of cases to be considered.
We have found hidden symmetries and constructed wide families of exact solutions of the equations $\mathcal F'_\beta$
with $\beta\in\mathbb R\setminus\{0,2,3,5\}$ that are associated with the codimension-one Lie reductions
to the linear $(1+1)$-dimensional heat equations with the zero or the inverse square potentials.
The most prominent among them are the solutions associated with the value~$\beta=1$
and their $G^\sim_{\mathcal F'}$-counterparts for $\beta=4$.
Another prominent case is given by some codimension-two Lie reductions for $\beta=-1$
and their $G^\sim_{\mathcal F'}$-counterparts for $\beta=6$,
which results in Whittaker equations.
What is particularly surprising is that these are both the prominent cases,
where constructed families of solutions can be significantly extended
using the action by recursion operators of~$\mathcal F'_\beta$ associated with Lie symmetries of~$\mathcal F'_\beta$,
which is the subject of Section~\ref{sec:FPGenBetaSolGeneration}.
In total, Section~\ref{sec:GenericCase} on its own gives
an exhaustive symmetry analysis of the equations $\mathcal F'_\beta$ for the generic case of $\beta$.

Summarizing the results of the present paper in Section~\ref{sec:Conclusion},
we highlight discovered remarkable and surprising properties
of the classes $\mathcal F'$ and $\mathcal F$ and their single elements,
emphasize novel features of our study and
outline research perspectives regarding these classes as well as the superclass~$\bar{\mathcal F}$.

For readers' convenience,
the constructed exact solutions of the equations of the form~\eqref{eq:FPGaugedClass} are marked by the bullet symbol~$\bullet$\,.
Given a class of differential equations,
$\pi$ denotes the natural projection of the joint space of the variables,
and the arbitrary elements onto the space of the variables only.

\section{Equivalence groupoids}\label{sec:EquivGroupoids}

We first compute the equivalence groupoids and the equivalence groups
of the classes~$\mathcal F$ and~$\mathcal F'$ using the advanced version of the direct method.
See, e.g., \cite{king1998a,vane2020a} and references therein for required notions and techniques.
Recall that the \emph{direct method} is based on the definition of admissible transformations
of a class of differential equations.
Let $\mathcal L|_{\mathcal S}=\{\mathcal L_\theta\mid\theta\in\mathcal S\}$ be a class
of systems of differential equations~$\mathcal L_\theta$
for unknown functions $u=(u^1,\ldots,u^m)$ of independent variables $x=(x_1,\ldots,x_n)$
with the arbitrary-element tuple~$\theta=(\theta^1,\dots,\theta^k)$
running through the solution set~$\mathcal S$ of an auxiliary system of differential equations and inequalities.
Suppose that a point transformation~$\Phi$ in the space with the coordinates $(x,u)$,
$\Phi$: $\tilde x=X(x,u)$, $\tilde u=U(x,u)$ with nonzero Jacobian $|\p(X,U)/\p(x,u)|$,
relates systems~$\mathcal L_\theta$ and $\mathcal L_{\tilde\theta}$ from the class~$\mathcal L|_{\mathcal S}$,
$\Phi_*\mathcal L_\theta=\mathcal L_{\tilde\theta}$.
Expressing the required derivatives of~$\tilde u$ with respect to~$\tilde x$ in terms of
the original derivatives of~$u$ with respect to~$x$ using the chain rule,
we substitute the derived expressions into the system $\mathcal L_{\tilde \theta}$,
obtaining the system $(\Phi^{-1})_*\mathcal L_{\tilde\theta}$,
which should be identically satisfied on the solutions of~$\mathcal L_{\theta}$.
The splitting of the obtained condition with respect to the parametric derivatives of the system~$\mathcal L_{\theta}$
usually results in an overdetermined coupled system of nonlinear differential equations,
which is highly difficult to solve if possible at all.
However, in the particular case of the classes~$\mathcal F$ and $\mathcal F'$,
we are able to perform the computations both efficiently and elegantly
using the embedding of these classes in the class $\bar{\mathcal F}$,
which is constituted by the equations of the form~\eqref{eq:FPclass},
and the explicit description of its known equivalence groupoid.
Embedding a class~$\mathcal L|_{\mathcal S}$ in a superclass with known or easier computed equivalence groupoid
for deriving principal constraints for admissible transformations within~$\mathcal L|_{\mathcal S}$,
which simplifies the further required computations, is the essence of the advanced version of the direct method.
A good choice for such a superclass is a normalized class of differential equations.

\begin{theorem}\label{thm:EquivalenceGroupFPsuperClass}
The class~$\bar{\mathcal F}$ is normalized.
Its (usual) equivalence pseudogroup~$G^\sim_{\bar{\mathcal F}}$ consists of the point transformations with the components
\begin{subequations}
\begin{gather}\label{eq:ClassFbarTransPart}
\tilde t=T(t,y),
\quad
\tilde x=X(t,x,y),
\quad
\tilde y=Y(t,y),
\quad
\tilde u=U^1(t,x,y)u + U^0(t,x,y),
\\[.5ex]\label{eq:A^0A^1Trans}
\tilde A^0=\frac{-U^1}{T_t+BT_y}E\frac1{U^1},
\quad
\tilde A^1=A^1\frac{X_x}{T_t+BT_y}-\frac{X_t+BX_y}{T_t+BT_y}+A^2\frac{X_{xx}-2X_xU^1_x/U^1}{T_t+BT_y},
\\[.5ex]\label{eq:A^2BCTrans}
\tilde A^2=A^2\frac{X_x^2}{T_t+BT_y},
\quad
\tilde B=\frac{Y_t+BY_y}{T_t+BT_y},
\quad
\tilde C=\frac{U^1}{T_t+BT_y}\left(C-E\frac{U^0}{U^1}\right),
\end{gather}
\end{subequations}
where $T$, $X$, $Y$, $U^1$ and $U^0$ are arbitrary smooth functions of their arguments with $(T_tY_y-T_yY_t)X_xU^1\ne0$,
and $E:=\p_t+B\p_y-A^2\p_{xx}-A^1\p_x-A^0$.
\end{theorem}

There are the following embeddings for~$\mathcal F'$, $\mathcal F$ and $\bar{\mathcal F}$
considered as sets of differential equations:
\[
\mathcal F'\hookrightarrow\mathcal F\hookrightarrow\bar{\mathcal F}.
\]
These embeddings arise from the embeddings
$\mathcal S_{\mathcal F'}\hookrightarrow\mathcal S_{\mathcal F}\hookrightarrow\mathcal S_{\bar{\mathcal F}}$
of the corresponding sets~$\mathcal S_{\mathcal F'}$, $\mathcal S_{\mathcal F'}$ and~$\mathcal S_{\bar{\mathcal F}}$
of arbitrary elements, which are respectively defined via the mappings
\[
\mathcal S_{\mathcal F'}\ni(\beta)\mapsto(0,\beta)\in\mathcal S_{\mathcal F},
\quad
\mathcal S_{\mathcal F} \ni (\alpha,\beta)\mapsto(x,0,0,0,|x-\alpha|^\beta)\in \mathcal S_{\bar{\mathcal F}},
\]
which can be treated as class reparameterizations.
Thus, the embedding $\mathcal F'\hookrightarrow\bar{\mathcal F}$ corresponds to their composition
$\mathcal S_{\mathcal F'}\hookrightarrow\mathcal S_{\bar{\mathcal F}}$.
The above maps also induce embeddings on the level of equivalence groupoids,
$\mathcal G^\sim_{\mathcal F'}\hookrightarrow\mathcal G^\sim_{\mathcal F}\hookrightarrow\mathcal G^\sim_{\bar{\mathcal F}}$.

Under the successive embeddings, the class~$\mathcal F'$
can be considered as a subclass of~$\mathcal F$ and a subclass~$\bar{\mathcal F}$,
which are singled out by the additional auxiliary constraints $\alpha=0$ and
$A^1=A^0=C=0$, $B=x$, $A^2_t=A^2_y=0$, $(x\ln|x|)A^2_x=A^2\ln A^2$, respectively.
At the same time, even after the embedding $\mathcal F\hookrightarrow\bar{\mathcal F}$,
the class $\mathcal F$ cannot be interpreted as a subclass of $\bar{\mathcal F}$
according to the formal definition of subclass given in~\cite{popo2010a}
since the set of $A^2$ of the form $A^2=|x-\alpha|^\beta$, where $\alpha$ and $\beta$ are arbitrary constants,
is the union of the solution sets of the systems
$A^2=1$ and $A^2_t=A^2_y=0$, $\Phi_{xx}=-\sgn(x-\alpha)\Phi_{x}\exp\left(\Phi\Phi_{xx}/\Phi_x^2\right)$ with $\Phi=\ln A^2$
but not the solution set of a regular system of differential equations and differential inequalities.
We can overcome this obstacle via extending the notion of class of differential equations
by allowing arbitrary class elements to run the union of solution sets of auxiliary systems
of differential equations and differential inequalities.

The aforementioned arguments allow us to use Theorem~\ref{thm:EquivalenceGroupFPsuperClass}
in the course of computing the equivalence groupoids~$\mathcal G^\sim_{\mathcal F}$ and~$\mathcal G^\sim_{\mathcal F'}$
and the equivalence groups~$G^\sim_{\mathcal F}$ and~$G^\sim_{\mathcal F'}$
of the classes~$\mathcal F$ and~$\mathcal F'$.
For this computation, the following elementary lemma is useful.

\begin{lemma}\label{lem:EqualityForArbitraryX}
If the equality	$C_1|x-\mu_1|^{\nu_1}|x-\mu_2|^{\nu_2} = C_2|x-\mu_3|^{\nu_3}$ holds,
where $C_1$, $C_2$, $\nu_1$, $\nu_2$, $\nu_3$, $\mu_1$, $\mu_2$, $\mu_3$ are constants with $C_1C_2\neq0$,
then $\mu_1=\mu_2=\mu_3$ and $\nu_3=\nu_1+\nu_2$, $C_1=C_2$.
\end{lemma}
\noprint{
\begin{proof}
Setting $x=\mu_1$ and $x=\mu_2$ in the equality we obtain $\mu_3=\mu_1$ and $\mu_3=\mu_2$, respectively, thus $\mu_1=\mu_2=\mu_3$.
The rest of the relations obviously follow from the previous one.
\end{proof}
}

\begin{theorem}\label{thm:EquivGroupoidsFF'Properties}
(i) The point transformations
$\mathscr S(c_1)$: $(\tilde t,\tilde x,\tilde y,\tilde u,\tilde\beta,\tilde\alpha)=(t,x+c_1,y+c_1t,u,\beta,\alpha+c_1)$
\noprint{
\begin{gather}\label{eq:EssEquivTransOfF}
\mathscr S(c_1)\colon\quad \tilde t=t,\quad
\tilde x=x+c_1,\quad
\tilde y=y+c_1t,\quad
\tilde u=u,\quad
\tilde\beta=\beta,\quad
\tilde\alpha=\alpha+c_1,
\end{gather}
}
where $c_1$ is arbitrary constant, constitute a one-parameter group of equivalence transformations
of the class~$\mathcal F$.

(ii) The wide family of admissible transformations
$\mathcal S_{\alpha\beta}:=\big((\alpha,\beta),\pi_*\mathscr S(-\alpha),(0,\beta)\big)$ of the class~$\mathcal F$
from the action groupoid of its equivalence group
maps this class onto the class~$\mathcal F'$ interpreted as a subclass of~$\mathcal F$.

(iii) The point transformation~$\mathscr J'$ given in~\eqref{eq:EssEquivTransOfF'} is a (discrete) equivalence transformation
of the class~$\mathcal F'$.

(iv) The class~$\mathcal F'$ is semi-normalized
with respect to the discrete equivalence subgroup generated by~$\mathscr J'$.
In other words, the equivalence groupoid~$\mathcal G^\sim_{\mathcal F'}$ of~$\mathcal F'$
is the Frobenius product of the action groupoid of this subgroup
and the fundamental equivalence groupoid~$\mathcal G^{\rm f}_{\mathcal F'}$ of~$\mathcal F'$.
\end{theorem}

\begin{proof}
Since items (i)--(iii) are straightforwardly verified,
except the discreteness of~$\mathscr J'$, which follows from item~(iv),
it suffices to prove only item~(iv).

The determining equations for admissible transformations of the class~$\mathcal F'$
can be derived by restricting their counterparts for the class~$\bar{\mathcal F}$
in view of the embedding~$\mathcal F'$ in~$\bar{\mathcal F}$.
Setting $(A^0,A^1,A^2,B,C)=(0,0,|x|^\beta,x,0)$ and
$(\tilde A^0,\tilde A^1,\tilde A^2,\tilde B,\tilde C)=(0,0,|\tilde x|^{\tilde\beta},\tilde x,0)$
in~\eqref{eq:A^0A^1Trans} and~\eqref{eq:A^2BCTrans},
we obtain
\begin{gather}\label{eq:DetEqsForEquivGroupoidOfF'}
\begin{split}
&X = \frac{Y_t+xY_y}{T_t+xT_y},
\quad
|X|^{\tilde\beta}=|x|^\beta\frac{(T_tY_y-T_yY_t)^2}{(T_t +xT_y)^5},
\quad
X_t+xX_y=|x|^\beta\left(X_{xx}-2X_x\frac{U^1_x}{U^1}\right),
\\
&\left(\frac1{U^1}\right)_t+x\left(\frac1{U^1}\right)_y=|x|^\beta\left(\frac1{U^1}\right)_{xx},
\quad
\left(\frac{U^0}{U^1}\right)_t+x\left(\frac{U^0}{U^1}\right)_y=|x|^\beta\left(\frac{U^0}{U^1}\right)_{xx}.
\end{split}
\end{gather}
In view of Lemma~\ref{lem:EqualityForArbitraryX}, the consequence
\begin{gather}\label{eq:DetEqsForEquivGroupoidOfF'Conseq}
|Y_t +xY_y|^{\tilde\beta}=(T_tY_y-T_yY_t)^2|x|^\beta|T_t +xT_y|^{\tilde\beta-5}\sgn(T_t +xT_y)
\end{gather}
of the first two equations of~\eqref{eq:DetEqsForEquivGroupoidOfF'} implies that
$\tilde\beta=\beta$ if $T_y=0$ or
$\tilde\beta=5-\beta$ if $T_y\ne0$ and $Y_y=0$ or
$\beta=5$ and $\tilde\beta\in\{0,5\}$ if $T_yY_y\ne0$.
Up to composing any admissible transformation of~$\mathcal F'$
with elements of the discrete equivalence subgroup generated by~$\mathscr J'$,
we can assume that $\tilde\beta=\beta$,
which means that this admissible transformation belongs to~$\mathcal G^{\rm f}_{\mathcal F'}$.
\end{proof}

\begin{corollary}\label{cor:EquivGroupoidsFF'Properties}
(i) Different equations~\smash{$\mathcal F'_\beta$} and~\smash{$\mathcal F'_{\tilde\beta}$}
are similar with respect to point transformations if and only if $\beta+\tilde\beta=5$.

(ii) Equations~\smash{$\mathcal F_{\alpha\beta}$} and~\smash{$\mathcal F_{\tilde\alpha\tilde\beta}$}
are similar with respect to point transformations if and only if either $\tilde\beta=\beta$ or $\beta+\tilde\beta=5$.

(iii) The equivalence groupoid~$\mathcal G^\sim_{\mathcal F}$ of~$\mathcal F$
is generated by admissible transformations~$\mathcal S_{\alpha\beta}$ and elements of~$\mathcal G^\sim_{\mathcal F'}$.
More specifically, for each admissible transformation
$\big((\alpha,\beta),\Phi,(\tilde\alpha,\tilde\beta)\big)$ of~$\mathcal F$,
we have $\Phi=\pi_*\mathscr S(\tilde\alpha)\circ\breve\Phi\circ\pi_*\mathscr S(-\alpha)$
for some point transformation~$\breve\Phi$ with $(\beta,\breve\Phi,\tilde\beta)\in\mathcal G^\sim_{\mathcal F'}$.
\end{corollary}

\section{Equivalence group and equivalence algebra\\ of the original class}\label{sec:EquivGroupClassF}

We compute the equivalence group of the class~$\mathcal F$ of Kolmogorov backward equations with power diffusivity
using two fundamentally different approaches,
the \emph{direct method} in this section
and the \emph{algebraic method} in Section~\ref{sec:AlgMetClassF}.
While each of them has its advantages and drawbacks as well as its scope of applicability,
they complement each other
and provide a robust cross-verification of the correctness of computations when used in conjunction.
Moreover, in the literature, no one has compared the capabilities of these two methods
when applied to constructing the equivalence group of a specific class;
see~\cite{bihl2011b} for a similar comparison in the case
of computing the point symmetry group of a single system of differential equation.
The class~$\mathcal F$ can serve as an excellent instructive example for this particular purpose.

\begin{theorem}\label{thm:EqGroupFPsubclass}
The equivalence group~$G^\sim_{\mathcal F}$ consists of the point transformations
\begin{gather}\label{eq:EquivalenceGroupFPpowerDif}
\begin{split}
&\tilde t=t+c_0,\quad
\tilde x=\epsilon x+c_1,\quad
\tilde y=\epsilon y+c_1t+c_2,\quad
\tilde u=c_3u+c_4(tx-y)+c_5x+c_6,\\
&\tilde\alpha=\epsilon\alpha+c_1,\quad
\tilde\beta=\beta,
\end{split}
\end{gather}
where  $\epsilon=\pm1$, $c_0$, \dots, $c_6$ are arbitrary constants with $c_3\neq0$.
\end{theorem}

\begin{proof}
We apply the advanced version of the direct method.
In view of the embedding $\mathcal F\hookrightarrow\bar{\mathcal F}$,
the transformational parts of admissible transformations of the class~$\mathcal F$ takes the form~\eqref{eq:ClassFbarTransPart},
where the involved parameter functions satisfy the equations~\eqref{eq:A^0A^1Trans} and~\eqref{eq:A^2BCTrans}
specified for this class as a subclass of~$\bar{\mathcal F}$,
\begin{gather}\label{eq:DetEqsForAdmTransOfF}
\begin{split}&
X = \frac{Y_t+xY_y}{T_t+xT_y},
\quad
|X-\tilde\alpha|^{\tilde\beta}=|x-\alpha|^\beta\frac{(T_tY_y-T_yY_t)^2}{(T_t +xT_y)^5},
\\[.5ex]&
X_t+xX_y=|x-\alpha|^\beta\left(X_{xx}-2X_x\frac{U^1_x}{U^1}\right),
\quad
\left(\frac1{U^1}\right)_t+x\left(\frac1{U^1}\right)_y=|x-\alpha|^\beta\left(\frac1{U^1}\right)_{xx},
\\[.5ex]&
\left(\frac{U^0}{U^1}\right)_t+x\left(\frac{U^0}{U^1}\right)_y=|x-\alpha|^\beta\left(\frac{U^0}{U^1}\right)_{xx}.
\end{split}
\end{gather}
The first two equations of~\eqref{eq:DetEqsForAdmTransOfF} imply
\begin{gather}\label{eq:YTconstraint}
\big|Y_t-\tilde\alpha T_t+x(Y_y-\tilde\alpha T_y)\big|^{\tilde\beta}
=(Y_yT_t-Y_tT_y)^2|x-\alpha|^\beta|T_t+xT_y|^{\tilde\beta-5}\sgn(T_t +xT_y).
\end{gather}
When looking for the equivalence transformation of~$\mathcal F$,
if only one of the pairs $(\alpha,\beta)$ or $(\tilde\alpha,\tilde\beta)$
appears in a derived equation, we can split it with respect to this pair.
In view of item~(ii) of Corollary~\ref{cor:EquivGroupoidsFF'Properties},
we have at most two cases for the $\beta$-components of equivalence transformation of~$\mathcal F$,
$\tilde\beta=\beta$ and $\tilde\beta=5-\beta$.
In the latter case, Lemma~\ref{lem:EqualityForArbitraryX}, the equation~\eqref{eq:YTconstraint}
and the nondegeneracy condition $T_tY_y-T_yY_t\ne0$
successively imply under varying~$\beta$ that $T_y=0$, $T_t>0$, $Y_y\ne0$, $|Y_y|^{\beta-2}=T_t^{\beta-3}$
and $\tilde\alpha T_t-Y_t=\alpha Y_y$.
Differentiating the last equation with respect to~$y$ and splitting the result with respect to~$\alpha$,
we obtain $Y_{ty}=Y_{yy}=0$.
Then we differentiate the same equation with respect to~$t$ and split the result with respect to~$\tilde\alpha$
to derive $T_{tt}=Y_{tt}=0$.
Now the equation $|Y_y|^{\beta-2}=T_t^{\beta-3}$ splits with respect to~$\beta$
into $T_t=1$ and $Y_y=\epsilon:=\pm1$.
Therefore, $T=t+c_0$, $Y=\epsilon y+c_1t+c_2$, and thus $X=\epsilon x+c_1$
for some constants~$c_0$, $c_1$ and~$c_2$.
The third equation of~\eqref{eq:DetEqsForAdmTransOfF} reduces to the equation $U^1_x=0$,
in view of which the fourth equation of~\eqref{eq:DetEqsForAdmTransOfF}
splits with respect to~$x$ into $U^1_t=U^1_y=0$.
Since $U^1$ is a constant, say $c_3$,
the last equation of~\eqref{eq:DetEqsForAdmTransOfF} splits with respect to~$\beta$ into
$U^0_{xx}=0$ and $U^0_y+xU^0_y=0$, i.e., $U^0=U^{01}(t,y)x+U^{00}(t,y)$,
where $U^{00}_t=U^{01}_y=U^{01}_t+U^{00}_y=0$.
Integrating these equation, we obtain $U^0=c_4(tx-y)+c_5x+c_6$
for some constants~$c_4$, $c_5$ and~$c_6$.
\end{proof}

\begin{remark}
The transformations of the form~\eqref{eq:EquivalenceGroupFPpowerDif}
with $\epsilon=1$ and $c_1=0$
constitute a normal subgroup~$\hat G^\cap_{\mathcal F}$ of $G^\sim_{\mathcal F}$,
which is isomorphic, via the natural projection onto the space $\mathbb R^4_{t,x,y,u}$,
to the kernel point symmetry group~$G^\cap_{\mathcal F}$ of the class~$\mathcal F$,
cf.\ \cite[Proposition~3]{card2011a}.
The quotient of~$G^\sim_{\mathcal F}$ with respect to~\smash{$\hat G^\cap_{\mathcal F}$} is isomorphic to the group
$\mathbb R\times\mathbb Z_2$
and can be naturally identified with the subgroup of~$G^\sim_{\mathcal F^\prime}$
that is generated by the transformations~$\mathscr J'$ and
\begin{gather*}
\mathscr S(c_1)\colon\quad
\tilde t=t,\quad
\tilde x=\epsilon x+c_1,\quad
\tilde y=\epsilon y+c_1t,\quad
\tilde u=u,\quad
\tilde\alpha=\epsilon\alpha+c_1,\quad
\tilde\beta=\beta.
\end{gather*}
Hence, these transformations can be considered as the only
essential equivalence transformations within the class~$\mathcal F$,
which establish relations between different elements of~$\mathcal F$.
\end{remark}

Knowing the equivalence group~$G^\sim_{\mathcal F}$ of the class $\mathcal F$,
we can easily construct the equivalence algebra~$\mathfrak g^\sim_{\mathcal F}$ of this class
as the set of generators of one-parameter subgroups of~$G^\sim_{\mathcal F}$.

\begin{corollary}\label{cor:EquivAlgF}
The equivalence algebra~$\mathfrak g^\sim_{\mathcal F}$ of the class~$\mathcal F$
is spanned by the vector fields
$\p_t$,
$\p_x+t\p_y+\p_\alpha$,
$\p_y$,
$u\p_u$,
$(tx-y)\p_u$,
$x\p_u$ and
$\p_u$.
\end{corollary}

\begin{corollary}\label{cor:DisceteEquivTransOfF}
The usual equivalence group~$G^\sim_{\mathcal F}$ of the class $\mathcal F$ is generated
by the one-parameter groups associated with the basis vector fields of
the equivalence algebra $\mathfrak g^\sim_{\mathcal F}$ of this class
and its two discrete equivalence transformations,
\begin{gather*}
\mathscr I_u\colon(t,x,y,u,\alpha,\beta)\mapsto(t,x,y,-u,\alpha,\beta)\quad
\mbox{and}\quad
\mathscr I_{\rm s}\colon(t,x,y,u,\alpha,\beta)\mapsto(t,-x,-y,u,-\alpha,\beta).
\end{gather*}
\end{corollary}

The algebra~$\mathfrak g^\sim_{\mathcal F}$ can be represented as the semidirect sum
$\mathfrak g^\sim_{\mathcal F}=\langle\p_t+t\p_y+\p_\alpha\rangle\lsemioplus\hat{\mathfrak g}^\cap_{\mathcal F}$,
where the ideal~$\hat{\mathfrak g}^\cap_{\mathcal F}$ spanned by the vector fields
$\p_t$, $\p_y$, $u\p_u$, $(tx-y)\p_u$, $x\p_u$ and~$\p_u$
is isomorphic to the kernel invariance algebra~$\mathfrak g^\cap_{\mathcal F}$ of equations from the class~$\mathcal F$,
see Corollary~\ref{cor:GroupClassificationOfF}.
Both the algebras $\mathfrak g^\sim_{\mathcal F}$ and $\hat{\mathfrak g}^\cap_{\mathcal F}$
are indecomposable, two-step solvable and not nilpotent.
The nilradicals of the algebras $\mathfrak g^\sim_{\mathcal F}$ and $\hat{\mathfrak g}^\cap_{\mathcal F}$
respectively are
\[\mathfrak n^\sim_{\mathcal F}=\langle\p_x+t\p_y+\p_\alpha,\p_t,\p_y,(tx-y)\p_u,x\p_u,\p_u\rangle
\quad\mbox{and}\quad
\mathfrak n^\cap_{\mathcal F}=\langle\p_t,\p_y,(tx-y)\p_u,x\p_u,\p_u\rangle.\]
The nilradical $\mathfrak n^\cap_{\mathcal F}$ and the algebra $\hat{\mathfrak g}^\cap_{\mathcal F}$
are isomorphic to the nilpotent algebra~$L_5^1$~\cite{moro1958a} (or $\mathrm g_{5.1}$ in the notation of~\cite{muba1963c})
and the solvable algebra~$g_{6.54}$ with $\lambda=\gamma=1$~\cite{muba1963a}, respectively.
The nilradical~$\mathfrak n^\sim_{\mathcal F}$ is indecomposable and six-dimensional
with the four-dimensional maximal abelian ideal $\langle\p_u,x\p_u,\p_t,\p_y\rangle$
and the one-dimensional center~$\langle\p_u\rangle$.
The lower central series of~$\mathfrak n^\sim_{\mathcal F}$ is
$\mathfrak n^\sim_{\mathcal F}\supset
\langle\p_u,x\p_u,\p_y\rangle\supset
\langle\p_u\rangle$.
These structural properties of the algebra~$\mathfrak n^\sim_{\mathcal F}$ help
us to identify it.
It is isomorphic to the algebra B.14 with $\gamma=-1$
from Morozov's classification of six-dimensional nilpotent algebras~\cite{moro1958a}.

\section{Equivalence group and equivalence algebra of the gauged class}\label{sec:EquivGroupClassF'}

The elements of the group~$G^\sim_{\mathcal F}$ that preserve the subclass $\mathcal F_0$
constitute a proper subgroup~$H$ of~$G^\sim_{\mathcal F}$.
Let $\varpi$ denote the natural projection from the space with coordinates $(t,x,y,u,\alpha,\beta)$
onto the space with coordinates $(t,x,y,u,\beta)$.
The pushforward $\varpi_*H$ of the group~$H$ by~$\varpi$ is a subgroup of~$G^\sim_{\mathcal F'}$.
Nonetheless, since the class $\mathcal F$ is not normalized, we cannot guarantee that $\varpi_*H=G^\sim_{\mathcal F'}$.
Moreover, we prove that $\varpi_*H\subsetneq G^\sim_{\mathcal F'}$.
We cannot also use the embedding of the groupoid~$\mathcal G^\sim_{\mathcal F'}$
in the groupoid~$\mathcal G^\sim_{\mathcal F}$ induced by the embedding of $\mathcal F'$ in $\mathcal F$
since the description of~$\mathcal G^\sim_{\mathcal F}$ in Corollary~\ref{cor:EquivGroupoidsFF'Properties}(iii)
is based on that of~$\mathcal G^\sim_{\mathcal F'}$.
This is why we proceed with computing the equivalence group~$G^\sim_{\mathcal F'}$
similarly to as it is done in Section~\ref{sec:EquivGroupClassF}.
We apply the advanced version of the direct method to the class~$\mathcal F'$,
which is based on the embedding of the groupoid~$\mathcal G^\sim_{\mathcal F'}$
in the groupoid~$\mathcal G^\sim_{\bar{\mathcal F}}$;
see Section~\ref{sec:AlgMetClassF'} for the same computation using the algebraic method.

\begin{theorem}\label{thm:EquivGroupClassF'}
The usual equivalence group%
\footnote{
Formally, the set of transformations $G^\sim_{\mathcal F'}$ together with
the standard transformation composition form a {\it pseudogroup} rather than a group,
see~\cite[Section~A]{kova2024a} for details.
However, $G^\sim_{\mathcal F'}$ can be turned into a group in a simple way
following the approach from \cite{kova2023a,kova2023b,kova2024a}.
It is clear that the composition of two transformations of the form~\eqref{eq:EquivGroupOfF'A}
is also of this form and its natural domain is the entire space $\mathbb R^5:=\mathbb R^5_{t,x,y,u,\beta}$.
Similarly, composing~\eqref{eq:EquivGroupOfF'A} with~\eqref{eq:EquivGroupOfF'B}
results in a transformation of the form~\eqref{eq:EquivGroupOfF'B}
with the natural domain $\mathbb R^5\setminus M_0$,
where $M_0\subset\mathbb R^5$ is the hyperplane defined by the equation $x=0$.
However, the composition of two transformations of the form~\eqref{eq:EquivGroupOfF'B}
is of the form~\eqref{eq:EquivGroupOfF'A}
and its domain is the set $\mathbb R^5\setminus M_0$.
We can modify this operation as follows:
$\Phi_1\circ^{\rm mod}\Phi_2=\Phi_1\circ\Phi_2$ if $\Phi_1$ and $\Phi_2$ are both of the form~\eqref{eq:EquivGroupOfF'A}
or of the forms~\eqref{eq:EquivGroupOfF'A} and~\eqref{eq:EquivGroupOfF'B}.
If $\Phi_1$ and $\Phi_2$ are both of the form~\eqref{eq:EquivGroupOfF'B},
then $\Phi_1\circ^{\rm mod}\Phi_2$ is the extension by continuity of $\Phi_1\circ\Phi_2$ to the set $\mathbb R^5$
and is of the form~\eqref{eq:EquivGroupOfF'A}.
In this way, the set $G^\sim_{\mathcal F'}$ is a group with the modified composition $\circ^{\rm mod}$.
}
$G^\sim_{\mathcal F'}$ of the class~$\mathcal F'$ consists of the transformations of the following two forms:
\begin{gather}\label{eq:EquivGroupOfF'A}
\tilde t=t+c_0,\ \
\tilde x=\epsilon x,\ \
\tilde y=\epsilon y+c_1,\ \
\tilde u=c_2u+c_3(tx-y)+c_4x+c_5,\ \
\tilde\beta=\beta,
\\\label{eq:EquivGroupOfF'B}
\tilde t=\varepsilon' y+c_0',\ \
\tilde x=\frac{\epsilon'}x,\ \
\tilde y=\epsilon'\varepsilon' t+c_1',\ \
\tilde u=\frac{c_2'}xu+c_3'\left(t-\frac yx\right)+\frac{c_4'}x+c_5',\ \
\tilde\beta= 5-\beta,
\end{gather}
where $\epsilon,\epsilon'\in\{\pm1\}$, $\varepsilon'=\sgn x$, $c_0,\dots, c_5$ and $c_0',\dots,c_5'$
are arbitrary constants with $c_2c_2'\neq0$.
\end{theorem}

\begin{proof}
Items~(iii) and~(iv) of Theorem~\ref{thm:EquivGroupoidsFF'Properties}
essentially simplify the construction of the equivalence group~$G^\sim_{\mathcal F'}$ of the class~$\mathcal F'$.
Up to composing with the discrete equivalence subgroup generated by~$\mathscr J'$,
it suffices to find only the equivalence transformations with the identity $\beta$-components.
Since $\beta$ varies when considering equivalence transformations and $\tilde\beta=\beta$,
the equation~\eqref{eq:DetEqsForEquivGroupoidOfF'Conseq}, Lemma~\ref{lem:EqualityForArbitraryX}
and the inequality $T_tY_y-T_yY_t\ne0$
implies $T_y=Y_t=0$, $|Y_y|^{\beta-2}=|T_t|^{\beta-3}\sgn T_t$
and thus $T_t$ and $Y_y$ are positive and nonzero constants, respectively.
Then the system~\eqref{eq:DetEqsForEquivGroupoidOfF'} reduces to
\begin{gather*}
X = \frac{Y_y}{T_t}x,\quad
|Y_y|^{\beta-2}=T_t^{\beta-3},\quad
U^1_x=U^1_t=U^1_y=0,\quad
U^0_t+xU^0_y=|x|^\beta U^0_{xx}
\end{gather*}
with $T_t>0$.
Its second and fourth equations further split with respect to~$\beta$ to
$T_t=|Y_y|=1$, $U^0_{xx}=0$ and $U^0_t+xU^0_y=0$.
Therefore, $U^0=U^{01}(t,y)x+U^{00}(t,y)$, where $U^{00}_t=U^{01}_y=U^{01}_t+U^{00}_y=0$.
The solution set of the derived equations corresponds to the family of transformations~\eqref{eq:EquivGroupOfF'A}.
Its complement in~$G^\sim_{\mathcal F'}$ is obtained by composing its elements with~$\mathscr J'$,
which leads to the family of transformations~\eqref{eq:EquivGroupOfF'B}.
\end{proof}

\begin{remark}\label{rem:EssEquivTrans}
The transformations of the form~\eqref{eq:EquivGroupOfF'A}
constitute a normal subgroup~$\hat G^\cap_{\mathcal F^\prime}$ of $G^\sim_{\mathcal F^\prime}$,
which is isomorphic, via the natural projection to the space $\mathbb R^4_{t,x,y,u}$,
to the kernel point symmetry group~$G^\cap_{\mathcal F^\prime}$ of the class~$\mathcal F^\prime$,
cf.\ \cite{card2011a}.
The quotient of~$G^\sim_{\mathcal F^\prime}$ with respect to~\smash{$\hat G^\cap_{\mathcal F^\prime}$} is isomorphic to $\mathbb Z_2$
and can be naturally identified with the subgroup of~$G^\sim_{\mathcal F^\prime}$
constituted by the identity transformation of $(t,x,y,u,\beta)$ and the transformation~\eqref{eq:EssEquivTransOfF'}.
Hence the latter transformation can be considered as the only essential equivalence transformation
within the class~$\mathcal F^\prime$, which establishes a relation between different elements of~$\mathcal F^\prime$.
\end{remark}

\begin{remark}\label{rem:GaugingBeta}
Modulo~$G^\sim_{\mathcal F^\prime}$-equivalence, it suffices to consider a single value of~$\beta$
in each pair $(\beta_1,\beta_2)$ with $\beta_1+\beta_2=5$.
As the canonical representative of such a pair, we choose the value $\beta\in(-\infty,5/2]$.
\end{remark}

\begin{corollary}
The equivalence algebra $\mathfrak g^\sim_{\mathcal F'}$ of the class $\mathcal F'$ is spanned by the vector fields
$\p_t$, $\p_y$, $u\p_u$, $(tx-y)\p_u$, $x\p_u$ and $\p_u$
and is isomorphic to the kernel Lie invariance algebra of the class~$\mathcal F^\prime$
via the natural projection to the space $\mathbb R^4_{t,x,y,u}$.
\end{corollary}

\begin{corollary}\label{cor:DisceteEquivTransOfF'}
The usual equivalence group~$G^\sim_{\mathcal F'}$ of the class~$\mathcal F'$ is generated
by the one-parameter groups associated with the basis vector fields of
the equivalence algebra $\mathfrak g^\sim_{\mathcal F'}$ of this class
and its discrete equivalence transformations
\begin{gather*}
\mathscr I'_u\colon(t,x,y,u,\beta)\mapsto(t,x,y,-u,\beta),\\
\mathscr I'_{\rm s}\colon(t,x,y,u,\beta)\mapsto(t,-x,-y,u,\beta),\\
\mathscr J'\colon(t,x,y,u,\beta)\mapsto(y\sgn x,1/x,t\sgn x,u/x,5-\beta).
\end{gather*}
\end{corollary}

In fact, there are only two independent%
\footnote{%
Throughout the paper,
the independence of discrete elements in a transformation (pseudo)group is understood
up to composing with each other and with elements from the identity component of this group.
The independence of genuinely hidden Lie-symmetries vector fields is understood
in a similar sense,
up to linearly combining with each other and with induced Lie-symmetries vector fields
of the corresponding reduced system.
}
discrete equivalence transformations within the class~$\mathcal F'$,
e.g., $\mathscr I'_u$ and $\mathscr I'_{\rm s}\circ\mathscr J'$.

\begin{corollary}\label{cor:DiscrPointSymGroupOfF'}
The quotient group of the group~$G^\sim_{\mathcal F'}$ with respect to its identity component
is isomorphic to the degree-four dihedral group~$\mathrm D_4$.
As the canonical subgroup of discrete equivalence transformations,
one can take those generated by~$\mathscr I'_u$, $\mathscr I'_{\rm s}$ and $\mathscr J'$,
and a minimal set of generators is $\{\mathscr I'_u,\,\mathscr I'_{\rm s}\circ\mathscr J'\}$.
\end{corollary}

\begin{proof}
By $\pi$ we denote the natural projection of the group~$G^\sim_{\mathcal F'}$
onto the quotient group of this group by its identity component.
The quotient group $\pi(G^\sim_{\mathcal F'})$ is generated by the cosets
$\pi(\mathscr I'_u)$, $\pi(\mathscr I'_{\rm s})$ and $\pi(\mathscr J')$.
Using the power-commutator presentation
$\langle a_1,a_2,a_3\mid {a_1^2=a_2^2=a_3^2=e}$, $[a_1,a_2]=a_3,\,[a_1,a_3]=[a_2,a_3]=e\rangle$
of the group $\mathrm D_4$,
we define an isomorphism between $\pi(G^\sim_{\mathcal F'})$ and $\mathrm D_4$
by the generator correspondences
$\pi(\mathscr I'_u)\mapsto a_3$, $\pi(\mathscr I'_{\rm s})\mapsto a_1$ and $\pi(\mathscr J')\mapsto a_2$.
Then the canonical presentation $\langle a,b\mid a^4=b^2=e,\,bab^{-1}=a^{-1}\rangle$ of~$\mathrm D_4$,
where $a=a_1a_2$ and $b=a_3$ implies that
a minimal set of generators of~$\pi(G^\sim_{\mathcal F'})$
is constituted by $\pi(\mathscr I'_{\rm s}\circ\mathscr J')$ and $\pi(\mathscr I'_u)$.
\end{proof}

It is clear that the kernel invariance algebras~$\mathfrak g^\cap_{\mathcal F}$ and~$\mathfrak g^\cap_{\mathcal F'}$
of equations from the classes~$\mathcal F$ and~$\mathcal F'$, respectively,
coincide, see Theorem~\ref{thm:GroupClassificationOfF'} and Corollary~\ref{cor:GroupClassificationOfF} below.
At the same time, the corresponding kernel point symmetry groups~$G^\cap_{\mathcal F}$ and~$G^\cap_{\mathcal F'}$
are different, $G^\cap_{\mathcal F}\varsubsetneq G^\cap_{\mathcal F'}$.
More specifically, the group~$G^\cap_{\mathcal F'}$ is generated by elements of~$G^\cap_{\mathcal F}$
and the transformation $(t,x,y,u)\mapsto(t,-x,-y,u)$.

Since $\mathfrak g^\sim_{\mathcal F'}\subsetneq\varpi_*\mathfrak g^\sim_{\mathcal F}$
and $G^\sim_{\mathcal F'}\not\subset\varpi_*G^\sim_{\mathcal F}$,
we can interpret $G^\sim_{\mathcal F'}$ as a \emph{nontrivial conditional equivalence group}
\cite[Section~2.6]{popo2010a}
of~$\mathcal F$ associated with the constraint $\alpha=0$,
but the algebra $\mathfrak g^\sim_{\mathcal F'}$ is trivial
when interpreted as a conditional equivalence algebra of~$\mathcal F$.

\section{Classification of Lie symmetries }\label{sec:GroupClassificationClassF'}

In view of Corollary~\ref{cor:EquivGroupoidsFF'Properties},
the $\mathcal G^\sim_{\mathcal F}$- and $\mathcal G^\sim_{\mathcal F'}$-equivalences
within the classes $\mathcal F$ and $\mathcal F'$, respectively, are consistent.
More specifically, the equations from the class $\mathcal F'$ are $\mathcal G^\sim_{\mathcal F'}$-equivalent
if and only if their counterparts in $\mathcal F$
under the embedding $\mathcal F'\hookrightarrow\mathcal F$ are $\mathcal G^\sim_{\mathcal F}$-equivalent.
The class~$\mathcal F'$ is semi-normalized in the usual sense,
i.e., the $\mathcal G^\sim_{\mathcal F'}$-equivalence coincide with the $G^\sim_{\mathcal F'}$-equivalence.
This is why it is convenient to first carry out the group classification of the class~$\mathcal F'$ and
then apply Corollary~\ref{cor:EquivGroupoidsFF'Properties}, Theorem~\ref{thm:EqGroupFPsubclass}
and Remark~\ref{rem:EssEquivTrans} to the obtained results for solving the group classification problem for the class~$\mathcal F'$.

According to the classical Lie algorithm of computing continuous point symmetries~\cite{blum2010A,blum1989A,olve1993A},
a Lie symmetry vector field~$Q$ of the equation~$\mathcal F'_\beta$ takes the general form
\[
Q=\tau(t,x,y,u)\p_t+\xi^x(t,x,y,u)\p_x+\xi^y(t,x,y,u)\p_y+\eta(t,x,y,u)\p_u
\]
with components that can be found from the infinitesimal invariance criterion
\begin{gather*}
Q_{(2)}\big(u_t+xu_y-|x|^\beta u_{xx}\big)\big|_{\mathcal F'_{\beta}}
=\big(\eta^t+x\eta^y+\xi^xu_y-|x|^\beta\eta^{xx}-\beta\xi^xx^{-1}|x|^\beta u_{xx}\big)\big|_{\mathcal F'_{\beta}}
=0.
\end{gather*}
Here $Q_{(2)}$ is the second-order prolongation of the vector field $Q$,
$\eta^t=\mathrm D_tQ[u]+\tau u_{tt}+\xi^xu_{tx}+\xi^yu_{ty}$,
$\eta^y=\mathrm D_yQ[u]+\tau u_{ty}+\xi^xu_{xy}+\xi^yu_{yy}$,
$\eta^{xx}=\mathrm D_x^2Q[u]+\tau u_{txx}+\xi^xu_{xxx}+\xi^yu_{xxy}$,
the characteristic $Q[u]$ of the vector field~$Q$ is defined as
$Q[u]:=\eta-\tau u_t-\xi^xu_x-\xi^yu_y$,
$\mathrm D_t$, $\mathrm D_x$, $\mathrm D_y$ are the total derivative operators with respect to $t$, $x$ and~$y$, respectively,
$\mathrm D_t=\p_t+u_t\p_u+u_{tt}\p_{u_t}+u_{tx}\p_{u_x}+u_{ty}\p_{u_y}+\cdots$
and similarly for $\mathrm D_x$ and $\mathrm D_y$.
The linear space of such vector fields with Lie bracket of vector fields
is the maximal Lie invariance algebra~$\mathfrak g_{\beta}$ of the equation~$\mathcal F'_{\beta}$.

The form of admissible transformations~\eqref{eq:ClassFbarTransPart}
implies the principle part of the determining equations for the components of the vector field~$Q$,
\begin{subequations}\label{eq:LieSymDetEquations}
\begin{gather}\label{eq:LieSymDetEqsTransformPart1}
\tau_u=\tau_x=0,\quad
\xi^x_u=0,\quad
\xi^y_u=\xi^y_x=0,\quad
\eta_{uu} = 0.
\end{gather}

The other determining equations, including the classifying ones, follow from the infinitesimal invariance criterion jointly with~\eqref{eq:LieSymDetEqsTransformPart1},
\begin{gather}\label{eq:LieSymDetEqsTransformPart2}
\begin{split}
&\eta_t+x\eta_y-|x|^\beta\eta_{xx}= 0,\\
&-2|x|^\beta\eta_{ux}=\xi^x_t+x\xi^x_y-|x|^\beta\xi^x_{xx},\\
&\xi^y_t+x\xi^y_y  =2x\xi^x_x+(1-\beta)\xi^x,\\
&\xi^y_t+x\xi^y_y  =\tau_tx+\tau_yx^2+\xi^x.
\end{split}
\end{gather}
\end{subequations}

Unlike the group classification in~\cite{zhan2020a}, we carry out the group classification of the class~$\mathcal F'$
up to the $G^\sim_{\mathcal F'}$-equivalence, which coincides with the $\mathcal G^\sim_{\mathcal F'}$-equivalence
in view of the semi-normalization of this class.
In other words, without loss of generality we can impose the gauge $\beta\leqslant5/2$,
which in its turn reduces the number of cases we need to consider.

\begin{theorem}\label{thm:GroupClassificationOfF'}
The kernel Lie invariance algebra~\smash{$\mathfrak g^\cap_{\mathcal F'}$} of the equations from the class~$\mathcal F'$ is
\[
\mathfrak g^\cap_{\mathcal F'}=\langle\mathcal P^t,\mathcal P^y,\mathcal I,(tx-y)\p_u,x\p_u,\p_u\rangle, \quad\mbox{where}\quad
\mathcal P^t:=\p_t,\quad\mathcal P^y:=\p_y,\quad\mathcal I:=u\p_u.
\]
Any equation $\mathcal F'_\beta$ from~$\mathcal F'$ is invariant with respect to the algebra
\[
\mathfrak g^{\rm gen}_\beta=\langle\mathcal P^t,\mathcal P^y,\mathcal I,\mathcal D^\beta,\mathcal Z(f^\beta)\rangle
\ \ \mbox{with}\ \
\mathcal D^\beta:=(2-\beta)t\p_t+x\p_x+(3-\beta)y\p_y,\
\mathcal Z(f^\beta):=f^\beta\p_u,
\]
where the parameter function~$f^\beta=f^\beta(t,x,y)$ runs through the solution set of this equation,
and $\beta\in(-\infty,5/2]$ modulo the $G^\sim_{\mathcal F^\prime}$-equivalence.
the maximal Lie invariance algebra~$\mathfrak g_\beta$ of the equation~$\mathcal F_\beta$
coincides with $\mathfrak g^{\rm gen}_\beta$ if and only if $\beta\in\mathbb R\setminus\{0,2,3,5\}$.
A~complete list of $G^\sim_{\mathcal F'}$-inequivalent essential Lie symmetry extensions
in the class~$\mathcal F'$ is exhausted by the following cases:
\begin{gather*}
\beta=2\colon\ \ \mathfrak g_2=\mathfrak g^{\rm gen}_2\dotplus\langle\mathcal K_2\rangle\quad\mbox{with}\quad
\mathcal K_2 = 2xy\p_x+y^2\p_y-xu\p_u,
\\[1ex]
\beta=0\colon\ \ \mathfrak g_0=\mathfrak g^{\rm gen}_0\dotplus\langle\mathcal K_0,\mathcal P^3,\mathcal P^2,\mathcal P^1\rangle
\quad\mbox{with}\quad\\\phantom{\beta=0\colon\ \ }
\mathcal K_0 =t^2\p_t+(tx+3y)\p_x+3ty\p_y-(x^2\!+2t)u\p_u,\\\phantom{\beta=0\colon\ \ }
\mathcal P^3 =3t^2\p_x+t^3\p_y+3(y-tx)u\p_u,\ \
\mathcal P^2 =2t\p_x+t^2\p_y-xu\p_u,\ \
\mathcal P^1 =\p_x+t\p_y.
\end{gather*}
\end{theorem}

\begin{proof}
Splitting the system of determining equations~\eqref{eq:LieSymDetEquations} with respect to the arbitrary element~$\beta$,
we find the kernel Lie invariance algebra $\mathfrak g^\cap_{\mathcal F'}$ of the equations from the class~$\mathcal F'$.

To find the algebra~$\mathfrak g_\beta$,
we start with solving the system of the equations~\eqref{eq:LieSymDetEqsTransformPart1} together with the fourth equation from~\eqref{eq:LieSymDetEqsTransformPart2},
obtaining
\begin{gather*}
\tau=\tau(t,y),\quad
\xi^x=-x^2\tau_y+x\xi^y_y-x\tau_t+\xi^y_t,\quad
\xi^y=\xi^y(t,y),\quad
\eta =\eta^1(t,x,y)u+\eta^0(t,x,y),
\end{gather*}
where $\tau$, $\xi^y$, $\eta^0$ and~$\eta^1$ are smooth functions of their arguments.
Substituting this solution into the third equation from~\eqref{eq:LieSymDetEqsTransformPart2}
and splitting the result with respect to $x$ lead to the system
\begin{gather*}
(\beta-5)\tau_y=0,\quad
(\beta-3)\tau_t-(\beta-2)\xi^y_y=0,\quad
\beta\xi^y_t=0,
\end{gather*}
which is of maximal rank as a system of differential equations
with the independent variables $(t,y)$ and the dependent variables $(\tau,\xi^y)$
if and only if $\beta\in\mathbb R\setminus\{0,2,3,5\}$.
In this case, the general solution of the system is given by
\begin{gather*}
\tau=c_1(2-\beta)t+c_2,\quad
\xi^y=c_1(3-\beta)y+c_3.
\end{gather*}
Substituting this solution into the rest of the system~\eqref{eq:LieSymDetEqsTransformPart2}
gives the algebra $\mathfrak g_\beta$ for $\beta\in\mathbb R\setminus\{0,2,3,5\}$,
or, equivalently, for $\beta\in(-\infty,5/2]\setminus\{0,2\}$ under the gauge $\beta\leqslant5/2$.

The equivalence transformation~\eqref{eq:EssEquivTransOfF'} reduces the cases $\beta=5$ and $\beta=3$
to the cases $\beta=0$ and $\beta=2$, respectively.
Solving the determining equations for the latter two cases results in the algebras $\mathfrak g_0$ and $\mathfrak g_2$.
\end{proof}

For any $\beta$, the vector fields $\mathcal{Z}(f^\beta)$ constitute
the infinite-dimensional abelian ideal $\mathfrak g^{\rm lin}_\beta$ of~$\mathfrak g_\beta$
that is associated with the linear superposition of solutions of $\mathcal F'_\beta$,
\smash{$\mathfrak g^{\rm lin}_\beta:=\{\mathcal{Z}(f^\beta)\}$}.
The algebra~$\mathfrak g_\beta$ splits over this ideal,
\smash{$\mathfrak g_\beta=\mathfrak g^{\rm ess}_\beta\lsemioplus\mathfrak g^{\rm lin}_\beta$},
where $\mathfrak g^{\rm ess}_\beta$ is a (finite-dimensional) subalgebra of $\mathfrak g_\beta$,
which is called the essential Lie invariance algebra of $\mathcal F'_\beta$.
In the above notation,
\begin{gather*}
\mathfrak g^{\rm ess}_\beta:=\langle \mathcal P^t,\,\mathcal P^y,\,\mathcal I,\,\mathcal D^\beta\rangle,\quad \beta\in\mathbb R\setminus\{0,2,3,5\},
\\
\mathfrak g^{\rm ess}_2:=\langle\mathcal P^y,\mathcal D_2,\mathcal K,\mathcal P^t,\mathcal I\rangle,\\
\mathfrak g^{\rm ess}_0:=\langle\mathcal P^t,\mathcal D_0,\mathcal K,\mathcal P^3,\mathcal P^2,\mathcal P^1,\mathcal P^0,\mathcal I\rangle.
\end{gather*}

Similar splitting occurs for the corresponding Lie symmetry pseudogroups $G_\beta$ for each~$\beta$.
The set of point symmetry transformations
$G^{\rm lin}_\beta=\{\mathscr Z(f)\colon\tilde t=t,\tilde x=x,\tilde y=y, \tilde u=f^\beta(t,x,y)\}$,
where \smash{$f^\beta$} is an arbitrary solution of the equation~\smash{$\mathcal F'_\beta$},
constitute a normal pseudosubgroup in~$G_\beta$.
Moreover, the pseudogroup~$G_\beta$ splits over~$G^{\rm lin}_\beta$,
$G_\beta=G^{\rm ess}_\beta\ltimes G^{\rm lin}_\beta$.
Here $G^{\rm ess}_\beta$ is a pseudosubgroup in~$G_\beta$ called essential Lie invariance pseudogroup.

Theorems~\ref{thm:EquivGroupoidsFF'Properties} and~\ref{thm:GroupClassificationOfF'}
and Corollary~\ref{cor:EquivGroupoidsFF'Properties} jointly imply
the solution of the group classification problem for the class~$\mathcal F$.

\begin{corollary}\label{cor:GroupClassificationOfF}
The kernel Lie invariance algebra~$\mathfrak g^\cap_{\mathcal F}$ of the equations from the class~$\mathcal F$
coincides with that for the class~$\mathcal F'$, $\mathfrak g^\cap_{\mathcal F}=\mathfrak g^\cap_{\mathcal F'}$.
Any equation $\mathcal F_{\alpha\beta}$ from~$\mathcal F$ is invariant with respect to the algebra
\[
\mathfrak g^{\rm gen}_{\alpha\beta}=\big\langle\mathcal P^t,\mathcal P^y,\mathcal I,\mathcal D^{\alpha\beta},\mathcal Z(f^{\alpha\beta})\big\rangle
\]
with $\mathcal D^{\alpha\beta}:=(2-\beta)t\p_t+(x-\alpha)\p_x+\big((3-\beta)y-\alpha t\big)\p_y$,
$\mathcal Z(f^{\alpha\beta}):=f^{\alpha\beta}\p_u$,
and the parameter function~$f^{\alpha\beta}=f^{\alpha\beta}(t,x,y)$ running through the solution set of this equation.
Modulo the \smash{$\mathcal G^\sim_{\mathcal F}$}-equivalence, we can assume $\beta\in(-\infty,5/2]$,
and a~complete list of \smash{$\mathcal G^\sim_{\mathcal F}$}-inequivalent essential Lie symmetry extensions
in the class $\mathcal F$ is exhausted by the counterparts of those in the class~$\mathcal F'$,
$\mathcal F_{00}$ and~$\mathcal F_{02}$.
An analogous list up to the \smash{$G^\sim_{\mathcal F}$}-equivalence
consists of the equations $\mathcal F_{00}$, $\mathcal F_{02}$, $\mathcal F_{03}$ and~$\mathcal F_{05}$.
\end{corollary}

The maximal Lie invariance algebras of the equations~$\mathcal F_{03}\sim\mathcal F'_3$ and~$\mathcal F_{05}\sim\mathcal F'_5$
are presented in Propositions~\ref{pro:F3MIA} and~\ref{pro:F5MIA}, respectively.

\section{Remarkable Fokker--Planck equation and its counterpart}\label{sec:RemarkableFPandCounterpart}

The equation~$\mathcal F'_0$ is of the simplest form
among the equations from the classes~$\mathcal F'$, $\mathcal F$ and even~$\bar{\mathcal F}$,
\begin{gather}\label{eq:RemarkableFP}
\mathcal F'_0\colon\quad
u_t+xu_y=u_{xx},
\end{gather}
and its essential Lie invariance algebra is of the greatest dimension, which is equal to eight.
Moreover, up to the equivalence with respect to point transformations,
this is a unique equation with the above property within each of the above classes.
This is why in~\cite{kova2023a} we called it the remarkable (1+2)-dimensional Fokker--Planck equation.
We briefly review the results of \cite{kova2023a,popo2024b} on this equation
and then extend them to the equation~\eqref{eq:FPWithRationalBeta} with $\beta=5$
using the point transformation~\eqref{eq:EssEquivTransMod}.

The maximal Lie invariance algebra of the equation~\eqref{eq:RemarkableFP} is
(see, e.g., \cite{kova2013a})
\begin{gather*}\label{eq:RemarkableFPMIA}
\mathfrak g_0:=\langle \mathcal P^t,\,\mathcal D,\,\mathcal K,\,
\mathcal P^3,\,\mathcal P^2,\,\mathcal P^1,\,\mathcal P^0,\,\mathcal I,\mathcal Z(f)\rangle,
\end{gather*}
where
\begin{gather*}
\mathcal P^t =\p_t,\ \
\mathcal D   =2t\p_t+x\p_x+3y\p_y-2u\p_u,\ \
\mathcal K   =t^2\p_t+(tx+3y)\p_x+3ty\p_y-(x^2\!+2t)u\p_u,\\[.5ex]
\mathcal P^3 =3t^2\p_x+t^3\p_y+3(y-tx)u\p_u,\ \
\mathcal P^2 =2t\p_x+t^2\p_y-xu\p_u,\ \
\mathcal P^1 =\p_x+t\p_y,\ \
\mathcal P^0 =\p_y,\\[.5ex]
\mathcal I   =u\p_u,\quad
\mathcal Z(f)=f(t,x,y)\p_u.
\end{gather*}
Here the parameter function~$f$ runs through the solution set of the equation~\eqref{eq:RemarkableFP}.
The essential Lie invariance algebra~$\mathfrak g_0^{\rm ess}$ splits over its radical
$\mathfrak r_0=\langle\mathcal P^3,\mathcal P^2,\mathcal P^1,\mathcal P^0,\mathcal I\rangle$,
which is isomorphic to the real rank-two Heisenberg algebra ${\rm h}(2,\mathbb R)$,
$\mathfrak g_0^{\rm ess}=\mathfrak f_0\lsemioplus\mathfrak r_0$,
where $\mathfrak f_0=\langle\mathcal P^t,\mathcal D,\mathcal K\rangle$ is a Levi factor of~$\mathfrak g_0^{\rm ess}$
and it is isomorphic to the real order-two special linear Lie algebra ${\rm sl}(2,\mathbb R)$.

\begin{theorem}[\cite{kova2023a}]\label{thm:RemarkableFPSymGroup}
The complete point symmetry pseudogroup~$G_0$ of the remarkable Fokker--Planck equation~\eqref{eq:RemarkableFP}
consists of the transformations of the form
\begin{gather}\label{eq:RemarkableFPSymGroup}
\begin{split}
&\tilde t=\frac{\alpha t+\beta}{\gamma t+\delta},
\quad
\tilde x=\frac{\hat x}{\gamma t+\delta}
-\frac{3\gamma\hat y}{(\gamma t+\delta)^2},
\quad
\tilde y=\frac{\hat y}{(\gamma t+\delta)^3},
\\[1ex]
&\tilde u=\sigma(\gamma t+\delta)^2\exp\left(
\frac{\gamma\hat x^2}{\gamma t+\delta}
-\frac{3\gamma^2\hat x\hat y}{(\gamma t+\delta)^2}
+\frac{3\gamma^3\hat y^2}{(\gamma t+\delta)^3}
\right)
\\
&\hphantom{\tilde u={}}
\times\exp\big(
3\lambda_3(y-tx)-\lambda_2x-(3\lambda_3^2t^3+3\lambda_3\lambda_2t^2+\lambda_2^2t)
\big)
\big(u+f(t,x,y)\big),
\end{split}
\end{gather}
where
$\hat x:=x+3\lambda_3t^2+2\lambda_2t+\lambda_1$,
$\hat y:=y+\lambda_3t^3+\lambda_2t^2+\lambda_1t+\lambda_0$;
$\alpha$, $\beta$, $\gamma$ and $\delta$ are arbitrary constants with $\alpha\delta-\beta\gamma=1$;
$\lambda_0$,~\dots, $\lambda_3$ and $\sigma$ are arbitrary constants with $\sigma\ne0$,
and $f$ is an arbitrary solution of~\eqref{eq:RemarkableFP}.	
\end{theorem}

Pulling back an arbitrary solution $u=h(t,x,y)$ of~\eqref{eq:RemarkableFP}
by an arbitrary point symmetry transformation of the form~\eqref{eq:RemarkableFPSymGroup},
we obtain, in the notation of Theorem~\ref{thm:RemarkableFPSymGroup},
the formula of generating new solutions of~\eqref{eq:RemarkableFP}
from known ones under the action of elements of~$G_0$,
\begin{gather}\label{eq:RemarkableFPNewSolutionsByG0}
\begin{split}
u={}&
\frac{{\rm e}^{\lambda_2x-3\lambda_3(y-tx)+3\lambda_3^2t^3+3\lambda_3\lambda_2t^2+\lambda_2^2t}}{\sigma(\gamma t+\delta)^2}
\exp\left(
-\frac{\gamma\hat x^2}{\gamma t+\delta}
+\frac{3\gamma^2\hat x\hat y}{(\gamma t+\delta)^2}
-\frac{3\gamma^3\hat y^2}{(\gamma t+\delta)^3}
\right)
\\
&
\times
h\left(
\frac{\alpha t+\beta}{\gamma t+\delta},\,
\frac{\hat x}{\gamma t+\delta}-\frac{3\gamma\hat y}{(\gamma t+\delta)^2},\,
\frac{\hat y}{(\gamma t+\delta)^3}
\right)
-f(t,x,y).
\end{split}
\end{gather}

The following families of $G_0$-inequivalent solutions of the equation~\eqref{eq:RemarkableFP}
that arise from its codimension-one Lie reductions were constructed in~\cite{kova2023a}:
\begin{gather}\label{eq:F0HeatSolution0}
\solution\beta=0\colon\quad
u=|x|^{-\frac14}\vartheta^\mu\Big(\tfrac94\tilde\varepsilon y,|x|^{\frac32}\Big)
\quad\mbox{with}\quad \mu=\tfrac5{36},\quad \tilde\varepsilon:=\sgn x,
\\[.5ex]\label{eq:F0HeatSolution1}
\solution\beta=0\colon\quad
u=|t|^{-\frac12}{\rm e}^{-\frac{x^2}{4t}}\vartheta^0
\Big(\tfrac13{t^3}+2\varepsilon t-t^{-1},2y-(t+\varepsilon t^{-1})x\Big)
\quad\mbox{with}\quad \varepsilon\in\{-1,1\},
\\[.5ex]\label{eq:F0HeatSolution2}
\solution\beta=0\colon\quad
u=\vartheta^0\Big(\tfrac13t^3,y-tx\Big),
\\[.5ex]\label{eq:F0HeatSolution3}
\solution\beta=0\colon\quad
u=\vartheta^0(t,x),
\end{gather}
where $\vartheta^\mu=\vartheta^\mu(z_1,z_2)$
is an arbitrary solution of the equation $\vartheta^\mu_1=\vartheta^\mu_{22}+\mu z_2^{-2}\vartheta^\mu$.%
\footnote{\label{fnt:SolutionsOfHeatEqs}
A~complete collection of inequivalent Lie invariant solutions of equations of such form with generic $\mu\ne0$
is presented in \cite[Section~A]{kova2023a}, see also~\cite{gung2018a,gung2018b}.
For $\mu=0$, this equation is just the (1+1)-dimensional linear heat equation.
An enhanced complete collection of inequivalent Lie invariant solutions of this equation
was presented in~\cite[Section A]{vane2021a}, following Examples 3.3 and 3.17 in~\cite{olve1993A}.
}

The algebra of generalized symmetries of the equation~\eqref{eq:RemarkableFP} was computed in~\cite{popo2024b}.
The properties of the associative algebra of differential recursion operators of~\eqref{eq:RemarkableFP}
played the key role in that computation.
Moreover, these operators were applied therein to generating new solutions of $\mathcal F_0$ from known ones.
In particular, using the solution families~\eqref{eq:F0HeatSolution0}--\eqref{eq:F0HeatSolution3}
of the Lie-invariant solutions as seed ones,
large families of essentially new solutions of~\eqref{eq:RemarkableFP} were generated.

More specifically, consider the operators in total derivatives
that are associated with the Lie-symmetry vector fields~$-\mathcal P^3$,
$-\mathcal P^2$, $-\mathcal P^1$, $-\mathcal P^0$
and~$-\mathcal P^t$, $-\mathcal D$, $-\mathcal K$ of~\eqref{eq:RemarkableFP},
\begin{gather*}
\mathrm P^3:=3t^2\mathrm D_x+t^3\mathrm D_y-3(y-tx),\ \
\mathrm P^2:=2t\mathrm D_x+t^2\mathrm D_y+x,\ \
\mathrm P^1:=\mathrm D_x+t\mathrm D_y,\ \
\mathrm P^0:=\mathrm D_y,\\
\mathrm P^t:=\mathrm D_t,\ \
\mathrm D  :=2t\mathrm D_t+x\mathrm D_x+3y\mathrm D_y+2,\ \
\mathrm K  :=t^2\mathrm D_t+(tx+3y)\mathrm D_x+3ty\mathrm D_y+x^2+2t.
\end{gather*}
These operators are recursion operators of the equation~\eqref{eq:RemarkableFP} \cite[Proposition~5.22]{olve1993A},
which we call Lie-symmetry operators.
By~$\Upsilon_{\mathfrak r_0}$ we denote the associative algebra
generated by the operators~$\mathrm P^3$, $\mathrm P^2$, $\mathrm P^1$ and~$\mathrm P^0$.
On solutions of the equation~\eqref{eq:RemarkableFP},
the operators~$\mathrm P^t$, $\mathrm D$ and~$\mathrm K$
are equivalent to the elements
\begin{gather*}
\hat{\mathrm P}^t:=(\mathrm P^1)^2-\mathrm P^2\mathrm P^0=\mathrm D_x^2-x\mathrm D_y,
\\
\hat{\mathrm D}  :=\mathrm P^2\mathrm P^1-\mathrm P^3\mathrm P^0+2=2t\mathrm D_x^2+x\mathrm D_x+(3y-2tx)\mathrm D_y+2,
\\
\hat{\mathrm K}  :=(\mathrm P^2)^2-\mathrm P^3\mathrm P^1=t^2\mathrm D_x^2+(3y+tx)\mathrm D_x+t(3y-tx)\mathrm D_y+x^2+2t
\end{gather*}
of the associative algebra~$\Upsilon_{\mathfrak r_0}$, respectively,
and thus they are inessential in the course of generating solutions.
The only nontrivial relations between the generators of the algebra~$\Upsilon_{\mathfrak r_0}$
are given by the commutator relations $[\mathrm P^1,\mathrm P^2]=1$, $[\mathrm P^0,\mathrm P^3]=-3$.
Using Bergman's diamond lemma~\cite{berg1978a}, which extends the Gr\"obner bases technique to noncommutative rings,
we explicitly construct a basis of the algebra~$\Upsilon_{\mathfrak r_0}$ as follows.
For any (ordered) basis $(\mathrm Q^1,\mathrm Q^2,\mathrm Q^3,\mathrm Q^4)$ of
the span $\langle\mathrm P^3,\mathrm P^2,\mathrm P^1,\mathrm P^0\rangle$,
the monomials of the form $(\mathrm Q^1)^{i_1}(\mathrm Q^2)^{i_2}(\mathrm Q^3)^{i_3}(\mathrm Q^4)^{i_4}$
with $(i_1,i_2,i_3,i_4)\in\mathbb N_0^4$ constitute a basis of the vector space~$\Upsilon_{\mathfrak r_0}$.

\begin{theorem}[{\cite{popo2024b}}]
The algebra of canonical representatives of generalized symmetries
of the remarkable Fokker--Planck equation~\eqref{eq:RemarkableFP} is
$\Sigma_0=\Lambda_0\lsemioplus\Sigma^{-\infty}_0$,
where
\begin{gather*}
\Lambda_0=\big\langle\big((\mathrm P^3)^{i_3}(\mathrm P^2)^{i_2}(\mathrm P^1)^{i_1}(\mathrm P^0)^{i_0}u\big)\p_u
\mid i_0,i_1,i_2,i_3\in\mathbb N_0\big\rangle,\quad
\Sigma^{-\infty}:=\big\{\mathcal Z(f)\big\}.
\end{gather*}
Here the parameter function~$f$ runs through the solution set of~\eqref{eq:RemarkableFP}.
\end{theorem}

Acting on an arbitrary solution $u=h(t,x,y)$ of the equation~\eqref{eq:RemarkableFP} by an element
\begin{gather*}
Q=\sum_{(i_1,i_2,i_3,i_4)\in\mathbb N_0^4}c_{i_1i_2i_3i_4}(\mathrm Q^1)^{i_1}(\mathrm Q^2)^{i_2}(\mathrm Q^3)^{i_3}(\mathrm Q^4)^{i_4}
\end{gather*}
of the algebra~$\Upsilon_{\mathfrak r_0}$,
where only finitely many real constants $c_{i_1i_2i_3i_4}$ are nonzero,
we obtain a solution $Qh$ of~\eqref{eq:RemarkableFP}.
In fact, applying this procedure to a known solution, one may obtain a solution that is known as well.
It is shown in~\cite[Section~5.3]{popo2024b} that for the solution
families~\eqref{eq:F0HeatSolution0},~\eqref{eq:F0HeatSolution1}, \eqref{eq:F0HeatSolution2} and~\eqref{eq:F0HeatSolution3},
only the action by the monomials $(\mathrm P^1)^k$, $(\mathrm P^1)^k$, $(\mathrm P^2)^k$ and $(\mathrm P^3)^k$, $k\in\mathbb N$,
respectively, in general leads to essentially new solutions of~\eqref{eq:RemarkableFP},
\begin{gather*}
\solution\beta=0\colon\quad
u=(\mathrm P^1)^k\Big(|x|^{-\frac14}\vartheta^\mu\big(\tfrac94\tilde\varepsilon y,|x|^{\frac32}\big)\Big)
\quad\mbox{with}\quad \mu=\tfrac5{36},\quad \tilde\varepsilon:=\sgn x,
\\
\solution\beta=0\colon\quad
u=(\mathrm P^1)^k\Big(|t|^{-\frac12}{\rm e}^{-\frac{x^2}{4t}}
\vartheta^0\big(\tfrac13{t^3}+2\varepsilon t-t^{-1},2y-(t+\varepsilon t^{-1})x\big)
\Big)
\quad\mbox{with}\quad\varepsilon\in\{-1,1\},
\\
\solution\beta=0\colon\quad
u=(\mathrm P^2)^k\vartheta^0\big(\tfrac13t^3,y-tx\big),\vphantom{\Big(}
\\
\solution\beta=0\colon\quad
u=(\mathrm P^3)^k\vartheta^0(t,x).\vphantom{\Big(}
\end{gather*}
Note that for the above cases of solution generation,
the generated solutions with $k=0,\dots,n-1$, $n\in\mathbb N$, span
the spaces of solutions of the equation~\eqref{eq:RemarkableFP}
that are respectively invariant with respect to its generalized symmetries
$\big((\mathrm P^3)^nu\big)\p_u$, ${\big((\mathrm P^2+\varepsilon\mathrm P^0)^nu\big)\p_u}$, 
$\big((\mathrm P^1)^nu\big)\p_u$ and $\big((\mathrm P^0)^nu\big)\p_u$, 
see \cite[Section~8]{kova2023a} and \cite[Section~5.3]{popo2024b}.

Pushing forward the above results for the equation~\eqref{eq:RemarkableFP} by the transformation~\eqref{eq:EssEquivTransOfF'},
we derive the analogous results for the equation~$\mathcal F'_5$.
Following Remark~\ref{rem:SimplerEquivTrans},
instead of the latter equation, we consider its counterpart without the absolute value sign,
\begin{gather}\label{eq:Power5FP}
u_t+xu_y=x^5u_{xx},
\end{gather}
and replace the transformation~\eqref{eq:EssEquivTransOfF'} by the simpler transformation~\eqref{eq:EssEquivTransMod}.

\begin{proposition}\label{pro:F5MIA}
The maximal Lie invariance algebra $\mathfrak g_5$ of the equation~\eqref{eq:Power5FP} is spanned by the vector fields
\begin{gather*}
\mathcal P^y=\p_y,\quad
\mathcal D  =3t\p_t-x\p_x+2y\p_y-3u\p_u,
\\
\mathcal K  =3ty\p_t-x(3tx+y)\p_x+y^2\p_y-(3tx+3y+x^{-2})u\p_u,
\\
\mathcal P^3=y^3\p_t-3x^2y^2\p_x-3(xy^2-t+x^{-1}y)u\p_u,
\quad
\mathcal P^2=y^2\p_t-2x^2y\p_x-(2yx+x^{-1})u\p_u,
\\
\mathcal P^1=y\p_t-x^2\p_x-xu\p_u,
\quad
\mathcal P^0=\p_t,\quad
\mathcal I  =u\p_u,\quad
\mathcal Z(f)=f(t,x,y)\p_u,
\end{gather*}
where the parameter function $f$ runs through the solution set of the equation~\eqref{eq:Power5FP}.
\end{proposition}

The algebra $\mathfrak g_5^{\rm ess}$ admits a Levi decomposition $\mathfrak g_5^{\rm ess}=\mathfrak f_5\lsemioplus\mathfrak r_5$,
where~$\mathfrak f_5$ and~$\mathfrak r_5$ are the images of~$\mathfrak f_0$ and~$\mathfrak r_0$
under the pushforward by~\eqref{eq:EssEquivTransOfF'}.

\begin{theorem}\label{thm:Power5FPSymGroup}
The complete point symmetry pseudogroup~$G_5$ of the equation~\eqref{eq:Power5FP}
consists of the transformations of the form
\begin{gather*}
\begin{split}
&\tilde t=\frac{\hat t}{(\gamma y+\delta)^3},\quad
\tilde x=\frac{x(\gamma y+\delta)^2}{\hat x(\gamma y+\delta)-3\gamma\hat tx},\quad
\tilde y=\frac{\alpha y+\beta}{\gamma y+\delta},
\\[1ex]
&\tilde u=\frac{\sigma(\gamma y+\delta)^4}{\hat x(\gamma y+\delta)-3\gamma\hat tx}
\exp\left(
\frac{\gamma\hat x^2}{x^2(\gamma y+\delta)}
-\frac{3\gamma^2\hat x\hat t}{x(\gamma y+\delta)^2}
+\frac{3\gamma^3\hat t^2}{(\gamma y+\delta)^3}
\right)
\\
&\hphantom{\tilde u={}}
\times\exp\big(
3\lambda_3(t-y/x)-\lambda_2/x-(3\lambda_3^2y^3+3\lambda_3\lambda_2y^2+\lambda_2^2y)
\big)
\big(u+f(t,x,y)\big),
\end{split}
\end{gather*}
where
$\hat x:=1+3\lambda_3xy^2+2\lambda_2xy+\lambda_1x$,
$\hat t:=t+\lambda_3y^3+\lambda_2y^2+\lambda_1y+\lambda_0$;
$\alpha$, $\beta$, $\gamma$ and $\delta$ are arbitrary constants with $\alpha\delta-\beta\gamma=1$;
$\lambda_0$,~\dots, $\lambda_3$ and $\sigma$ are arbitrary constants with $\sigma\ne0$,
and $f$ is an arbitrary solution of~\eqref{eq:Power5FP}.	
\end{theorem}

The counterpart of~\eqref{eq:RemarkableFPNewSolutionsByG0} for
generating new solutions of the equation~\eqref{eq:Power5FP}
using elements of~$G_5$
takes, in the notation of Theorem~\ref{thm:Power5FPSymGroup},
the following form:
\begin{gather*}
\begin{split}
u={}&
\frac{{\rm e}^{\lambda_2/x-3\lambda_3(t-y/x)+3\lambda_3^2y^3+3\lambda_3\lambda_2y^2+\lambda_2^2y}}{\sigma(\gamma y+\delta)^4}
\exp\left(
-\frac{\gamma\hat x^2}{x^2(\gamma y+\delta)}
+\frac{3\gamma^2\hat x\hat t}{x(\gamma y+\delta)^2}
-\frac{3\gamma^3\hat t^2}{(\gamma y+\delta)^3}
\right)
\\[.5ex]
&
\times
\big(\hat x(\gamma y+\delta)-3\gamma\hat tx\big)
h\left(
\frac{\hat t}{(\gamma y+\delta)^3},\,
\frac{x(\gamma y+\delta)^2}{\hat x(\gamma y+\delta)-3\gamma\hat tx},\,
\frac{\alpha y+\beta}{\gamma y+\delta}
\right)
-f(t,x,y).
\end{split}
\end{gather*}

The essential Lie-invariant solutions of the equation~\eqref{eq:Power5FP} are exhausted by the families
\begin{gather*}
\solution\beta=5\colon\quad
u=|x|^{\frac54}\vartheta^\mu\big(\tfrac94\tilde\varepsilon t,|x|^{-\frac32}\big)
\quad\mbox{with}\quad \mu=\tfrac5{36},\quad \tilde\varepsilon:=\sgn x,
\\[.5ex]
\solution\beta=5\colon\quad
u=x|y|^{-\frac12}{\rm e}^{-\frac1{4x^2y}}\vartheta^0
\Big(\tfrac13{y^3}+2\varepsilon y-y^{-1},2t-(y+\varepsilon y^{-1})x^{-1}\Big)
\quad\mbox{with}\quad \varepsilon\in\{-1,1\},
\\[.5ex]
\solution\beta=5\colon\quad
u=x\vartheta^0\Big(\tfrac13y^3,t-yx^{-1}\Big),
\\[.5ex]
\solution\beta=5\colon\quad
u=x\vartheta^0\big(y,x^{-1}\big),
\end{gather*}
where $\vartheta^\mu=\vartheta^\mu(z_1,z_2)$ is an arbitrary solution
of the equation $\vartheta^\mu_1=\vartheta^\mu_{22}+\mu z_2^{-2}\vartheta^\mu$,
see footnote~\ref{fnt:SolutionsOfHeatEqs}.

The counterpart~$\Upsilon_{\mathfrak r_5}$ of the algebra~$\Upsilon_{\mathfrak r_0}$
is generated by the differential operators
\begin{gather*}
\mathrm P^3:=y^3\mathrm D_t-3x^2y^2\mathrm D_x+3(xy^2-t+x^{-1}y),\ \
\mathrm P^2:=y^2\mathrm D_t-2x^2y\mathrm D_x+(2yx+x^{-1}),
\\
\mathrm P^1:=y\mathrm D_t-x^2\mathrm D_x+x,\ \
\mathrm P^0:=\mathrm D_t,
\end{gather*}
which are associated with the Lie-symmetry vector fields $-\mathcal P^3$, $-\mathcal P^2$,$-\mathcal P^1$ and $-\mathcal P^0$
of the equation~\eqref{eq:Power5FP}.
The following assertion is just a particular case of a more general result~\cite[Corollary~20]{popo2024b}.

\begin{corollary}
The algebra of canonical representatives of generalized symmetries
of the equation~\eqref{eq:Power5FP} is
$\Sigma_5=\Lambda_5\lsemioplus\Sigma^{-\infty}_5$,
where
\begin{gather*}
\Lambda_5=\big\langle\big((\mathrm P^3)^{i_3}(\mathrm P^2)^{i_2}(\mathrm P^1)^{i_1}(\mathrm P^0)^{i_0}u\big)\p_u
\mid i_0,i_1,i_2,i_3\in\mathbb N_0\big\rangle,\quad
\Sigma^{-\infty}:=\big\{\mathcal Z(f)\big\}.
\end{gather*}
Here the parameter function~$f$ runs through the solution set of~\eqref{eq:Power5FP}.
\end{corollary}

To switch from the algebra~$\Sigma_5$ to the algebra of standard canonical representatives of generalized symmetries
of the equation~\eqref{eq:Power5FP} as an evolution equation,
we need to replace, in view of this equation, the operator~$\mathrm D_t$ in~$\mathrm P^3$,~\dots, $\mathrm P^0$
by the operator~$x^5\mathrm D_x^{\,2}-x\mathrm D_y$,
which complicates the expanded form of the representatives and increases their order.

The above solutions of the equation~\eqref{eq:RemarkableFP} generated by its recursion operators
are mapped to the following solutions of~\eqref{eq:Power5FP}:
\begin{gather*}
\solution\beta=5\colon\quad
u=(\mathrm P^1)^k\Big(|x|^{\frac54}\vartheta^\mu\big(\tfrac94\tilde\varepsilon t,|x|^{-\frac32}\big)\Big)
\quad\mbox{with}\quad \mu=\tfrac5{36},\quad \tilde\varepsilon:=\sgn x,
\\[.5ex]
\solution\beta=5\colon\quad u=(\mathrm P^1)^k\Big(x|y|^{-\frac12}{\rm e}^{-\frac1{4yx^2}}
\vartheta^0\big(\tfrac13y^3+2\varepsilon y-y^{-1},2t-x^{-1}(y+\varepsilon y^{-1})\big)\Big),
\ \varepsilon\in\{-1,1\},
\\
\solution\beta=5\colon\quad u=(\mathrm P^2)^k\Big(x\vartheta^0\big(\tfrac13y^3,t-x^{-1}y\big)\Big),
\\
\solution\beta=5\colon\quad u=(\mathrm P^3)^k\big(x\vartheta^0(y,x^{-1})\big).
\end{gather*}

\section{Fine Kolmogorov backward equation and its counterpart}\label{sec:FineFPandCounterpart}

Another distinguished equation within the classes~$\mathcal F'$, $\mathcal F$ and~$\bar{\mathcal F}$ is
the Kolmogorov backward equation $\mathcal F'_2$,
\begin{gather}\label{eq:FineFP}
\mathcal F'_2\colon\quad
u_t+xu_y=x^2u_{xx}.
\end{gather}
Its essential Lie invariance algebra is five-dimensional and nonsolvable, and modulo the point equivalence,
it is the only equation from the above classes with this property.
This is why in~\cite{kova2024a}, we refer to $\mathcal F'_2$ as to the {\it fine Kolmogorov backward equation}.
Some of the results from~\cite{kova2024a} are reviewed below.

The maximal Lie invariance algebra~$\mathfrak g_2$ of~\eqref{eq:FineFP} is spanned by the vector fields
\begin{gather*}
\mathcal P^y=\p_y,\quad
\mathcal D  =x\p_x+y\p_y,\quad
\mathcal K  = 2xy\p_x+y^2\p_y-ux\p_u,
\\
\mathcal P^t=\p_t,\quad
\mathcal I  =u\p_u,\quad
\mathcal Z(f)=f(t,x,y)\p_u,
\end{gather*}
where the parameter function $f$ ranges through the solution set of~\eqref{eq:FineFP}.

\begin{theorem}[\cite{kova2024a}]\label{thm:FineFPSymGroup}
The point symmetry pseudogroup~$G_2$ of the fine Kolmogorov backward equation~\eqref{eq:FineFP}
consists of the point transformations of the form
\begin{gather*}
\begin{split}
&\tilde t=t+\lambda,
\quad
\tilde x=\frac{\alpha\delta-\beta\gamma}{(\gamma y+\delta)^2}x,
\quad
\tilde y=\frac{\alpha y+\beta}{\gamma y+\delta},
\\[1ex]
&\tilde u=\sigma\exp\left(
-\frac{\gamma x}{\gamma y+\delta}
\right)
\big(u+f(t,x,y)\big)
\end{split}
\end{gather*}
where
$\alpha$, $\beta$, $\gamma$ and $\delta$ are arbitrary constants with $\alpha\delta-\beta\gamma=\pm1$
that are defined up to simultaneously alternating their signs;
$\lambda$ and $\sigma$ are arbitrary constants with $\sigma\ne0$,
and $f$ is an arbitrary solution of the equation~\eqref{eq:FineFP}.	
\end{theorem}

Given an arbitrary seed solution $u=h(t,x,y)$ of~\eqref{eq:FineFP}, 
an element of the group~$G_2$ represented in the notation of Theorem~\ref{thm:FineFPSymGroup} maps it 
to the solution
\begin{gather*}
\tilde u=\sigma\exp\left(-\frac{\gamma x}{\alpha-\gamma y}\right)
h\left(t-\lambda,
\frac{\alpha\delta-\beta\gamma}{(\alpha-\gamma y)^2}x,
\frac{\delta y-\beta}{\alpha-\gamma y}\right)
+f(t,x,y).
\end{gather*}

Essential $G_2^{\rm ess}$-inequivalent Lie-invariant solutions of~\eqref{eq:FineFP} in closed form
were constructed in~\cite{kova2024a} only by codimension-one Lie reductions,
\begin{gather*}
\solution\beta=2\colon\quad
u={\rm e}^{\mu t}|x|^{1/4}\vartheta^{4\mu+\frac34}(\varepsilon y,2\sqrt{|x|})
\quad\mbox{with}\quad \varepsilon:=\sgn x,
\\
\solution\beta=2\colon\quad
u={\rm e}^{-\frac14t}|x|^{1/2}\vartheta^0(t,\ln|x|),
\end{gather*}
where $\vartheta^\mu=\vartheta^\mu(z_1,z_2)$
is an arbitrary solution of the linear (1+1)-dimensional heat equation with the inverse square potential
$\vartheta^\mu_1=\vartheta^\mu_{22}+\mu z_2^{-2}\vartheta^\mu$, see footnote~\ref{fnt:SolutionsOfHeatEqs}.

For each of the values $\mu_n=\frac14n(n+1)-\frac3{16}$, $n\in\mathbb N$, the first of the above solution families
can be expressed using Darboux transformations in terms of the general solution~$\vartheta^0$
of the linear (1+1)-dimensional heat equation, see Section~\ref{sec:DarbouxTrans},
\begin{gather*}
\solution\beta=2\colon\quad
u={\rm e}^{\mu_nt}|x|^{\frac{2n+1}4}\left(\p_x-\frac1{2x}\right)^n\vartheta^0(\varepsilon y,2\sqrt{|x|})
\quad\mbox{with}\quad \mu_n:=\frac{n(n+1)}4-\frac3{16}.
\end{gather*}
\noprint{
\begin{gather*}
z=2\sqrt{|x|},\quad \p_z-\frac1z=\varepsilon\sqrt{|x|}\Big(\p_x-\frac1{2x}\Big),
\\
\Big[\p_x-\frac k{2x},(\varepsilon\sqrt{|x|})^{k-1}\Big]
\\ \quad{}
=(\varepsilon\sqrt{|x|})^{k-1}\frac{k-1}{2x}+(\varepsilon\sqrt{|x|})^{k-1}\Big(\p_x-\frac k{2x}\Big)
-(\varepsilon\sqrt{|x|})^{k-1}\Big(\p_x-\frac k{2x}\Big)
\\ \quad
=\varepsilon\sqrt{|x|})^{k-1}\frac{k-1}{2x}.
\end{gather*}
}

In~\cite{kova2024a} we also constructed two families of solutions of~\eqref{eq:FineFP}
that are invariant with respect to generalized symmetries of~\eqref{eq:FineFP},
\begin{gather*}
\solution\beta=2\colon\quad
u={\rm e}^{-\frac14 t}|x|^{n-\frac12}\sum_{s=0}^{n-1}\frac{y^sx^{-s}}{s!}
\left(\prod_{k=n-s}^{n-1}(2kx\p_x+k^2)\right)\vartheta^0(t,\ln|x|),
\\
\solution\beta=2\colon\quad
u=|x|^{1/4}\Big(
t\vartheta^{3/4}_2(\varepsilon y,2\sqrt{|x|})-
(\tfrac14t-1)|x|^{-1/2}\vartheta^{3/4}(\varepsilon y,2\sqrt{|x|})
\Big),
\end{gather*}
where $\vartheta^\mu_2$ denotes the derivative of~$\vartheta^\mu$ with respect to its second argument.
According to \cite[Theorem~14]{kova2024a},
for any $n\in\mathbb N$, the solutions from the former family,
which are invariant with respect to the generalized symmetries $(\mathrm D_y^nu)\p_u$,
can be expressed in terms of the iterative action by the Lie-symmetry operator~$\mathrm K:=2xy\mathrm D_x+y^2\mathrm D_y+x$
on Lie invariant solutions of~\eqref{eq:FineFP} from the second family of such solutions,
\smash{$u={\rm K}^{n-1}\big({\rm e}^{-\frac14t}|x|^{1/2}\theta(t,\ln|x|)\big)$}.
Moreover, this is the only nontrivial case of generating new solutions of~\eqref{eq:FineFP}
by acting Lie-symmetry operators.

Following Remark~\ref{rem:SimplerEquivTrans},
instead of~$\mathcal F'_3$ and the equivalence transformation~\eqref{eq:EssEquivTransOfF'},
we can consider the equation
\begin{gather}\label{eq:Power3FP}
u_t+xu_y=x^3u_{xx}
\end{gather}
as the counterpart of~$\mathcal F'_2$ with respect to the simpler transformation~\eqref{eq:EssEquivTransMod}
and push forward the above results for~$\mathcal F'_2$ by the latter transformation.

\begin{proposition}\label{pro:F3MIA}
The maximal Lie invariance algebra $\mathfrak g_3$ of the equation~\eqref{eq:Power3FP} is spanned by the vector fields
\begin{gather*}
\mathcal P^t=\p_t,\ \
\mathcal D=t\p_t-x\p_x-u\p_u,\ \
\mathcal K=t^2\p_t-2tx\p_x-(2t+x^{-1})u\p_u,
\\
\mathcal P^y=\p_y,\ \
\mathcal I=u\p_u,\ \
\mathcal Z(f)=f(t,x,y)\p_u,
\end{gather*}
where the parameter function $f$ runs through the solution set of the equation~\eqref{eq:Power3FP}.
\end{proposition}

\begin{theorem}\label{thm:Power3FPSymGroup}
The point symmetry pseudogroup~$G_3$ of the equation~\eqref{eq:Power3FP}
consists of the point transformations of the form
\begin{gather}\label{eq:Power3FPSymGroup}
\begin{split}
&\tilde t=\frac{\alpha t+\beta}{\gamma t+\delta},
\quad
\tilde x=\frac{(\gamma t+\delta)^2}{\alpha\delta-\beta\gamma}x,
\quad
\tilde y=y+\lambda,
\\[1ex]
&\tilde u=\sigma\frac{(\gamma t+\delta)^2}{\alpha\delta-\beta\gamma}\exp\left(
-\frac{\gamma x^{-1}}{\gamma t+\delta}
\right)
\big(u+f(t,x,y)\big),
\end{split}
\end{gather}
where
$\alpha$, $\beta$, $\gamma$ and $\delta$ are arbitrary constants with $\alpha\delta-\beta\gamma=\pm1$
that are defined up to simultaneously alternating their signs;
$\lambda$ and $\sigma$ are arbitrary constants with $\sigma\ne0$,
and $f$ is an arbitrary solution of the equation~\eqref{eq:Power3FP}.	
\end{theorem}

Thus, under the action of an element of~$G_3$, which is of the form~\eqref{eq:Power3FPSymGroup},
a given solution $u=h(t,x,y)$ of~\eqref{eq:Power3FP} is mapped to the solution
\begin{gather*}
\tilde u=\sigma\frac{\alpha\delta-\beta\gamma}{(\alpha-\gamma t)^2}\exp\left(-\frac{\gamma x^{-1}}{\alpha-\gamma t}\right)
h\left(
\frac{\delta t-\beta}{\alpha-\gamma t},
\frac{(\alpha-\gamma t)^2}{\alpha\delta-\beta\gamma}x,
y-\lambda
\right)
+f(t,x,y).
\end{gather*}

In the above notation, the constructed $G_2^{\rm ess}$-inequivalent closed-form Lie invariant solutions of $\mathcal F'_3$
are exhausted by the following:
\begin{gather*}
\solution\beta=3\colon\quad
u={\rm e}^{\mu y}|x|^{3/4}\vartheta^{4\mu+\frac34}(\varepsilon t,2|x|^{-1/2})
\quad\mbox{with}\quad \varepsilon:=\sgn x,
\\
\solution\beta=3\colon\quad
u={\rm e}^{-\frac14y}|x|^{1/2}\vartheta^0(y,-\ln|x|).
\end{gather*}

In the similar way for each of the values $\mu_n=\frac14n(n+1)-\frac3{16}$, $n\in\mathbb N$,
the first of the above solution families can be expressed using Darboux transformations in terms of the general solution~$\vartheta^0$
of the linear (1+1)-dimensional heat equation,
\begin{gather*}
\solution\beta=3\colon\quad
u={\rm e}^{\mu_ny}x|x|^{-\frac{2n+1}4}\left(x^2\p_x+\frac x2\right)^n\vartheta^0(\varepsilon t,2|x|^{-1/2})
\quad\mbox{with}\quad \mu_n:=\frac{n(n+1)}4-\frac3{16}.
\end{gather*}
\noprint{
\begin{gather*}
\solution\beta=3\colon\quad
u={\rm e}^{\mu_ny}|x|^{\frac34(2n+1)}
\left(\p_x+\frac{4n-3}{2x}\right)\left(\p_x+\frac{4n-7}{2x}\right)\cdots\left(\p_x+\frac5{2x}\right)\left(\p_x+\frac1{2x}\right)
\vartheta^0(\varepsilon t,2|x|^{-1/2})
\quad\mbox{with}\quad \mu_n:=\frac{n(n+1)}4-\frac3{16}.
\\
\left(\p_x+\frac1{2x}\right)x^{2k}f=4k\frac{x^{2k}}{2x}+x^{2k}f_x+\frac{x^{2k}}{2x}f
=x^{2k}\left(\p_x+\frac{4k+1}{2x}\right)f
\end{gather*}
}

The above solutions of the equation~\eqref{eq:FineFP}
that are invariant with respect to generalized symmetries
are mapped to the following solutions of~\eqref{eq:Power3FP}:
\begin{gather*}
\solution\beta=3\colon\quad
u={\rm e}^{-\frac14 y}x|x|^{\frac12-n}\sum_{s=0}^{n-1}\frac{t^sx^s}{s!}
\left(\prod_{k=n-s}^{n-1}(-2kx\p_x+k^2)\right)\vartheta^0(y,-\ln|x|),
\\
\solution\beta=3\colon\quad
u=x|x|^{-1/4}\Big(
xy\vartheta^{3/4}_2(\varepsilon t,2|x|^{-1/2})-
(\tfrac14y-1)|x|^{1/2}\vartheta^{3/4}(\varepsilon t,2|x|^{-1/2})
\Big).
\end{gather*}

\section{Generic case}\label{sec:GenericCase}

As the generic case, we consider the case $\beta\ne0,2,3,5$.
In view of Remark~\ref{rem:GaugingBeta}, up to the $G^\sim_{\mathcal F'}$-equivalence,
it suffices to consider only the equations
\begin{gather*}
\mathcal F'_\beta\colon\quad
u_t+xu_y=|x|^\beta u_{xx}, \quad \beta\in(-\infty,5/2]\setminus\{0,2\}.
\end{gather*}
According to Section~\ref{sec:GroupClassificationClassF'},
the essential Lie invariance algebra~$\mathfrak g^{\rm ess}_\beta$ of the equation~$\mathcal F'_\beta$
with such a value of~$\beta$ is spanned by the vector fields
\begin{gather*}
\mathcal P^t       :=\p_t,\quad
\mathcal P^y       :=\p_y,\quad
\mathcal D^\beta   :=(2-\beta)t\p_t+x\p_x+(3-\beta)y\p_y,\quad
\mathcal I         :=u\p_u.
\end{gather*}
Up to the skew-symmetry of the Lie bracket of vector fields, the only nonzero commutation relations
among the basis elements of~$\mathfrak g^{\rm ess}_\beta$ are the following:
\begin{gather*}
[\mathcal P^t,\mathcal D^\beta]=(2-\beta)\mathcal P^t,\quad
[\mathcal P^y,\mathcal D^\beta]=(3-\beta)\mathcal P^y.
\end{gather*}
The algebra~$\mathfrak g^{\rm ess}_\beta$ is isomorphic to the algebra $A_{3.4}^a\oplus A_1$ with $a=(2-\beta)/(3-\beta)$,
see~\cite{popo2003a} for notation, which is consistent with Mubarakzyanov's algebra numeration~\cite{muba1963b}.

\subsection{Point symmetry pseudogroup}\label{sec:GenCaseSymGroup}

\begin{theorem}\label{thm:GenCaseSymGroup}
(i) For $\beta\in\mathbb R\setminus\{0,2,5/2,3,5\}$, the point symmetry pseudogroup~$G_\beta$ of the equation~$\mathcal F'_\beta$
consists of the point transformations
\begin{gather}\label{eq:GenCaseSymGroupBetaNe52}
\tilde t=|\alpha|^{2-\beta}t+\lambda_0,
\quad
\tilde x=\alpha x,
\quad
\tilde y=\alpha|\alpha|^{2-\beta}y+\lambda_1,
\quad
\tilde u=\sigma u+f(t,x,y),
\end{gather}
where
$\alpha$, $\lambda_0$, $\lambda_1$ and $\sigma$ are arbitrary constants with $\alpha\sigma\ne0$,
and $f$ is an arbitrary solution of~$\mathcal F'_\beta$.

(ii) In comparison with the generic case of~$\beta$,
the point symmetry pseudogroup $G_{5/2}$ of the equation~$\mathcal F'_{5/2}$
is extended by the point transformations
\begin{gather}\label{eq:GenCaseSymGroupBeta=52}
\tilde t=\sgn(x)\alpha|\alpha|^{2-\beta}y+\lambda_1,
\quad
\tilde x=\frac\alpha x,
\quad
\tilde y=\sgn(x)|\alpha|^{2-\beta}t+\lambda_0,
\quad
\tilde u=\sigma\frac ux+f(t,x,y).
\end{gather}
\end{theorem}

\begin{proof}
The standard algebraic method and its known modifications for computing point symmetry (pseudo)groups of systems of differential equations,
see, e.g., \cite{bihl2011b,boyk2022a,card2013a,hydo1998a,hydo1998b,hydo2000b,hydo2000A,malt2021a,opan2020a}
and references therein,
do not work properly in the case of linear well-determined systems of partial differential equations.
This is why we combine the algebraic part of this method with the advanced direct method in an original way.
More specifically, Theorem~\ref{thm:EquivalenceGroupFPsuperClass} implies
that for any fixed~$\beta$, any point symmetry transformation~$\Phi$ of the equation~$\mathcal F'_\beta$
is of the form~\eqref{eq:ClassFbarTransPart}, where
$U^0/U^1$ is a solution of~$\mathcal F'_\beta$ in view of the last equation in~\eqref{eq:A^2BCTrans}
and the homogeneousness of~$\mathcal F'_\beta$ (i.e., $C=0$ and \smash{$\tilde C=0$}).
Thus, up to composing with transformations of linear superposition of solutions,
we can assume that $U^0=0$.
Then the transformation~$\Phi$ preserves the algebra~$\mathfrak g^{\rm ess}_\beta$
and thus induces an automorphism of this algebra.

Fixing the basis $(\mathcal P^t,\mathcal P^y,\mathcal D^\beta,\mathcal I)$,
we identify each automorphism of~$\mathfrak g^{\rm ess}_\beta$ with its matrix in this basis.
The automorphisms of the algebra $A_{3.4}^a\oplus A_1$, which is isomorphic to~$\mathfrak g^{\rm ess}_\beta$,
were first computed in~\cite[Table~1]{chri2003a} in a consistent basis
using a different algebra notation
($A_{3.4}\oplus A_1$ instead of $A_{3.4}^{-1}\oplus A_1$ and
$A_{3.5}^\alpha\oplus A_1$ instead of $A_{3.4}^\alpha\oplus A_1$ for $-1<\alpha<1$);
see also \cite[Section~A.2]{popo2003a} for a more convenient exposition.
In the case $\beta\ne5/2$,
the automorphism group ${\rm Aut}(\mathfrak g^{\rm ess}_\beta)$ of~$\mathfrak g^{\rm ess}_\beta$
is constituted by the matrices
\begin{gather}\label{eq:Aut_g_beta}
\begin{pmatrix}
a_{11} &   0    &  a_{13}  &   0   \\
  0    & a_{22} &  a_{23}  &   0   \\
  0    &   0    &    1     &   0   \\
  0    &   0    &  a_{43}  & a_{44}
\end{pmatrix},
\end{gather}
where all the involved parameters~$a_{ij}$ are arbitrary real constants with $a_{11}a_{22}a_{44}\ne0$.
If $\beta=5/2$, then the group ${\rm Aut}(\mathfrak g^{\rm ess}_\beta)$
is extended, in comparison with the generic case, by the matrices
\begin{gather*}
\begin{pmatrix}
  0    & a_{12} &  a_{13}  &   0   \\
a_{21} &   0    &  a_{23}  &   0   \\
  0    &   0    &   -1     &   0   \\
  0    &   0    &  a_{43}  & a_{44}
\end{pmatrix},
\end{gather*}
where all the involved parameters~$a_{ij}$ are arbitrary real constants with $a_{12}a_{21}a_{44}\ne0$.
The singularity of the value $\beta=5/2$ is related to the fact that the equivalence transformation~$\mathscr J'$
preserves this value of~$\beta$ and, therefore,
it induces an additional point symmetry transformation of~$\mathcal F'_{5/2}$.
Factoring out the induced transformation from~$G_{5/2}$,
we can consider only automorphisms of the form~\eqref{eq:Aut_g_beta} for each relevant value~$\beta$.

Pushing forward the basis elements of~$\mathfrak g^{\rm ess}_\beta$ by~$\Phi$,
we obtain the equalities
\begin{gather*}
\Phi_*(\mathcal P^t)=a_{11}\mathcal P^t,\quad
\Phi_*(\mathcal P^y)=a_{22}\mathcal P^y,\quad
\Phi_*(\mathcal I)=a_{44}\mathcal I,\\
\Phi_*(\mathcal D^\beta)=\mathcal D^\beta+a_{13}\mathcal P^t+a_{23}\mathcal P^y+a_{43}\mathcal I.
\end{gather*}
for some values the involved parameters~$a_{ij}$.
Splitting these equalities componentwise results in the following system of determining equations
for the components of~$\Phi$:
\begin{gather*}
T_t=a_{11},\quad
X_t=Y_t=U^1_t=0,
\quad
Y_y=a_{22},\quad
T_y=X_y=U^1_y=0,
\quad
U^1=a_{44},
\quad
xX_x=X,
\end{gather*}
which integrates to
\[
T=a_{11}t+b_1,\quad
X=cx,\quad
Y=a_{22}y+b_2,\quad
U=a_{44}u,
\]
where $b_1$, $b_2$ and $c$ are real constants with $c\ne0$.
Proceeding with the direct method, we take the equation~$\mathcal F'_\beta$ in the variables with tildes,
expand, using the chain rule, all the involved derivatives in terms of the variables without tildes,
substitute the expression for $u_t$ in view of the equation~$\mathcal F'_\beta$ into the expanded equation
and split the result with respect to all the remaining jet variables.
This gives two equations for the constant parameters~$a_{11}$, $a_{22}$ and~$c$,
\[
a_{11}=|c|^{2-\beta},\quad
a_{22}=c|c|^{2-\beta}.
\]
The final step is to redenote the remaining constant parameters $b_1$, $b_2$, $c$ and~$a_{44}$.
\end{proof}

\begin{remark}
The essential point symmetry group~$G^{\rm ess}_\beta$
coincides with the kernel point symmetry group~$G^\cap_{\mathcal F'}$ of the class~$\mathcal F'$
if $\beta\in\mathbb R\setminus\{0,2,5/2,3,5\}$
and with the group $\pi_*G^\sim_{\mathcal F'}$ if $\beta\ne5/2$,
where~$\pi$ is the natural projection $\pi\colon\mathbb R^5_{t,x,y,u,\beta}\to\mathbb R^4_{t,x,y,u}$,
cf.\ Remark~\ref{rem:EssEquivTrans}.
In other words, the group~$G^{\rm ess}_\beta$ is extended for value $\beta=5/2$ by the element~$\pi_*\mathscr J'$.
\end{remark}

\begin{remark}
A complete list of independent discrete point symmetries
of the equation~$\mathcal F'_\beta$ is exhausted by the transformations
$\pi_*\mathscr I'_u$ and $\pi_*\mathscr I'_{\rm s}$ if $\beta\ne5/2$
and by the transformations $\pi_*\mathscr I'_u$ and~$\pi_*(\mathscr I'_{\rm s}\circ\mathscr J')$ if $\beta\ne5/2$,
cf.\  Remark~\ref{rem:EssEquivTrans} and Corollary~\ref{cor:DiscrPointSymGroupOfF'}.
\end{remark}

\subsection{Classification of subalgebras}\label{sec:SubalgGbeta}

The established isomorphism between~\smash{$\mathfrak g^{\rm ess}_\beta$} and~$A_{3.4}^a\oplus A_1$ with $a=(2-\beta)/(3-\beta)$
and results of~\cite{pate1977a} (in a different algebra notation with
$A_{3.4}\oplus A_1$ instead of $A_{3.4}^{-1}\oplus A_1$ and
$A_{3.5}^\alpha\oplus A_1$ instead of $A_{3.4}^\alpha\oplus A_1$ for $-1<\alpha<1$)
straightforwardly give the classification of subalgebras of~$\mathfrak g^{\rm ess}_\beta$.

\begin{lemma}\label{lem:SubalgerbasGenCase}
A complete list of inequivalent (with respect to inner automorphisms)
subalgebras of the algebra $\mathfrak g_\beta^{\rm ess}$ with $\beta\in\mathbb R\setminus\{0,2,3,5\}$
is exhausted by the following subalgebras:
\begin{gather*}
1{\rm D}\colon\
\mathfrak s_{1.1}^\kappa              :=\langle\mathcal D^\beta+\kappa\mathcal I\rangle,\quad
\mathfrak s_{1.2}^{\varepsilon,\kappa}:=\langle\mathcal P^y+\varepsilon\mathcal P^t+\kappa\mathcal I\rangle,\quad
\mathfrak s_{1.3}^\delta              :=\langle\mathcal P^y+\delta\mathcal I\rangle,
\\
\hphantom{1{\rm D}\colon\ }
\mathfrak s_{1.4}^\delta              :=\langle\mathcal P^t+\delta\mathcal I\rangle,\quad
\mathfrak s_{1.5}                     :=\langle\mathcal I\rangle,
\\[.5ex]
2{\rm D}\colon\
\mathfrak s_{2.1}^\delta:=\langle\mathcal P^y,\mathcal P^t+\delta\mathcal I\rangle,\quad
\mathfrak s_{2.2}^{\varepsilon,\kappa}:=\langle\mathcal P^y+\varepsilon\mathcal I, \mathcal P^t+\kappa\mathcal I\rangle,\quad
\mathfrak s_{2.3}^\kappa:=\langle\mathcal D^\beta+\kappa\mathcal I,\mathcal P^y\rangle,\quad
\\
\hphantom{2{\rm D}\colon\ }
\mathfrak s_{2.4}^\kappa:=\langle\mathcal D^\beta+\kappa\mathcal I,\mathcal P^t\rangle,\
\mathfrak s_{2.5}       :=\langle\mathcal D^\beta,\mathcal I\rangle,\
\mathfrak s_{2.6}^\delta:=\langle\mathcal P^y+\delta\mathcal P^t,\mathcal I\rangle,\
\mathfrak s_{2.7}                     :=\langle\mathcal P^t,\mathcal I\rangle,
\\[.5ex]
3{\rm D}\colon\
\mathfrak s_{3.1}^\kappa:=\langle\mathcal D^\beta+\kappa\mathcal I,\mathcal P^t,\mathcal P^y\rangle,\quad
\mathfrak s_{3.2}       :=\langle\mathcal D^\beta,\mathcal P^y,\mathcal I\rangle,\quad
\mathfrak s_{3.3}       :=\langle\mathcal D^\beta,\mathcal P^t,\mathcal I\rangle,\quad
\end{gather*}
where $\delta\in\{-1, 0, 1\}$, $\varepsilon\in\{-1, 1\}$ and $\kappa$ is an arbitrary constant.
\end{lemma}

\begin{corollary}
(i) For the generic case of $\beta\in\mathbb R\setminus\{0,2,5/2,3,5\}$,
a complete list of $G^{\rm ess}_\beta$-inequivalent
subalgebras of the algebra $\mathfrak g_\beta^{\rm ess}$ is exhausted by the subalgebras
\begin{gather*}
1{\rm D}\colon\
\mathfrak s_{1.1}^\kappa ,\quad
\mathfrak s_{1.2}^{1,\kappa},\quad
\mathfrak s_{1.3}^{\delta'},\quad
\mathfrak s_{1.4}^\delta,\quad
\mathfrak s_{1.5},
\\
2{\rm D}\colon\
\mathfrak s_{2.1}^\delta,\quad
\mathfrak s_{2.2}^{1,\kappa},\quad
\mathfrak s_{2.3}^\kappa,\quad
\mathfrak s_{2.4}^\kappa,\quad
\mathfrak s_{2.5},\quad
\mathfrak s_{2.6}^{\delta'},\quad
\mathfrak s_{2.7},
\\
3{\rm D}\colon\
\mathfrak s_{3.1}^\kappa,\quad
\mathfrak s_{3.2},\quad
\mathfrak s_{3.3},\quad
\end{gather*}
where $\delta'\in\{0,1\}$ and $\kappa$ is an arbitrary constant.

(ii) An analogous list for $\beta=5/2$ is given by the subagebras
\begin{gather*}
1{\rm D}\colon\
\mathfrak s_{1.1}^\kappa ,\quad
\mathfrak s_{1.2}^{1,\kappa},\quad
\mathfrak s_{1.3}^{\delta'},\quad
\mathfrak s_{1.5},
\\
2{\rm D}\colon\
\mathfrak s_{2.1}^\delta,\quad
\mathfrak s_{2.2}^{1,\kappa},\quad
\mathfrak s_{2.3}^\kappa,\quad
\mathfrak s_{2.5},\quad
\mathfrak s_{2.6}^{\delta'},\quad
\mathfrak s_{2.7},
\\
3{\rm D}\colon\
\mathfrak s_{3.1}^\kappa,\quad
\mathfrak s_{3.2}.
\end{gather*}
\end{corollary}

Among the subalgebras listed in Lemma~\ref{lem:SubalgerbasGenCase},
only the following satisfy the transversality condition and are thus appropriate for Lie reduction:
\begin{gather*}
1{\rm D}\colon\
\mathfrak s_{1.1}^\kappa,\quad
\mathfrak s_{1.2}^{\varepsilon,\kappa},\quad
\mathfrak s_{1.3}^\delta,\quad
\mathfrak s_{1.4}^\delta,
\\
2{\rm D}\colon\
\mathfrak s_{2.1}^\delta,\quad
\mathfrak s_{2.2}^{\varepsilon,\kappa},\quad
\mathfrak s_{2.3}^\kappa,\quad
\mathfrak s_{2.4}^\kappa,
\\
3{\rm D}\colon\	
\mathfrak s_{3.1}^\kappa.
\end{gather*}

Instead of the classification of Lie reductions of a single equation~$\mathcal F'_\beta$
modulo its essential point-symmetry (pseudo)group~\smash{$G^{\rm ess}_{\beta}$},
it is convenient to classify Lie reductions within the class~$\mathcal F'$
up to the $G^\sim_{\mathcal F'}$-equivalence.
More specifically, pushing forward the subalgebras~$\mathfrak s_{1.3}^\delta$ and~$\mathfrak s_{2.3}^\kappa$
of the algebra~\smash{$\mathfrak g^{\rm ess}_\beta$} by the equivalence transformation~\eqref{eq:EssEquivTransOfF'},
we obtain the subalgebras~$\mathfrak s_{1.4}^\delta$ and~$\mathfrak s_{2.4}^\kappa$
of the algebra~\smash{$\mathfrak g^{\rm ess}_{5-\beta}$}, respectively.
In view of this, we can without loss of generality omit the Lie reductions
with respect to the subalgebras~$\mathfrak s_{1.4}^\delta$ and~$\mathfrak s_{2.4}^\kappa$.
However, this requires us to relax the restriction on the parameter~$\beta$,
from $\beta\in(-\infty,5/2]\setminus\{0,2\}$ to $\beta\in\mathbb R\setminus\{0,2,3,5\}$.
The suggested adjustment simplifies the Lie reduction procedure for
the equations~$\mathcal F'_\beta$ with $\beta\in\mathbb R\setminus\{0,2,3,5\}$
by reducing the number of cases to be considered.

\subsection{Lie reductions of codimension one}\label{sec:FPGenBetaLieRedCoD1}

We use the subalgebras $\mathfrak s_{1.1}^\kappa$, $\mathfrak s_{1.2}^{\varepsilon,\kappa}$
and $\mathfrak s_{1.3}^\delta$ for carrying out codimension-one Lie reductions.
For each of these subalgebras, we construct an ansatz for $u$
with new unknown function $w=w(z_1,z_2)$ of two new independent variables $(z_1,z_2)$
and the corresponding reduced partial differential equation for~$w$,
which are listed below.
In what follows the subscripts in~$w_1$ and~$w_2$ indicate the derivatives of~$w$ with respect to~$z_1$ and~$z_2$, respectively.

\medskip\noindent
{\bf 1.1.}
$\mathfrak s_{1.1}^\kappa=\langle\mathcal D^\beta+\kappa\mathcal I\rangle$ with $\kappa\in\mathbb R$.
\begin{gather*}
z_1:=y|t|^{-(\beta-3)/(\beta-2)},\quad
z_2:=x|t|^{1/(\beta-2)},\quad
u   =w(z_1,z_2)|t|^{\kappa(\beta-3)/(\beta-2)};
\\
(\epsilon(\beta-2)z_2+(3-\beta)z_1)w_1=\epsilon(\beta-2)|z_2|^\beta w_{22}-z_2w_2-\kappa(\beta-3)w,\quad
\text{where}\quad\epsilon=\sgn t.
\end{gather*}
For each value of $(\kappa,\epsilon)$,
the maximal Lie invariance algebra of reduced equation~$1.1^{\kappa\epsilon}$
is trivial as for a linear homogeneous equation
since it is spanned by the vector fields~$w\p_w$ and~$f(z_1,z_2)\p_w$,
where the parameter function~$f$ runs through the solution set of reduced equation~$1.1^{\kappa\epsilon}$.%
\footnote{\label{fnt:MLIARedEqs}
The construction of the maximal Lie invariance algebras
of reduced equations~$1.1^{\kappa\epsilon}$ and~$1.2^{\kappa\epsilon}$
is too complicated.
The corresponding system of determining equations cannot be solved by standard symbolic computation programs
for studying overdetermined systems of differential equations.
The main reason for this is the parameterization of the determining equations
by~$\beta$ and~$\kappa$ in a cumbersome way.
To overcome the computational obstacles, we combine
general techniques from the theory of compatibility of systems of differential equations
and certain specific tools with
the classical results by Lie on group classification of (1+1)-dimensional linear evolution equations~\cite{lie1881a};
see modern treatment of these results in~\cite{opan2022a}.
}
Moreover, this entire algebra is induced by
the normalizer ${\rm N}_{\mathfrak g_\beta}(\mathfrak s_{1.1}^\kappa)$
of the subalgebra $\mathfrak s_{1.1}^\kappa$ in the algebra~$\mathfrak g_\beta$,
\[
{\rm N}_{\mathfrak g_\beta}(\mathfrak s_{1.1}^\kappa)
=\langle\mathcal D^\beta,\mathcal I\rangle\lsemioplus\langle\hat f(t,x,y)\p_u\rangle,
\]
where the parameter function~$\hat f$ runs through the set of $\mathfrak s_{1.1}^\kappa$-invariant solutions of $\mathcal F'_\beta$.

\medskip\noindent
{\bf 1.2.}
$\mathfrak s_{1.2}^{\varepsilon,\kappa}=\langle\mathcal P^y+\varepsilon\mathcal P^t+\kappa\mathcal I\rangle$
with $\varepsilon=\pm1$ and $\kappa\in\mathbb R$.
\begin{gather*}
z_1:=y-\varepsilon t,\quad
z_2:=x,\quad
u   ={\rm e}^{\varepsilon\kappa t}w(z_1,z_2);
\\
(z_2-\varepsilon)w_1=|z_2|^\beta w_{22}-\varepsilon\kappa w.
\end{gather*}
For each value of $(\varepsilon,\kappa)$,
the maximal Lie invariance algebra of reduced equation $1.2^{\varepsilon\kappa}$
coincides with the span $\langle\p_1,w\p_w,f(z_1,z_2)\p_w\rangle$,
where the parameter function~$f$ runs through the solution set of this equation,
see again footnote~\ref{fnt:MLIARedEqs}.
This entire algebra is induced by the normalizer ${\rm N}_{\mathfrak g_\beta}(\mathfrak s_{1.2}^{\varepsilon,\kappa})$
of the subalgebra $\mathfrak s_{1.2}^{\varepsilon,\kappa}$ in the algebra $\mathfrak g_\beta$,
\[
{\rm N}_{\mathfrak g_\beta}(\mathfrak s_{1.2}^{\varepsilon,\kappa})
=\langle\mathcal P^t,\mathcal P^y,\mathcal I\rangle\lsemioplus\langle\hat f(t,x,y)\p_u\rangle,
\]
where the parameter function~$\hat f$ runs through
the set of $\mathfrak s_{1.2}^{\varepsilon,\kappa}$-invariant solutions of $\mathcal F'_\beta$.

\medskip\noindent
{\bf 1.3.}
$\mathfrak s_{1.3}^\delta=\langle\mathcal P^y+\delta\mathcal I\rangle$ with $\delta\in\{-1, 0, 1\}$.
\begin{gather*}
z_1:=t,\quad
z_2:=x,\quad
u   ={\rm e}^{\delta y}w(z_1,z_2);
\\
w_1=|z_2|^\beta w_{22}-\delta z_2w.
\end{gather*}

Using the point transformation $\tilde z_1=z_1$, $\tilde z_2=2(\beta-2)^{-1}z_2|z_2|^{-\beta/2}$
and $\tilde w=|z_2|^{-\beta/4}w$, we further map the reduced equation to its canonical form
within the class of (1+1)-dimensional linear second-order evolution equations,
\begin{gather}\label{eq:GenRedEq}
\tilde w_1=\tilde w_{22}+
\left(\frac{\beta(\beta-4)}{4\tilde z_2^2(\beta-2)^2}
-\epsilon\delta\left(\frac{(\beta-2)^2}4\tilde z_2^2\right)^{\frac1{2-\beta}}
\right)\tilde w,
\end{gather}
where $\epsilon=\sgn z_2$.

For any~$\beta$,
the normalizer~${\rm N}_{\mathfrak g_\beta}(\mathfrak s_{1.3}^\delta)$
of the subalgebra $\mathfrak s_{1.3}^\delta$ in the algebra $\mathfrak g_\beta$ is
\begin{gather*}
{\rm N}_{\mathfrak g_\beta}(\mathfrak s_{1.3}^\delta)=\mathfrak g_\beta^{\rm ess}\lsemioplus\langle\hat f(t,x,y)\p_u\rangle
\quad\mbox{if}\quad \delta=0,
\\
{\rm N}_{\mathfrak g_\beta}(\mathfrak s_{1.3}^\delta)
=\langle\mathcal P^t,\mathcal P^y,\mathcal I\rangle\lsemioplus\langle\hat f(t,x,y)\p_u\rangle
\quad\mbox{if}\quad \delta=\pm1,
\end{gather*}
where the parameter function~$\hat f$ runs through the set of~$\mathfrak s_{1.3}^\delta$-invariant solutions
of the equation~$\mathcal F'_\beta$.

The equation~\eqref{eq:GenRedEq} with $\delta=0$ and $\beta=4$ coincides with the linear (1+1)-dimensional heat equation
$\tilde w_1=\tilde w_{22}$.
Denoting its general solution by $\tilde w=\vartheta^0(\tilde z_1,\tilde z_2)$, cf.\ footnote~\ref{fnt:SolutionsOfHeatEqs},
results in the solution family
\begin{gather*}
\solution\beta=4\colon\
u=x\vartheta^0(t,x^{-1}).
\end{gather*}

The maximal Lie invariance algebra of the linear (1+1)-dimensional heat equation $\tilde w_1=\tilde w_{22}$
is spanned by the vector fields
\begin{gather*}
\p_1,\quad
\mathcal D:=2z_1\p_1+z_2\p_2-\tfrac12\tilde w\p_{\tilde w},\quad
\mathcal K:=z_1^2\p_1+z_1z_2\p_2-\tfrac14(z_2^2+2z_1)\tilde w\p_w,
\\
z_1\p_2-\tfrac12z_2\tilde w\p_{\tilde w},\quad
\p_2,\quad
\tilde w\p_{\tilde w},\quad
f(z_1,z_2)\p_{\tilde w},
\end{gather*}
where the parameter function~$f$ runs through the solution set of this equation.

The subalgebra $\langle\p_1,
\mathcal D,
\tilde w\p_{\tilde w},
f(z_1,z_2)\p_{\tilde w}\rangle$
is induced by the normalizer~${\rm N}_{\mathfrak g_4}(\mathfrak s_{1.3}^0)$.
Thus, the vector fields~$\mathcal K$, $z_1\p_2-\tfrac12z_2\tilde w\p_{\tilde w}$ and~$\p_2$
exhaust a complete list of independent genuine hidden symmetries of the equation~$\mathcal F'_4$.

If $\delta=0$ and $\beta\ne4$, then the equation~\eqref{eq:GenRedEq} is the linear (1+1)-dimensional heat equation
with the inverse square potential,
\[
\tilde w_1=\tilde w_{22}+\frac{\beta(\beta-4)}{4\tilde z_2^2(\beta-2)^2}\tilde w,
\]
and thus its general solution can be denoted as $\tilde w=\vartheta^\mu(\tilde z_1,\tilde z_2)$
with $\mu=\beta(\beta-4)(\beta-2)^{-2}/4$.
The maximal Lie invariance algebra of this equation is spanned by the vector fields
\begin{gather*}
\p_{\tilde z_1},\quad
\tilde{\mathcal D}:=\tilde z_1\p_{\tilde z_1}+\tfrac12\tilde z_2\p_{\tilde z_2}-\tfrac14\tilde w\p_{\tilde w},\quad
\tilde{\mathcal K}:=\tilde z_1^2\p_{\tilde z_1}+\tilde z_1\tilde z_2\p_{\tilde z_2}
-\tfrac14(\tilde z_2^2-2\tilde z_2)\tilde w\p_{\tilde w},\\
\tilde w\p_{\tilde w},\quad
f(\tilde z_1,\tilde z_2)\p_{\tilde w},
\end{gather*}
where $f(\tilde z_1,\tilde z_2)$ is an arbitrary solution of it.
The subalgebra
$\langle\p_{\tilde z_1},\tilde{\mathcal D},\tilde w\p_{\tilde w},f(\tilde z_1,\tilde z_2)\p_{\tilde w}\rangle$
is induced by the normalizer~${\rm N}_{\mathfrak g_\beta}(\mathfrak s_{1.3}^0)$.
Thus, the vector field $\tilde{\mathcal K}$
is a single independent genuine hidden Lie-symmetry of the equation~$\mathcal F'_\beta$ with $\beta\ne4$.

Pulling back the general solution $\tilde w=\vartheta^\mu(\tilde z_1,\tilde z_2)$ with $\mu=\beta(\beta-4)(\beta-2)^{-2}/4$
by the above transformation and substituting the resulting expression into the ansatz,
we obtain the following solution family of the equation~$\mathcal F'_\beta$:
\begin{gather}\label{eq:FP_gauged_classMainSolutionFamily}
\solution\beta\in\mathbb R\setminus\{0,2,3,5\}\colon\
u=|x|^{\beta/4}\vartheta^\mu(t,2(\beta-2)^{-1}x|x|^{-\beta/2}).
\end{gather}

Let $\delta=\pm1$ and $\beta\ne1$.
The maximal Lie invariance algebra of the reduced equation $1.3^\delta$
is equal to the span $\langle\p_1,w\p_w,f(z_1,z_2)\p_w\rangle$
with $f$ running through the solution set of this equation,
which is entirely induced by the normalizer~${\rm N}_{\mathfrak g_\beta}(\mathfrak s_{1.3}^\delta)$.

A specific value of~$\beta$ for $\delta=\pm1$ is $\beta=1$.
In this case, we can alternate the sign of~$\tilde z_2$ and set $\tilde z_2=2|z_2|^{1/2}$.
The corresponding modified reduced equations~\eqref{eq:GenRedEq} with
$\epsilon\delta=-1$ and $\epsilon\delta=1$,
$
\tilde w_1=\tilde w_{22}-\left(\frac34\tilde z_2^{-2}+\frac14\epsilon\delta\tilde z_2^2\right)\tilde w,
$
are respectively mapped by the point transformations
\begin{gather*}
\hat z_1=\frac12\tan\tilde z_1,\
\hat z_2=\frac{\sqrt2}2\frac{\tilde z_2}{\cos\tilde z_1},\
\hat w=\sqrt{|\cos\tilde z_1|}{\rm e}^{-\frac14\tilde z_2^{\,2}\tan\tilde z_1}\tilde w,
\\
\hat z_1=\frac14{\rm e}^{2\tilde z_1},\
\hat z_2=\frac{\sqrt2}2{\rm e}^{\tilde z_1}\tilde z_2,\
\hat w={\rm e}^{-\frac14\tilde z_2^{\,2}-\frac12\tilde z_1}\tilde w
\end{gather*}
to the equation of the same form with $\delta=0$ in terms of the variables with hats.
As a result, instead of the single solution family~\eqref{eq:FP_gauged_classMainSolutionFamily} as in the generic case of~$\beta$,
we construct three families of solutions that are parameterized by
arbitrary solutions of the linear (1+1)-dimensional heat equation with the inverse square potential
with $\mu=-\frac34$,
\begin{gather}\nonumber
\solution\beta=1\colon\
u=|x|^{1/4}\vartheta^\mu(t,2\sqrt{|x|}),
\\ \label{eq:Beta1ExactSol2}
\solution\beta=1\colon\
u=|x|^{1/4}{\rm e}^{\delta y}\frac{{\rm e}^{|x|\tan t}}{\sqrt{|\cos t|}}
\vartheta^\mu\left(\frac12\tan t,\frac{\sqrt{2|x|}}{\cos t}\right),
\\ \label{eq:Beta1ExactSol3}
\solution\beta=1\colon\
u=|x|^{1/4}{\rm e}^{x+\frac12t+\delta y}\vartheta^\mu\left(\frac14{\rm e}^{2t},\sqrt{2|x|}{\rm e}^t\right).
\end{gather}
The maximal Lie invariance algebras of the modified reduced equations~\eqref{eq:GenRedEq}
with $\epsilon\delta=-1$ and $\epsilon\delta=1$ are respectively spanned by the vector fields
\begin{gather}\label{eq:MIAbeta1-1}
\begin{split}&
2\cos(2\tilde z_1)\p_{\tilde z_1}-2\sin(2\tilde z_1)\tilde z_2\p_{\tilde z_2}+(\cos(2z_1)\tilde z_2^2+\sin(2z_1))\tilde w\p_{\tilde w},
\\&
2\sin(2\tilde z_1)\p_{\tilde z_1}+2\cos(2z_1)\tilde z_2\p_{\tilde z_2}+(\sin(2\tilde z_1)\tilde z_2^2-\cos(2z_1))\tilde w\p_{\tilde w},
\\&
\p_{\tilde z_1},\quad
\tilde w\p_{\tilde w},\quad
f(\tilde z_1,\tilde z_2)\p_{\tilde w}
\end{split}
\end{gather}
and
\begin{gather}\label{eq:MIAbeta1+1}
\begin{split}&
{\rm e}^{-2\tilde z_1}\big(2\p_{\tilde z_1}-2\tilde z_2\p_{\tilde z_2}+(1-\tilde z_2^2)\tilde w\p_{\tilde w}\big),\quad
{\rm e}^{2\tilde z_1}\big(2\p_{\tilde z_1}+2\tilde z_2\p_{\tilde z_2}-(1+\tilde z_2^2)\tilde w\p_{\tilde w}\big),
\\&
\p_{\tilde z_1},\quad
\tilde w\p_{\tilde w},\quad
f(\tilde z_1,\tilde z_2)\p_{\tilde w},
\end{split}
\end{gather}
and the normalizer~${\rm N}_{\mathfrak g_1}(\mathfrak s_{1.3}^\delta)$ induces
the subalgebras
$\langle\p_{\tilde z_1},\tilde w\p_{\tilde w},f(\tilde z_1,\tilde z_2)\p_{\tilde w}\rangle$
in each of these algebras.
Here the parameter function~$f$ runs through the solution set of the respective equation.
Therefore, a complete list of independent genuine hidden symmetries of the equation~$\mathcal F'_1$
related to each of the reductions under consideration
is exhausted by the first two vector fields presented in~\eqref{eq:MIAbeta1-1} and~\eqref{eq:MIAbeta1+1}.

\begin{remark}
Mapping the obtained solutions using the discrete equivalence transformation~\eqref{eq:EssEquivTransOfF'}
or its counterpart~\eqref{eq:EssEquivTransMod} for equations~\eqref{eq:FPWithRationalBeta},
we obtain the solutions of~$\mathcal F'_\beta$ invariant with respect to the subalgebra
$\mathfrak s_{1.4}^\delta=\langle\mathcal P^t+\delta\mathcal I\rangle$,
\begin{gather}\nonumber
\solution\beta=1\colon\
u=\vartheta^0(y,x),
\\\nonumber
\solution\beta\in\mathbb R\setminus\{0,2,3,5\}\colon\
u=|x|^{(1-\beta)/4}\vartheta^\mu\big(y\sgn x,2(3-\beta)^{-1}|x|^{(3-\beta)/2}\sgn x\big),
\\\nonumber
\solution\beta=4\colon\
u=|x|^{3/4}\vartheta^{-3/4}\big(y,2|x|^{-1/2}\big),
\\\label{eq:Beta4ExactSol2}
\solution\beta=4\colon\
u=|x|^{3/4}{\rm e}^{\delta t}\frac{{\rm e}^{|x|^{-1}\tan y}}{\sqrt{|\cos y|}}
\vartheta^{-3/4}\left(\frac12\tan y,\frac{\sqrt{2|x|^{-1}}}{\cos y}\right),
\\\label{eq:Beta4ExactSol3}
\solution\beta=4\colon\
u=|x|^{3/4}{\rm e}^{x^{-1}+\frac12y+\delta t}\vartheta^{-3/4}\left(\frac14{\rm e}^{2y},\sqrt{2|x|^{-1}}{\rm e}^y\right).
\end{gather}
\end{remark}

\subsection{Lie reductions of codimensions two and three}\label{sec:FPGenBetaLieRedCoD2}

The Lie reductions with respect to the subalgebras
$\mathfrak s_{2.1}^\delta$,
$\mathfrak s_{2.2}^{\varepsilon,0}$,
$\mathfrak s_{2.3}^\kappa$,
$\mathfrak s_{2.4}^\kappa$ and
$\mathfrak s_{3.1}^\kappa$
do not give new solutions
in comparison with the two-step Lie reductions,
where the first step is associated with the subalgebra~$\mathfrak s_{1.3}^0$
and leads to (1+1)-dimensional linear heat equations with zero or inverse square potentials
as reduced equations, and such equations are well studied
within the framework of symmetry analysis of differential equations.

This is why the only essential codimension-two Lie reductions
of the equations from the class~$\mathcal F'$ with $\beta\notin\{0,2,3,5\}$ are those
associated with the subalgebras
$\mathfrak s_{2.2}^{\varepsilon,\kappa}:=\langle\mathcal P^y+\varepsilon\mathcal I, \mathcal P^t+\kappa\mathcal I\rangle$,
where $\varepsilon\in\{-1, 1\}$ and $\kappa\ne0$.
A corresponding ansatz is $u=e^{\varepsilon y+\kappa t}\varphi(\omega)$ with $\omega=x$,
and the reduced equation for~$\varphi$ is
\[
|\omega|^\beta\varphi_{\omega\omega}=(\varepsilon\omega+\kappa)\varphi.
\]
In some particular cases of parameters, we can construct closed-form solutions of this equation.

For example, in the cases $\beta=1$ and $\beta=4$,
the counterparts of this equation for the corresponding equations of the form~\eqref{eq:FPWithRationalBeta}
coincide with the equations
\cite[Chapter~C, Section~2.273, Eq.~(12)]{kamk1977A}
with  $a=-\kappa$, $b^2=\varepsilon$, $c=1$ and $a=-\varepsilon$, $b^2=\kappa$, $c=-1$, respectively.
At the same time, in view of results of Section~\ref{sec:FPGenBetaLieRedCoD1} on
$\mathfrak s_{1.3}^\delta$-invariant solutions with $\delta\ne0$,
the consideration of the Lie reduction with respect to the subalgebras~$\mathfrak s_{2.2}^{\varepsilon,\kappa}$
makes sense only for $\beta\ne1,4$.

For the case $\beta=-1$,
the change of variable $\tilde\omega=\omega+\frac12\varepsilon\kappa$
results in the equation \cite[Chapter~C, Section~2.273, Eq.~(11)]{kamk1977A}
with $a=0$, $b=-\frac14\varepsilon\kappa^2$ and $c^2=-\varepsilon$,
whose general solution is
\begin{gather*}
\varphi=|\tilde\omega|^{-1/2}\mathop{\rm Re}\Big((C_1+{\rm i}C_2)W_{\frac{{\rm i}\kappa^2}{16},\frac14}({\rm i}\tilde\omega)\Big)
\quad\mbox{if}\quad\varepsilon=1,
\\
\varphi=|\tilde\omega|^{-1/2}\Big(
 C_1M_{\frac{\kappa^2}{16},\frac14}(\tilde\omega)
+C_2W_{\frac{\kappa^2}{16},\frac14}(\tilde\omega)\Big)
\quad\mbox{if}\quad\varepsilon=-1.
\end{gather*}
Here $M_{a,b}(z)$ and~$W_{a,b}(z)$ are the Whittaker functions,
which constitute the fundamental solution set of the Whittaker equation
\begin{gather*}
\varphi_{zz}+\left(-\frac14+\frac az+\frac{1/4-b^2}{z^2}\right)\varphi=0
\end{gather*}
with constant parameters~$a$ and~$b$.
The corresponding solutions of the equation~$\mathcal F'_{-1}$ are
\begin{gather}\label{eq:Beta-1RedCod2Solution1}
\solution\beta=-1\colon\
u=\Big|x+\frac\kappa2\Big|^{-1/2}{\rm e}^{y+\kappa t}\mathop{\rm Re}\Big(
(C_1+{\rm i}C_2)W_{\frac{{\rm i}\kappa^2}{16},\frac14}\big({\rm i}(x+\tfrac12\kappa)\big)
\Big),
\\ \label{eq:Beta-1RedCod2Solution2}
\solution\beta=-1\colon\
u=\Big|x-\frac\kappa2\Big|^{-1/2}{\rm e}^{-y+\kappa t}\Big(
 C_1M_{\frac{\kappa^2}{16},\frac14}\big(x-\tfrac12\kappa\big)
+C_2W_{\frac{\kappa^2}{16},\frac14}\big(x-\tfrac12\kappa\big)\Big).
\end{gather}

Using the transformation~\eqref{eq:EssEquivTransMod}, we can map these solutions to solution of the equation~$\mathcal F'_6$,
\begin{gather*}
\solution\beta=6\colon\
u=x\Big|x^{-1}+\frac\kappa2\Big|^{-1/2}{\rm e}^{t+\kappa y}\mathop{\rm Re}\Big(
(C_1+{\rm i}C_2)W_{\frac{{\rm i}\kappa^2}{16},\frac14}\big({\rm i}(x^{-1}+\tfrac12\kappa)\big)
\Big),
\\
\solution\beta=6\colon\
u=x\Big|x^{-1}-\frac\kappa2\Big|^{-1/2}{\rm e}^{-t+\kappa y}\Big(
 C_1M_{\frac{\kappa^2}{16},\frac14}\big(x^{-1}-\tfrac12\kappa\big)
+C_2W_{\frac{\kappa^2}{16},\frac14}\big(x^{-1}-\tfrac12\kappa\big)\Big).
\end{gather*}

\subsection{Generating solutions}\label{sec:FPGenBetaSolGeneration}

The simplest way to generate solutions for an equation~$\mathcal F'_\beta$ with $\beta\in\mathbb R\setminus\{0,2,3,5\}$
is to use its point symmetry pseudogroup from Theorem~\ref{thm:GenCaseSymGroup}.
More specifically, given a solution $u=h(t,x,y)$ of~$\mathcal F'_\beta$,
the action of a transformation of the form~\eqref{eq:GenCaseSymGroupBetaNe52} leads to the solution
\begin{gather*}
u=\sigma h\big(|\alpha|^{\beta-2}t-\lambda_0,\alpha^{-1}x,\alpha^{-1}|\alpha|^{\beta-2}y-\lambda_1\big)+f(t,x,y),
\end{gather*}
of the same equation.
In the case $\beta=5/2$, we can in addition use a transformation of the form~\eqref{eq:GenCaseSymGroupBeta=52}
and derive the solution
\begin{gather*}
u=\sigma h\big(\sgn(x)|\alpha|^{\beta-2}y-\lambda_0,\alpha x^{-1},\sgn(x)\alpha^{-1}|\alpha|^{\beta-2}t-\lambda_1\big)+f(t,x,y).
\end{gather*}

Since each equation $\mathcal F'_\beta$ is linear, we can also apply its generalized symmetries
to the solution generation following~\cite{kova2024a}.
The differential operators in total derivatives that are associated with the Lie-symmetry vector fields
$-\mathcal P^t$, $-\mathcal P^y$, $-\mathcal D^\beta$ and $\mathcal I$ are
\begin{gather*}
\mathrm P^t:=\mathrm D_t,\quad
\mathrm P^y:=\mathrm D_y,\quad
\mathrm D^\beta:=(2-\beta)t\mathrm D_t+x\mathrm D_x+(3-\beta)y\mathrm D_y
\end{gather*}
and the identity operator, respectively.
These differential operators generate the associative algebra $\Upsilon_{\mathfrak s}$
associated with the subalgebra $\mathfrak s=\langle\mathcal P^t,\mathcal P^y,\mathcal D^\beta\rangle\simeq A_{3.4}^a$
of the algebra $\mathfrak g_\beta$, where $a:=(\beta-2)/(\beta-3)$.
The operators $\mathrm P^t$ and $\mathrm D^\beta$ on the solutions of $\mathcal F'_\beta$
respectively coincide with the operators
\begin{gather*}
\hat{\mathrm P}^t:=-x\mathrm D_y+|x|^\beta\mathrm D_x^2
\quad\mbox{and}\quad
\hat{\mathrm D}^\beta:=(2-\beta)t\hat{\mathrm P}^t+x\mathrm D_x+(3-\beta)y\mathrm D_y.
\end{gather*}
Consider the associative algebra $\bar{\Upsilon}_{\mathfrak s}$ generated by the operators
$\hat{\mathrm P}^t$, $\mathrm P^y$ and $\hat{\mathrm D}^\beta$.
In other words, the algebra $\bar{\Upsilon}_{\mathfrak s}$ is subjected to the presentation
\begin{gather}\label{eq:UpsPresent}
\bar{\Upsilon}_{\mathfrak s}=
\big\langle\hat{\mathrm P}^t,\mathrm P^y,\hat{\mathrm D}^\beta\mid
[\hat{\mathrm P}^t,\hat{\mathrm D}^\beta]=(2-\beta)\hat{\mathrm P}^t,\,
[\mathrm P^y,\hat{\mathrm D}^\beta]=(3-\beta)\mathrm P^y,\,
[\hat{\mathrm P}^t,\mathrm P^y]=0
\big\rangle.
\end{gather}

\begin{lemma}
The monomials $\mathbf Q^\alpha:=(\mathrm Q^1)^{\alpha_1}(\mathrm Q^2)^{\alpha_2}(\mathrm Q^3)^{\alpha_3}$,
where $\alpha=(\alpha_1,\alpha_2,\alpha_3)\in\mathbb N_0^{\,\,3}$ and
$(\mathrm Q^1,\mathrm Q^2,\mathrm Q^3)$ is any fixed ordering of~$\hat{\mathrm P}^t$, $\mathrm P^y$ and $\hat{\mathrm D}^\beta$,
constitute a basis of the algebra $\bar{\Upsilon}_{\mathfrak s}$.
\end{lemma}

\begin{proof}
We again make use of Bergman's diamond lemma~\cite{berg1978a}.
Without loss of generality, we fix the ordering $\hat{\mathrm D}^\beta<\mathrm P^y<\hat{\mathrm P}^t$
and extend it to the degree lexicographic order on the words
in the alphabet $(\hat{\mathrm D}^\beta,\mathrm P^y,\hat{\mathrm P}^t)$.
According to the chosen ordering,
the relations in the presentation~\eqref{eq:UpsPresent} of the algebra $\bar{\Upsilon}_{\mathfrak s}$
can be written as the following reduction system:
\begin{gather*}
\hat{\mathrm P}^t\hat{\mathrm D}^\beta=\hat{\mathrm D}^\beta\hat{\mathrm P}^t+(2-\beta)\hat{\mathrm P}^t,\quad
\mathrm P^y\hat{\mathrm D}^\beta=\hat{\mathrm D}^\beta\mathrm P^y+(3-\beta)\mathrm P^y,\quad
\hat{\mathrm P}^t\mathrm P^y=\mathrm P^y\hat{\mathrm P}^t.
\end{gather*}
It has exactly one overlap ambiguity, $\hat{\mathrm P}^t\mathrm P^y\hat{\mathrm D}^\beta$,
which is resolvable since
\begin{gather*}
\begin{split}
\hat{\mathrm P}^t\mathrm P^y\hat{\mathrm D}^\beta
&{}=\mathrm P^y\hat{\mathrm P}^t\hat{\mathrm D}^\beta
=\mathrm P^y\hat{\mathrm D}^\beta\hat{\mathrm P}^t+(2-\beta)\mathrm P^y\hat{\mathrm P}^t
=\hat{\mathrm D}^\beta\mathrm P^y\hat{\mathrm P}^t+(3-\beta)\mathrm P^y\hat{\mathrm P}^t+(2-\beta)\mathrm P^y\hat{\mathrm P}^t,
\\[1ex]
\hat{\mathrm P}^t\mathrm P^y\hat{\mathrm D}^\beta
&{}=\hat{\mathrm P}^t\hat{\mathrm D}^\beta\mathrm P^y+(3-\beta)\hat{\mathrm P}^t\mathrm P^y
=\hat{\mathrm D}^\beta\hat{\mathrm P}^t\mathrm P^y+(2-\beta)\hat{\mathrm P}^t\mathrm P^y+(3-\beta)\hat{\mathrm P}^t\mathrm P^y
\\
&{}=\hat{\mathrm D}^\beta\mathrm P^y\hat{\mathrm P}^t+(2-\beta)\mathrm P^y\hat{\mathrm P}^t+(3-\beta)\mathrm P^y\hat{\mathrm P}^t.
\end{split}
\end{gather*}
Thus, the diamond lemma implies that the monomials of the form
$(\hat{\mathrm P}^t)^{\alpha_1}(\mathrm P^y)^{\alpha_2}(\hat{\mathrm D}^\beta)^{\alpha_3}$
constitute an $\mathbb R$-basis of the algebra $\bar{\Upsilon}_{\mathfrak s}$.
Similar arguments work for all the other orderings of $\hat{\mathrm P}^t$, $\mathrm P^y$ and~$\hat{\mathrm D}^\beta$.
\end{proof}

\begin{corollary}
The algebra $\bar{\Upsilon}_{\mathfrak s}$ is isomorphic to the universal enveloping algebra $\mathfrak U(A_{3.4}^a)$
of the Lie algebra $A_{3.4}^a=\langle e_1,e_2,e_3\rangle$, $[e_1,e_3]=e_1$, $[e_2,e_3]=ae_2$,
where $a=(2-\beta)/(3-\beta)$.
\end{corollary}

\begin{proof}
The correspondence $-\mathcal P^t\mapsto\hat{\mathrm P}^t$, $-\mathcal P^y\mapsto\mathrm P^y$ and $-\mathcal D^\beta\mapsto\hat{\mathrm D}^\beta$
by linearity extends to the Lie algebra homomorphism~$\varphi$ from~$\mathfrak s=\langle\mathcal P^t,\mathcal P^y,\mathcal D^\beta\rangle$
to the Lie algebra~\smash{$\bar{\Upsilon}_{\mathfrak s}^{(-)}$}
associated with the associative algebra~$\bar{\Upsilon}_{\mathfrak s}$,
$\varphi\colon\mathfrak s\to\smash{\bar{\Upsilon}_{\mathfrak s}^{(-)}}$.
By the universal property of the universal enveloping algebra~$\mathfrak U(\mathfrak s)$,
the Lie algebra homomorphism~$\varphi$ extends to the (unital) associative algebra homomorphism
$\hat\varphi\colon\mathfrak U(\mathfrak s)\to\bar{\Upsilon}_{\mathfrak s}$,
i.e., $\varphi=\hat\varphi\circ\iota$ as homomorphisms of vector spaces,
where $\iota\colon\mathfrak s\to\mathfrak U(\mathfrak s)$ is the canonical embedding of $\mathfrak s$ in $\mathfrak U(\mathfrak s)$.
Since the algebra $\Upsilon_{\mathfrak s}$ is generated by $\varphi(\mathfrak s)$,
the homomorphism $\hat\varphi$ is surjective.
Fixed an ordering of the basis elements of~$\mathfrak s$,
in view of the Poincar\'e--Birkhoff--Witt theorem, the homomorphism $\hat\varphi$ maps
the corresponding Poincar\'e--Birkhoff--Witt basis of $\mathfrak U(\mathfrak s)$ to a basis of $\bar{\Upsilon}_{\mathfrak s}$.
Therefore, $\hat\varphi$ is an isomorphism.
The isomorphism $\mathfrak s\simeq A_{3.4}^a$
implies the isomorphism $\mathfrak U(\mathfrak s)\simeq\mathfrak U(A_{3.4}^a)$.
\end{proof}

\begin{corollary}\label{cor:SualgOfGenSymAlg}
The algebra $\bar{\Upsilon}_{\mathfrak s}$ is isomorphic to
the algebra $\big\langle(\mathbf Q^\alpha u)\p_u\mid\alpha\in\mathbb N_0^{\,\,3}\big\rangle$
of canonical representatives
for a subalgebra of the quotient algebra of generalized symmetries
of an equation~$\mathcal F'_\beta$ with $\beta\in\mathbb R\setminus\{0,2,3,5\}$
with respect to the equivalence of generalized symmetries.
\end{corollary}

Similarly to what has been indicated in Sections~\ref{sec:RemarkableFPandCounterpart} and~\ref{sec:FineFPandCounterpart},
we use elements of the algebra $\bar{\Upsilon}_{\mathfrak s}$ to generate solutions of the equation~$\mathcal F_\beta'$.
More specifically, acting on a solution $u=h(t,x,y)$ of this equation by an arbitrary element
\[
Q=\sum_{(\alpha_1,\alpha_2,\alpha_3)\in\mathbb N_0^{\,\,3}}
c_{\alpha_1\alpha_2\alpha_3}(\mathrm Q^1)^{\alpha_1}(\mathrm Q^2)^{\alpha_2}(\mathrm Q^3)^{\alpha_3}
\]
of the algebra $\bar{\Upsilon}_{\mathfrak s}$, gives the solution $Qh$ of the same equation.
Here all but finitely many (real) constants $c_{\alpha_1\alpha_2\alpha_3}$ are equal to zero,
and any ordering in $\{\hat{\mathrm D}^\beta,\mathrm P^y,\hat{\mathrm P}^t\}$  can be fixed.
In view of Corollary~\ref{cor:SualgOfGenSymAlg}, the latter generation is consistent with the
generation of solutions using generalized symmetries as indicated in~\cite[Section~6]{kova2024a},
however, since the algebra of generalized symmetries of~$\mathcal F'_\beta$ is yet not described,
we use its known subalgebra.

Nevertheless, depending on~$h$ and~$c_{\alpha_1\alpha_2\alpha_3}$,
this formula can result in a known or even the zero solution.
We analyze which solutions generated according to this
from Lie invariant solutions of~$\mathcal F'_\beta$ may be of interest.
It suffices to consider only the solutions that are invariant with respect to
the subalgebras
$\mathfrak s_{1.1}^\kappa=\langle\mathcal D^\beta+\kappa\mathcal I\rangle$,
$\mathfrak s_{1.2}^{\varepsilon,\kappa}=\langle\mathcal P^y+\varepsilon\mathcal P^t+\kappa\mathcal I\rangle$,
$\mathfrak s_{1.3}^\delta=\langle\mathcal P^y+\delta\mathcal I\rangle$,
$\mathfrak s_{2.2}^{\varepsilon,\kappa}=\langle\mathcal P^y+\varepsilon\mathcal I, \mathcal P^t+\kappa\mathcal I\rangle$,
with $\varepsilon=\pm1$, $\delta\in\{-1, 0, 1\}$, $\kappa\in\mathbb R$
and, for the last subalgebra family, $\kappa\ne0$,
see Section~\ref{sec:SubalgGbeta}.
These solutions were discussed in Sections~\ref{sec:FPGenBetaLieRedCoD1} and~\ref{sec:FPGenBetaLieRedCoD2}.
In view of the structure of the algebra~$\mathfrak g_\beta$ with $\beta\in\mathbb R\setminus\{0,2,3,5\}$,
the analysis reduces to that for single monomials
$(\hat{\mathrm D}^\beta)^{\alpha_1}(\mathrm P^y)^{\alpha_2}(\hat{\mathrm P}^t)^{\alpha_3}$.

\medskip\par\noindent
1.1. The vector field $\mathcal D^\beta+\kappa\mathcal I$ corresponds to the differential operator $\hat{\mathrm D}^\beta-\kappa$.
Since $\hat{\mathrm P}^t\hat{\mathrm D}^\beta=(\hat{\mathrm D}^\beta+2-\beta)\hat{\mathrm P}^t$ and
$\mathrm P^y\hat{\mathrm D}^\beta=(\hat{\mathrm D}^\beta+3-\beta)\mathrm P^y$, we have
\begin{gather*}
(\hat{\mathrm D}^\beta)^{\alpha_3}(\hat{\mathrm P}^t)^{\alpha_1}(\mathrm P^y)^{\alpha_2}(\hat{\mathrm D}^\beta-\kappa)
=(\hat{\mathrm D}^\beta+\alpha_1(2-\beta)+\alpha_2(3-\beta)-\kappa)
(\hat{\mathrm D}^\beta)^{\alpha_3}(\hat{\mathrm P}^t)^{\alpha_1}(\mathrm P^y)^{\alpha_2}.
\end{gather*}
In other words, the monomial $(\hat{\mathrm P}^t)^{\alpha_1}(\mathrm P^y)^{\alpha_2}(\hat{\mathrm D}^\beta)^{\alpha_3}$
maps any $\mathfrak s_{1.1}^\kappa$-invariant solution to $\mathfrak s_{1.1}^{\tilde\kappa}$-invariant solution
with $\tilde\kappa=\kappa-\alpha_1(2-\beta)-\alpha_2(3-\beta)$.
Thus, any element of $\bar{\Upsilon}_{\mathfrak s}$ maps the space spanned by $\mathfrak s_{1.1}^\kappa$-invariant solutions
with $\kappa$ running through~$\mathbb R$ into itself.

\medskip\par\noindent
1.2. It is clear that any polynomial in $\hat{\mathrm P}^t$ and $\mathrm P^y$ commutes with
$Q:=\mathrm P^y+\varepsilon\hat{\mathrm P}^t-\kappa$.
Therefore, its action maps the set of $\mathfrak s_{1.2}^\kappa$-invariant solutions into itself.
At the same time, the operator $\hat{\mathrm D}^\beta$ does not commute with~$Q$,
$
Q\hat{\mathrm D}^\beta=(\hat{\mathrm D}^\beta+2-\beta)Q+\mathrm P^y+(2-\beta)\kappa.
$
Thus, for an $\mathfrak s_{1.2}^{\varepsilon,\kappa}$-invariant solution~$h$,
the solution $\hat{\mathrm D}^\beta h$ is $\mathfrak s_{1.2}^{\varepsilon,\kappa}$-invariant
if and only if the solution $h$ is in addition $\langle\mathcal P^y+(\beta-2)\kappa\mathcal I\rangle$-invariant.
This is why to generate new solutions of~$\mathcal F'_\beta$ from $\mathfrak s_{1.2}^{\varepsilon,\kappa}$- invariant solutions
within the framework discussed, up to linearly combining solutions,
it suffices to iteratively act by the operator~$\hat{\mathrm D}^\beta$
on other $\mathfrak s_{1.2}^{\varepsilon,\kappa}$-invariant solutions if they are known.

\medskip\par\noindent
1.3. It is clear that acting by a polynomial in $\mathrm P^y$ or $\hat{\mathrm P}^t$
on $\mathfrak s_{1.3}^\delta$-invariant solutions does not result in new solutions
since it commutes with the operator $\mathrm P^y-\delta$.
To describe the action of powers of~$\hat{\mathrm D}^\beta$, we split the consideration into two cases.

Assume that $\delta=0$.
Despite the fact that $\mathrm P^y$ and $\hat{\mathrm D}^\beta$ do not commute,
the action by $\hat{\mathrm D}^\beta$ does not give us new solutions since
$\mathrm P^y\hat{\mathrm D}^\beta=(\hat{\mathrm D}^\beta+3-\beta)\mathrm P^y$.
Therefore, any element of $\bar{\Upsilon}_{\mathfrak s}$ maps
the space of $\mathfrak s_{1.3}^0$-invariant solutions of~\smash{$\mathcal F'_\beta$} into itself.

In the case $\delta\ne0$,
the equality $(\mathrm P^y+\delta)\hat{\mathrm D}^\beta=(\hat{\mathrm D}^\beta+3-\beta)(\mathrm P^y+\delta)-\delta(3-\beta)$
implies that for any $\mathfrak s_{1.3}^\delta$-invariant solution $h$,
the solution $\hat{\mathrm D}^\beta h$ is not $\mathfrak s_{1.3}^\delta$-invariant.
This is why to generate new solutions of~$\mathcal F'_\beta$ from $\mathfrak s_{1.3}^\delta$-invariant solutions
within the framework discussed, up to linearly combining solutions,
it again suffices to iteratively act by the operator~\smash{$\hat{\mathrm D}^\beta$}.
Since for the specific case of $\beta=1$, we constructed
the families~\eqref{eq:Beta1ExactSol2} and~\eqref{eq:Beta1ExactSol3} of
$\mathfrak s_{1.3}^\delta$-invariant solutions with $\delta\ne0$,
we can extend them to more new wide families of solutions of the equation~\smash{$\mathcal F'_\beta$},
\begin{gather*}
\solution\beta=1\colon\
u=(\hat{\mathrm D}^1)^k\left(|x|^{1/4}{\rm e}^{\delta y}\frac{{\rm e}^{|x|\tan t}}{\sqrt{|\cos t|}}
\vartheta^{-3/4}\left(\frac{\tan t}2,\frac{\sqrt{2|x|}}{\cos t}\right)\right),
\\
\solution\beta=1\colon\
u=(\hat{\mathrm D}^1)^k\left(
|x|^{1/4}{\rm e}^{x+\frac12t+\delta y}\vartheta^{-3/4}\left(\frac14{\rm e}^{2t},\sqrt{2|x|}{\rm e}^t\right)\right),
\end{gather*}
where $\hat{\mathrm D}^1:
=x(t\mathrm D_x+1)\mathrm D_x+(2y-tx)\mathrm D_y$.
These solution families can be mapped by the transformation~\eqref{eq:EssEquivTransMod}
to new solution families  of the equation~\smash{$\mathcal F'_4$},
which can be rearranged up to linearly combining.
An equivalent but simpler way is to act on solutions of the form~\eqref{eq:Beta4ExactSol2} and~\eqref{eq:Beta4ExactSol3}
by the operator $\hat{\mathrm D}^4:
=-2tx^4\mathrm D_x^2+x\mathrm D_x+(2tx-y)\mathrm D_y$,
which is associated with the Lie-symmetry vector field~$-\mathcal D^4$ of~$\mathcal F'_4$
and thus is the counterpart of $\hat{\mathrm D}^1$,
\begin{gather*}
\solution\beta=4\colon\
u=(\hat{\mathrm D}^4)^k\left(|x|^{3/4}{\rm e}^{\delta t}\frac{{\rm e}^{|x|^{-1}\tan y}}{\sqrt{|\cos y|}}
\vartheta^{-3/4}\left(\frac{\tan y}2,\frac{\sqrt{2|x|^{-1}}}{\cos y}\right)\right),
\\
\solution\beta=4\colon\
u=(\hat{\mathrm D}^4)^k\left(|x|^{3/4}{\rm e}^{x^{-1}+\frac12y+\delta t}\vartheta^{-3/4}\left(\frac14{\rm e}^{2y},\sqrt{2|x|^{-1}}{\rm e}^y\right)\right),
\end{gather*}
Note that $\mathcal J'_*\hat{\mathrm D}^1=-\hat{\mathrm D}^4+1$.

\medskip\par\noindent
2.2. Similarly to case 1.3,
the operators $\mathrm P^y-\varepsilon$ and $\hat{\mathrm P}^t-\kappa$
commute with the operators $\hat{\mathrm P}^t$ and $\mathrm P^y$
and permute with the operator $\hat{\mathrm D}^\beta$ in the following way:
\begin{gather*}
(\mathrm P^y-\varepsilon)\hat{\mathrm D}^\beta
=(\hat{\mathrm D}^\beta+3-\beta)(\mathrm P^y-\varepsilon)+(3-\beta)\varepsilon,
\\
(\hat{\mathrm P}^t-\kappa)\hat{\mathrm D}^\beta
=(\hat{\mathrm D}^\beta+2-\beta)(\hat{\mathrm P}^t-\kappa)+(2-\beta)\kappa,
\end{gather*}
therefore,
for any $\mathfrak s_{2.2}^{\varepsilon,\kappa}$-invariant solution $h$,
the solution $\hat{\mathrm D}^\beta h$ is not $\mathfrak s_{2.2}^{\varepsilon,\kappa}$-invariant.
Therefore, to construct essentially new solutions of $\mathcal F'_\beta$,
it suffices, up to linearly combining solutions,
to act by the powers of the operator \smash{$\hat{\mathrm D}^\beta$}.

In Section~\ref{sec:FPGenBetaLieRedCoD2}, we have constructed
the solution families~\eqref{eq:Beta-1RedCod2Solution1} and~\eqref{eq:Beta-1RedCod2Solution2}
for the specific case $\beta=-1$.
Analogously to generating solutions in case~1.3 with $\delta\ne0$,
we extend these solutions to more new solution families of the equation~$\mathcal F'_{-1}$
and contract their counterparts for the equation~$\mathcal F'_6$,
\begin{gather*}
\solution\beta=-1\colon\
u=(\hat{\mathrm D}^{-1})^k\mathop{\rm Re}\left(
\Big|x+\frac\kappa2\Big|^{-1/2}e^{y+\kappa t}
(C_1+{\rm i}C_2)W_{\frac{{\rm i}\kappa^2}{16},\frac14}\Big({\rm i}\big(x+\tfrac12\kappa\big)\Big)\right),
\\[1ex]
\solution\beta=-1\colon\
u=(\hat{\mathrm D}^{-1})^k\left(
\Big|x-\frac\kappa2\Big|^{-1/2}e^{-y+\kappa t}\Big(
 C_1M_{\frac{\kappa^2}{16},\frac14}\big(x-\tfrac12\kappa\big)
+C_2W_{\frac{\kappa^2}{16},\frac14}\big(x-\tfrac12\kappa\big)\Big)\right),
\\[1ex]
\solution\beta=6\colon\
u=(\hat{\mathrm D}^6)^k\mathop{\rm Re}\left(
x\Big|x^{-1}+\frac\kappa2\Big|^{-1/2}e^{t+\kappa y}
(C_1+{\rm i}C_2)W_{\frac{{\rm i}\kappa^2}{16},\frac14}\Big({\rm i}\big(x^{-1}+\tfrac12\kappa\big)\Big)\right),
\\[1ex]
\solution\beta=6\colon\
u=(\hat{\mathrm D}^6)^k\left(x\Big|x^{-1}-\frac\kappa2\Big|^{-1/2}e^{\kappa y-t}\Big(
 C_1M_{\frac{\kappa^2}{16},\frac14}\big(x^{-1}-\tfrac12\kappa\big)
+C_2W_{\frac{\kappa^2}{16},\frac14}\big(x^{-1}-\tfrac12\kappa\big)\Big)
\right),
\end{gather*}
where $\hat{\mathrm D}^{-1}
:=3tx^{-1}\mathrm D_x^2+x\mathrm D_x+(4y-3tx)\mathrm D_y$
and $\hat{\mathrm D}^6
:=-4tx^6\mathrm D_x^2+x\mathrm D_x+(4tx-3y)\mathrm D_y$.

\section{Conclusion}\label{sec:Conclusion}

The class of ultraparabolic (1+2)-dimensional Kolmogorov backward equations with power diffusivity
$\mathcal F$ and its properly gauged counterpart~$\mathcal F'$
have many outstanding properties
from the perspective of symmetry analysis of differential equations.
As expected, the generic elements of the class~$\mathcal F'$ are related to one of the inequivalent cases of Lie symmetry extension
within the superclass~$\bar{\mathcal F}$ of the ultraparabolic linear (1+2)-dimensional second-order partial differential equations,
which is a subcase of a case with at least four-dimensional decomposable solvable essential Lie invariance algebras.
Moreover, specific exponents 0, 2, 3 and~5 of power diffusivity correspond to additional Lie symmetry extensions.

At the same time, the number of surprising properties of the classes~$\mathcal F$ and~$\mathcal F'$
is essentially greater than those we expected.

The class~$\mathcal F'$ admits the nontrivial discrete equivalence transformation~\eqref{eq:EssEquivTransOfF'},
which is the only essential among elements of the equivalence group~$G^\sim_{\mathcal F'}$
in the course of group classification of this class.
Under passing to the class~$\mathcal F$,
this transformation induces elements from the complement of the action groupoid of the group~$G^\sim_{\mathcal F}$
in the groupoid~$\mathcal G^\sim_{\mathcal F}$, and thus it is associated with no element of~$G^\sim_{\mathcal F}$.
This shows the naturality and, moreover, the necessity of the gauge~$\alpha=0$.

There are only two $G^\sim_{\mathcal F'}$-inequivalent cases of essential Lie symmetry extensions within the class~$\mathcal F'$.
They are associated with
the remarkable Fokker--Planck equation~$\mathcal F'_0$ and
the fine Kolmogorov backward equation~$\mathcal F'_2$.
Furthermore, these are equations that are singular with respect to their symmetry properties
in the entire superclass~$\bar{\mathcal F}$,
representing the cases with maximum eight-dimensional and nonsolvable five-dimensional
essential Lie symmetry extensions, respectively.
As a result, the study of~$\mathcal F'_0$ and~$\mathcal F'_2$ in \cite{kova2023a,kova2024a,popo2024b}
provides a significant part of symmetry analysis of elements of not only the classes~$\mathcal F$ and~$\mathcal F'$
but also their superclass~$\bar{\mathcal F}$.

The equations~$\mathcal F'_0$ and~$\mathcal F'_2$ are also distinguished by other symmetry properties.
They admit a number of nontrivial hidden Lie symmetries that are associated with their codimension-one Lie reductions
to the (1+1)-dimensional linear heat equations with inverse square or zero potentials.
Moreover, all the closed-form Lie invariant solutions that can be constructed for~$\mathcal F'_0$ and~$\mathcal F'_2$
are related to such reductions.
These solutions constitute wide families and, furthermore, can be used as seeds for generating new solutions
via acting by recursion operators related to Lie symmetries,
and such generation indeed efficiently works for~$\mathcal F'_0$ and~$\mathcal F'_2$.

Even more surprising is that,
among the other equations in the class~$\mathcal F'$ up to the $G^\sim_{\mathcal F'}$-equivalence,
only the equation~$\mathcal F'_1$ possesses all the above properties
of~$\mathcal F'_0$ and~$\mathcal F'_2$, except the extension of essential Lie invariance algebras.
This is why the singularity of~$\mathcal F'_1$ within the class~$\mathcal F'$ can be found out only after
the comprehensive study of Lie reductions of the equations from this class.

Another but less singular equation is $\mathcal F'_{-1}$.
Its singularity becomes apparent only in the course of carrying out its codimension-two Lie reductions
to ordinary differential equations that integrate,
in contrast to their analogues for regular elements of the class~$\mathcal F'$,
in terms of Whittaker functions.
The solutions of $\mathcal F'_{-1}$ constructed in this way
can also be used as seeds for generating new solutions
via acting by a recursion operator arising from a Lie symmetry.

We have found hidden Lie symmetries for the generic equations from the class~$\mathcal F'$,
but their number is essentially less than for the singular equations from these class
with $\beta\in\{0,1,2,3,4,5\}$.
All of these symmetries  are associated with codimension-one Lie reductions.

In total, we have singled out eight values of the arbitrary element~$\beta$,
including $G^\sim_{\mathcal F'}$-equivalent counterparts, $-1$, 0, 1, 2, 3, 4, 5 and~6.
An open question is whether there exist other properties that distinguish the listed values.
It is also not clear whether the study of more complicated structures like generalized symmetries or reduction modules
can allow one to reveal the singularity of some values of~$\beta$
that are regular from the point of view of Lie and hidden Lie symmetries,
the construction of families of Lie-invariant solutions and their extension by recursion operators.

The results of the present paper has several novel features.

We not only have completely solved the properly posed group classification problems
for the classes~$\mathcal F'$ and~$\mathcal F$ with respect to equivalences of two kinds
but have also computed the point symmetry pseudogroups of all equations from these classes.
The point symmetry pseudogroups of the equations $\mathcal F'_\beta$
with $\beta=0$, $\beta=5$, $\beta=2$, $\beta=3$ and $\beta\in\mathbb R\setminus\{0,2,3,5\}$
are given in Theorems~\ref{thm:RemarkableFPSymGroup}, \ref{thm:Power5FPSymGroup},
\ref{thm:FineFPSymGroup}, \ref{thm:Power3FPSymGroup} and~\ref{thm:GenCaseSymGroup}, respectively,
which provides the exhaustive descriptions of the fundamental groupoids~$\mathcal G^{\rm f}_{\mathcal F'}$
and, hence,~$\mathcal G^{\rm f}_{\mathcal F}$.
Combining the above result on~$\mathcal G^{\rm f}_{\mathcal F'}$ with Theorem~\ref{thm:EquivGroupoidsFF'Properties}(iv),
we obtain the complete description of the entire equivalence groupoid $\mathcal G^\sim_{\mathcal F'}$.

\begin{theorem}\label{thm:GroupoidF'}
The class $\mathcal F'$ is semi-normalized.
Its equivalence groupoid~$\mathcal G^\sim_{\mathcal F'}$
is the Frobenius product of the action groupoid of the subgroup
generated by the discrete equivalence transformation~$\mathscr J'$
and the fundamental groupoid
$\mathcal G^{\rm f}_{\mathcal F'}=\{(\beta,\Phi,\beta)\mid \Phi\in G_\beta,\,\beta\in\mathbb R\}$,
where $G_\beta$ is the point-symmetry (pseudo)group of the equation~$\mathcal F'_\beta$,
presented in Theorems~\ref{thm:RemarkableFPSymGroup},
\ref{thm:Power5FPSymGroup},
\ref{thm:FineFPSymGroup},
\ref{thm:Power3FPSymGroup}
and~\ref{thm:GenCaseSymGroup}
for $\beta=0$, $\beta=5$, $\beta=2$, $\beta=3$ and $\beta\in\mathbb R\setminus\{0,2,3,5\}$,
respectively.
\end{theorem}

Corollary~\ref{cor:EquivGroupoidsFF'Properties} combined with Theorem~\ref{thm:GroupoidF'}
completes the description of $\mathcal G^\sim_{\mathcal F}$.

\begin{theorem}
The equivalence groupoid~$\mathcal G^\sim_{\mathcal F}$ is constituted by the admissible transformations
\begin{gather*}
\big((\alpha,\beta),\pi_*\mathscr S(\tilde\alpha)\circ\Phi\circ\pi_*\mathscr S(-\alpha),(\tilde\alpha,\tilde\beta)\big),
\end{gather*}
where $(\beta,\Phi,\tilde\beta)\in\mathcal G^\sim_{\mathcal F'}$.
\end{theorem}

We have first applied Hydon's automorphism-based version of the algebraic method
\cite{hydo1998a,hydo1998b,hydo2000b,hydo2000A}
to compute equivalence groups of classes of differential equations,
which are the groups~$G^\sim_{\mathcal F}$ and~$G^\sim_{\mathcal F'}$.
At the same time, the required automorphism groups
of the associated equivalence algebras~$\mathfrak g^\sim_{\mathcal F}$ and~$\mathfrak g^\sim_{\mathcal F'}$
cannot be found in the straightforward ``brute-force'' way.
This is why we needed first to construct a sufficient number of megaideals of these algebras
and then to take into account the derived constraints on entries of the automorphism matrices
in the course of the regular computation of the automorphism groups.
It was unexpected that
the megaideals preliminarily constructed using their basic properties exhaust
the entire sets of megaideals of the corresponding algebras.

The groups~$G^\sim_{\mathcal F}$ and~$G^\sim_{\mathcal F'}$ have also been constructed
using the original advanced version of the direct method.
There are two outputs of the double computation.
It provides a cross-check of the correctness of the obtained groups,
which is important
since knowing them underlies all aspects of symmetry analysis of equations from the classes~$\mathcal F$ and~$\mathcal F'$.
Due to the double computation, we can also compare the effectiveness of the applied method.
We did not expect that the advanced version of the direct method turns out to be more efficient than
the algebraic method.

In general, the study of the particular equivalence groupoids and equivalence groups in the present paper
gives a justified insight into further development of computational tools for such studies.

To find the point symmetry pseudogroup~$G_\beta$ of a generic equation~$\mathcal F'_\beta$ in the class~$\mathcal F'$,
we have used an original combination of the above methods, see the proof of Theorem~\ref{thm:GenCaseSymGroup}.
We have first applied the advanced version of the direct method to prove
the factorization of the pseudogroup~$G_\beta$ and the maximal Lie invariance algebra~$\mathfrak g_\beta$ of~$\mathcal F'_\beta$
into their essential parts and those associated with the linear superposition of solutions.
The later parts are known and can be neglected.
The fact that the adjoint action of the group~$G^{\rm ess}_\beta$ preserves the algebra~$\mathfrak g^{\rm ess}_\beta$
has allowed us to use techniques of Hydon's algebraic method
for obtaining a restrictive form for elements of the group~$G^{\rm ess}_\beta$,
which is parameterized by six constants.
Then we complete the computation of this group by the direct method,
deriving two constraints for these parameters.

The classes~$\mathcal F$ and~$\mathcal F'$ are constituted by linear differential equations.
This is why it was natural and instructive to generate new wide solution families of equations from these classes
by acting with recursion operators associated with their essential Lie symmetries.
We have singled out, up to the $G^\sim_{\mathcal F}$- and~$G^\sim_{\mathcal F'}$-equivalences
within the respective classes
and up to the $G_\beta$-equivalence for each particular value of~$\beta$,
all nontrivial generations among the above ones and then carried out them.
This provided the first example in the literature
on exhaustive study of such generations for a class of linear differential equations.
The first similar examples for single differential linear equations
were presented in~\cite{kova2024a} and~\cite{popo2024b} for the fine Kolmogorov backward equation~$\mathcal F'_2$
and the remarkable Fokker--Planck equation~$\mathcal F'_0$, respectively.

Despite the diverse results obtained for the classes~$\mathcal F$ and~$\mathcal F'$ and their particular elements
in~\cite{kova2023a,kova2024a,popo2024b} and in the present paper,
there are still many open problems on these objects, not to mention the superclass~$\bar{\mathcal F}$.

We have comprehensively studied the generation of solutions of equations from~$\mathcal F'$
by their recursion operators that are related to their Lie symmetries.
General linear recursion operators, which are associated with linear genuinely generalized symmetries,
are also relevant for such generation
but their classification is much more complicated than the classification of recursion operators stemming from Lie symmetries.
We proved in~\cite{popo2024b} that the entire algebra of generalized symmetries
of the remarkable Fokker--Planck equation~$\mathcal F'_0$ is generated
with successively acting by its Lie-symmetry operators on the trivial Lie symmetry $u\p_u$.
We also reasonably conjecture the analogous claim for the fine Kolmogorov backward equation~$\mathcal F'_2$ in \cite[Section~6]{kova2024a}.
At the same time, we have no conjecture on generalized symmetries of equations $\mathcal F'_\beta$
with $\beta\in\mathbb R\setminus\{0,2,3,5\}$.

In the context of the theory of singular reduction modules~\cite{boyk2016a}
and in view of the ultraparabolicity of equations from the classes~$\mathcal F$ and~$\mathcal F'$,
it will be instructive to study reduction modules of these equations.
In comparison with Lie symmetries, the main obstacle of finding non-Lie reduction modules is that
it reduces to solving nonlinear determining equations even for linear differential equations,
which makes the related computations highly nontrivial and cumbersome.
For example, even for the simplest class constituted by
linear second-order partial differential equations with two independent variables,
an analogous study was carried out only for the subclass of evolution equations~\cite{popo2008b}.

Local conservation laws of equations from the class~$\mathcal F$ were described in~\cite{zhan2020a},
which can be viewed as a preliminary step for the further consideration of nonlocal structures
associated with these equations, e.g., potential conservation laws and potential symmetries.
With this objective in mind, one should remember the principal difference
between the two- and multi-dimensional cases in this aspect~\cite[Theorem~2.7]{anco1997b}.
More specifically, for the existence of nontrivial potential structures in the latter case,
one should gauge potentials arising from the conservation laws,
and the selection of suitable and natural gauges can be still carried out only in a heuristic way.

All the above problems can be extended to equations from the superclass~$\bar{\mathcal F}$.
This class has been extensively studied for the past couple of decades.
Its partial preliminary group classification was carried out in~\cite{davy2015a}.
Some subclasses of~$\bar{\mathcal F}$ were considered within the Lie-symmetry framework
in \cite{gung2018b,kova2013a,spic2011a,spic1997a,spic1999a,zhan2020a}.
Despite the number of papers on this subject, there are still many open problems
in the symmetry analysis of the entire class~$\bar{\mathcal F}$,
its subclasses and even particular equations from this class,
as can be seen from the above discussion on the classes~$\mathcal F$ and~$\mathcal F'$
and on the equations $\mathcal F'_0$ and $\mathcal F'_2$.
We are close to completely solve the group classification problem for the class~$\bar{\mathcal F}$
using the algebraic approach,
and this involves almost all recently developed advanced techniques of group classification
and the famous description of finite-dimensional Lie algebras of vector fields on the real plane,
see~\cite[pp.~122--178]{lie1893A} for the first presentation and
\cite{gonz1992b} for the modern and enhanced version of this result.

\appendix

\section{Algebraic method for computing equivalence groups}\label{sec:AlgMet}

The straightforward construction of the equivalence group of a class of differential equations
by the direct method in general leads to cumbersome and complicated computations.
We simplify finding the equivalence groups~$G^\sim_{\mathcal F}$ and~$G^\sim_{\mathcal F'}$
due to using an advanced version of the direct method
and knowing the equivalence groupoid of the superclass~$\bar{\mathcal F}$,
which is described in Theorem~\ref{thm:EquivalenceGroupFPsuperClass}.
Nevertheless, it is still instructive
to find the groups~$G^\sim_{\mathcal F}$ and~$G^\sim_{\mathcal F'}$ by Hydon's \emph{algebraic method}
\cite{hydo1998a,hydo1998b,hydo2000b,hydo2000A} extended to equivalence groups in~\cite{bihl2015a},
which in addition leads to the double-check of Theorems~\ref{thm:EqGroupFPsubclass} and~\ref{thm:EquivGroupClassF'}.
The algebraic method is based on knowing the corresponding equivalence algebra,
which should be computed independently using the infinitesimal approach.
In the course of this computation, we can assume
that the constant arbitrary elements~$\alpha$ and~$\beta$ of the class~$\mathcal F$
(resp.\ $\beta$ of the class~$\mathcal F'$)
play the role of additional dependent variables satisfying the equations
$\alpha_t=\alpha_x=\alpha_y=0$ and $\beta_t=\beta_x=\beta_y=0$ (resp.\ $\beta_t=\beta_x=\beta_y=0$).
The corresponding infinitesimal invariance condition should be split
with respect to both the parametric jet variables and the arbitrary elements~$\alpha$ and~$\beta$
(resp.\ the arbitrary element~$\beta$).
Since the equivalence algebras~$\mathfrak g^\sim_{\mathcal F}$ and~$\mathfrak g^\sim_{\mathcal F'}$ are finite-dimensional,
we proceed with the original automorphism-based version of the algebraic method by Hydon
\cite{hydo1998a,hydo1998b,hydo2000b,hydo2000A}.
In fact, this is the first example of such a computation in the literature,
see~\cite[Section~4]{bihl2015a} for the first application of the megaideal-based version
of the same method to finding the equivalence group of a class of differential equations.

\subsection{The original class}\label{sec:AlgMetClassF}

Despite the fact that the dimension of the algebra~$\mathfrak g:=\mathfrak g^\sim_{\mathcal F}$ is not high,
see Corollary~\ref{cor:EquivAlgF},
the direct construction of the automorphism group ${\rm Aut}(\mathfrak g)$
using symbolic computation programs seems expensive, if at all possible.
To make the construction feasible, we first find a sufficient number of \emph{megaideals}~\cite{popo2003a} of~$\mathfrak g$,
i.e., linear subspaces of~$\mathfrak g$ that are stable with respect to the action of~${\rm Aut}(\mathfrak g)$;
another name for megaideals is \emph{fully characteristic ideals} \cite[Exercise~14.1.1]{hilg2011A}.
The choice of a basis of~$\mathfrak g$ agreed with the megaideal hierarchy~$\mathfrak g$
guarantees the presence of an essential number of zero entries in the matrix of a general automorphism of~$\mathfrak g$.
The more megaideals we construct, the simpler computations become.

In view of elementary properties of megaideals,
all elements of the derived series, the upper and the lower central series of~$\mathfrak g$,
including the center~$\mathfrak z(\mathfrak g)$ and the derivative~$\mathfrak g'$ of~$\mathfrak g$,
are megaideals of~$\mathfrak g$.
Thus,
\begin{gather*}
\mathfrak m_1:=\mathfrak g''=\langle\p_u\rangle,\ \
\mathfrak m_2:=\mathfrak z(\mathfrak g')=\langle\p_u,x\p_u\rangle,\ \
\mathfrak m_4:=\mathfrak g^3_{\mathcal F}=[\mathfrak g,\mathfrak g'_{\mathcal F}]=\langle\p_u,x\p_u,(tx-y)\p_u\rangle
\end{gather*}
are megaideals of $\mathfrak g$.
A useful tool for constructing megaideals of Lie algebras is~\cite[Proposition~1]{card2013a},
which states that if $\mathfrak i_0$, $\mathfrak i_1$ and $\mathfrak i_2$ are the megaideals of~$\mathfrak g$,
then the set~$\mathfrak s$ of elements from~$\mathfrak i_0$ whose commutators with
arbitrary elements from~$\mathfrak i_1$ belong to~$\mathfrak i_2$ is also a megaideal of~$\mathfrak g$.
By considering  $(\mathfrak i_0,\mathfrak i_1,\mathfrak i_2)=(\mathfrak g,\mathfrak g,\mathfrak m_1)$
and $(\mathfrak i_0,\mathfrak i_1,\mathfrak i_2)=(\mathfrak g,\mathfrak g,\mathfrak m_3)$
we respectively obtain the megaideals
\begin{gather*}
\mathfrak m_3:=\langle\p_u,\p_y\rangle,\quad
\mathfrak m_5:=\langle\p_u,\p_y,\p_x+t\p_y+\p_\alpha\rangle.
\end{gather*}
The sum of megaideals is a megaideal.
Taking $(\mathfrak i_0,\mathfrak i_1,\mathfrak i_2)=(\mathfrak m,\mathfrak g,\{0\})$,
we derive that the centralizer $\mathrm C_{\mathfrak g}(\mathfrak m)$ of a megaideal~$\mathfrak m$ is a megaideal.
By this, we obtain the megaideal
\begin{gather*}
\mathfrak m_6:=\mathrm C_{\mathfrak g}(\mathfrak m_2+\mathfrak m_3)=\langle\p_u,x\p_u,\p_y,\p_t\rangle.
\end{gather*}
Applying the technique suggested in~\cite[Lemma~10]{boyk2022a},
we construct one more megaideal in $\mathfrak g$.

\begin{proposition}
The span $\mathfrak m_7:=\langle\p_u,x\p_u,(tx-y)\p_u,u\p_u\rangle$ is a megaideal.
\end{proposition}

\begin{proof}
Consider the megaideal $\mathfrak s=\langle\p_u,x\p_u,(tx-y)\p_u,u\p_u,\p_y\rangle$ of $\mathfrak g$
obtained using~\cite[Proposition~1]{card2013a} for
$(\mathfrak i_0,\mathfrak i_1,\mathfrak i_2)=(\mathfrak g,\mathfrak g,\mathfrak m_4)$.
The megaideals~$\mathfrak m_2$, $\mathfrak m_4$ and~$\mathfrak m_3+\mathfrak m_4$ of $\mathfrak g$
are contained in~$\mathfrak s$.
Fix an arbitrary element $\Phi$ of~${\rm Aut}(\mathfrak g)$.
Since
$\Phi(\mathfrak m_2)=\mathfrak m_2$,
$\Phi(\mathfrak m_4)=\mathfrak m_4$,
$\Phi(\mathfrak m_3+\mathfrak m_4)=\mathfrak m_3+\mathfrak m_4$,
$\Phi(\mathfrak s)=\mathfrak s$,
$(tx-y)\p_u\in\mathfrak m_4\setminus\mathfrak m_2$ and
$u\p_u\in\mathfrak s\setminus(\mathfrak m_3+\mathfrak m_4)$,
we have~that
\begin{gather*}
\Phi((tx-y)\p_u)=\hat c_1\p_u+\hat c_2x\p_u+\hat c_3(tx-y)\p_u,
\\
\Phi(u\p_u)=c_1\p_u+c_2x\p_u+c_3(tx-y)\p_u+c_4u\p_u+c_5\p_y
\end{gather*}
with $c_4\hat c_3\ne0$.
Therefore,
\begin{gather*}
\Phi([(tx-y)\p_u,u\p_u])=\Phi((tx-y)\p_u)=\hat c_1\p_u+\hat c_2x\p_u+\hat c_3(tx-y)\p_u
\\
\qquad=\big[\Phi((tx-y)\p_u),\Phi(u\p_u)\big]=c_4(\hat c_1\p_u+\hat c_2x\p_u+\hat c_3(tx-y)\p_u)+c_5\hat c_3\p_u,
\end{gather*}
which gives $c_4\hat c_3=\hat c_3$ and $\hat c_1=c_4\hat c_1+c_5\hat c_3$, i.e.,
that $c_4=1$ and $c_5=0$.
Hence $\Phi(u\p_u)\in\langle u\p_u\rangle+\mathfrak m_2+\mathfrak m_3=\mathfrak m_7$.
Thus, we get that the span~$\mathfrak m_7$ is stable under any $\Phi\in{\rm Aut}(\mathfrak g)$.
\end{proof}

Fixing a basis of a Lie algebra, we identify automorphisms of this algebra with their matrices.
In the algebra~$\mathfrak g$, we fix the basis
$(\p_u,x\p_u,(tx-y)\p_u,\p_y,u\p_u,\p_t,\p_x+t\p_y+\p_\alpha)$.
In view of the presence of the megaideals $\mathfrak m_1$, \dots, $\mathfrak m_7$,
any automorphism $\Phi=(a_{ij})_{i,j=1,\dots,7}$ of~$\mathfrak g$
satisfies the following constraints on its entries:
\begin{gather*}
a_{ij}=0,\ \ 1\leqslant j<i\leqslant 7,\quad
a_{24}=a_{34}=a_{45}=a_{36}=a_{56}=a_{27}=a_{37}=a_{57}=a_{67}=0.
\end{gather*}
We take into account these constraints and proceed with the direct construction of the group ${\rm Aut}(\mathfrak g)$
using a symbolic computation program.
As a result, we obtain that the automorphism group ${\rm Aut}(\mathfrak g)$ of~$\mathfrak g$
consists of the matrices
\begin{gather*}
\begin{pmatrix}
a_{33}a_{66}a_{77}  & -a_{33}a_{46}  & a_{13}                 & a_{35}a_{66}a_{77} & a_{15} &  a_{35}a_{46} &  a_{35}a_{47}-a_{25}a_{77}\\
      0             &  a_{33}a_{66}  & a_{33}a_{47}a_{77}^{-1}& 0                  & a_{25} & -a_{35}a_{66} &  0                        \\
      0             &  0             & a_{33}                 & 0                  & a_{35} & 0             &  0                        \\
      0             &  0             & 0                      & a_{66}a_{77}       & 0      & a_{46}        &  a_{47}                   \\
      0             &  0             & 0                      & 0                  & 1      & 0             &  0                        \\
      0             &  0             & 0                      & 0                  & 0      & a_{66}        &  0                        \\
      0             &  0             & 0                      & 0                  & 0      & 0             &  a_{77}                   \\
\end{pmatrix},
\end{gather*}
where the remaining parameters~$a_{ij}$ are arbitrary real constants with $a_{33}a_{66}a_{77}\ne0$.
It is clear from the form of automorphisms that
a complete list of essential megaideals of the algebra $\mathfrak g$,
which are not the sums of other megaideals, is exhausted by that constructed above:
\begin{gather*}
\mathfrak m_1,\quad
\mathfrak m_2,\quad
\mathfrak m_3,\quad
\mathfrak m_4,\quad
\mathfrak m_5,\quad
\mathfrak m_6,\quad
\mathfrak m_7.
\end{gather*}
\noprint{
\begin{gather*}
\mathfrak m_1=\langle\p_u\rangle,\quad
\mathfrak m_2=\langle\p_u,x\p_u\rangle,\quad
\mathfrak m_3=\langle\p_u,\p_y\rangle,\quad
\mathfrak m_4=\langle\p_u,x\p_u,(tx-y)\p_u\rangle,\\
\mathfrak m_5=\langle\p_u,\p_y,\p_x+t\p_y+\p_\alpha\rangle,\quad
\mathfrak m_6=\langle\p_u,x\p_u,\p_y,\p_t\rangle,\quad
\mathfrak m_7=\langle\p_u,x\p_u,(tx-y)\p_u,u\p_u\rangle.
\end{gather*}
}
The remaining nonzero megaideals of the algebra $\mathfrak g$ are
\noprint{
$\mathfrak m_2+\mathfrak m_3$,
$\mathfrak m_3+\mathfrak m_4$,
$\mathfrak m_2+\mathfrak m_5$,
$\mathfrak m_3+\mathfrak m_7$,
$\mathfrak m_4+\mathfrak m_6$,
$\mathfrak m_4+\mathfrak m_5$,
$\mathfrak m_5+\mathfrak m_6$,
$\mathfrak m_5+\mathfrak m_7$,
$\mathfrak m_6+\mathfrak m_7$,
$\mathfrak m_4+\mathfrak m_5+\mathfrak m_6$,
$\mathfrak m_5+\mathfrak m_6+\mathfrak m_7=\mathfrak g$.
}
\begin{gather*}
\mathfrak m_2+\mathfrak m_3,\quad
\mathfrak m_3+\mathfrak m_4,\quad
\mathfrak m_2+\mathfrak m_5,\quad
\mathfrak m_3+\mathfrak m_7,\quad
\mathfrak m_4+\mathfrak m_6,\quad
\mathfrak m_4+\mathfrak m_5,\quad
\mathfrak m_5+\mathfrak m_6,\\
\mathfrak m_5+\mathfrak m_7,\quad
\mathfrak m_6+\mathfrak m_7,\quad
\mathfrak m_4+\mathfrak m_5+\mathfrak m_6,\quad
\mathfrak m_5+\mathfrak m_6+\mathfrak m_7=\mathfrak g.
\end{gather*}
\noprint{
\begin{gather*}
\mathfrak m_2+\mathfrak m_3=\langle\p_u,x\p_u,\p_y\rangle,\quad
\mathfrak m_3+\mathfrak m_4=\langle\p_u,x\p_u,(tx-y)\p_u,\p_y\rangle,
\\
\mathfrak m_2+\mathfrak m_5=\langle\p_u,x\p_u,\p_y,\p_x+t\p_y+\p_\alpha\rangle,\quad
\mathfrak m_3+\mathfrak m_7=\langle\p_u,x\p_u,(tx-y)\p_u,\p_y,u\p_u\rangle,
\\
\mathfrak m_4+\mathfrak m_5=\langle\p_u,x\p_u,(tx-y)\p_u,\p_y,\p_x+t\p_y+\p_\alpha\rangle,
\\
\mathfrak m_4+\mathfrak m_6=\langle\p_u,x\p_u,(tx-y)\p_u,\p_y,\p_t\rangle,
\\
\mathfrak m_5+\mathfrak m_6=\langle\p_u,x\p_u,\p_y,\p_t,\p_x+t\p_y+\p_\alpha\rangle,
\\
\mathfrak m_5+\mathfrak m_7=\langle\p_u,x\p_u,(tx-y)\p_u,\p_y,u\p_u,\p_x+t\p_y+\p_\alpha\rangle,
\\
\mathfrak m_6+\mathfrak m_7=\langle\p_u,x\p_u,(tx-y)\p_u,\p_y,u\p_u,\p_t\rangle,
\\
\mathfrak m_4+\mathfrak m_5+\mathfrak m_6=\langle\p_u,x\p_u,(tx-y)\p_u,\p_y,\p_t,\p_x+t\p_y+\p_\alpha\rangle,\quad
\mathfrak m_5+\mathfrak m_6+\mathfrak m_7=\mathfrak g.
\end{gather*}
}

\begin{lemma}
A complete list of independent discrete equivalence transformations of the class $\mathcal F$ of Kolmogorov backward equations
with power diffusivity~\eqref{eq:FPsubclass} is exhausted by two involutions,
one alternating the sign of the dependent variable
and one simultaneously alternating the signs of $x$, $y$ and $\alpha$,
$(t,x,y,u,\alpha,\beta)\mapsto(t,x,y,-u,\alpha,\beta)$
and
$(t,x,y,u,\alpha,\beta)\mapsto(t,-x,-y,u,-\alpha,\beta)$.
\end{lemma}

\begin{proof}
The inner automorphism group ${\rm Inn}(\mathfrak g)$ is constituted by the matrices of the form
\begin{gather*}
\begin{pmatrix}
{\rm e}^{\delta_5}&\delta_7{\rm e}^{\delta_5}& -\delta_4{\rm e}^{\delta_5}& \delta_3&\delta_1&-\delta_3\delta_7&\delta_3\delta_6-\delta_2\\
      0           &{\rm e}^{\delta_5}        &  \delta_6{\rm e}^{\delta_5}& 0       &\delta_2&-\delta_3        & 0                       \\
      0           &  0                       &  {\rm e}^{\delta_5}        & 0       &\delta_3& 0               & 0                       \\
      0           &  0                       & 0                          & 1       & 0      &-\delta_7        & \delta_6                \\
      0           &  0                       & 0                          & 0       & 1      & 0               & 0                       \\
      0           &  0                       & 0                          & 0       & 0      & 1               & 0                       \\
      0           &  0                       & 0                          & 0       & 0      & 0               & 1                       \\
\end{pmatrix},
\end{gather*}
where the parameters~$\delta_1$, \dots, $\delta_7$ are arbitrary real constants.
Consider the subgroup $D\subset{\rm Aut}(\mathfrak g)$ consisting of the diagonal matrices
\[
\mathop{\rm diag}(\varepsilon b_1b_2,\varepsilon b_1,\varepsilon,b_1b_2,1,b_1,b_2),
\]
where $b_1$ and~$b_2$ are arbitrary real constants with $b_1b_2\ne0$,
and $\varepsilon=\pm1$.
The group of outer automorphisms ${\rm Out}(\mathfrak g):={\rm Aut}(\mathfrak g)/{\rm Inn}(\mathfrak g)$
is canonically isomorphic to~$D$.
Moreover, the group ${\rm Aut}(\mathfrak g)$ splits over its normal subgroup ${\rm Inn}(\mathfrak g)$,
${\rm Aut}(\mathfrak g)=D\ltimes{\rm Inn}(\mathfrak g)$.

The general form of a fiber-preserving point transformation~$\mathcal T$
acting on the foliated space $\mathbb R^4_{t,x,y,u}\times\mathbb R^2_{\alpha,\beta}$ is
\begin{gather*}
\mathcal T\colon\quad
(\tilde t,\tilde x,\tilde y,\tilde u)=(T,X,Y,U)(t,x,y,u),\quad
(\tilde\alpha,\tilde\beta)=(A,B)(t,x,y,u,\alpha,\beta).
\end{gather*}
If $\mathcal T$ is an equivalence transformation of the class~$\mathcal F$,
then the pushforward~$\mathcal T_*$ of vector fields by~$\mathcal T$ is an automorphism of the algebra~$\mathfrak g$.
Up factoring out continuous equivalence transformations, which are known,
we can assume that $\mathcal T_*\in D$.
Therefore,
\begin{gather*}
\mathcal T_*(\p_u)      =\varepsilon b_1b_2\p_{\tilde u},\quad
\mathcal T_*(x\p_u)     =\varepsilon b_1\tilde x\p_{\tilde u},\quad
\mathcal T_*((tx-y)\p_u)=\varepsilon(\tilde t\tilde x-\tilde y)\p_{\tilde u},\quad
\mathcal T_*(\p_y)      =b_1b_2\p_{\tilde y}
\\
\mathcal T_*(u\p_u)     =\tilde u\p_{\tilde u},\quad
\mathcal T_*(\p_t)      =b_1\p_{\tilde t},\quad
\mathcal T_*(\p_x+t\p_y+\p_\alpha)=b_2(\p_{\tilde x}+\tilde t\p_{\tilde y}+\p_{\tilde\alpha}).
\end{gather*}
This results in a system of determining equations for the components of $\mathcal T$,
\begin{gather*}
X_t=Y_t=U_t=A_t=B_t=0,\quad
T_t=b_1,
\\
T_x=U_x=B_x+B_\alpha=0,\quad
X_x=A_x+A_\alpha=b_2,\quad
Y_x+tY_y=b_2T,
\\
T_y=X_y=U_y=A_y=B_y=0,\quad
Y_y=b_1b_2,
\\
T_u=X_u=Y_u=A_u=B_u=0,\quad
U_u=\varepsilon b_1b_2,
\\
xU_u=\varepsilon b_1X,
\quad
(tx-y)U_u=\varepsilon(TX-Y),
\quad
uU_u=U,
\end{gather*}
whose general solution is
\begin{gather*}
T=b_1t,\quad
X=b_2x,\quad
Y=b_1b_2y,\quad
U=\varepsilon b_1b_2u,\\
A=b_2\alpha+A^0(\omega,\beta),\quad
B=B(\omega,\beta),\quad \omega:=x-\alpha.
\end{gather*}
Using the chain rule we express all the required transformed derivatives in terms of the initial ones,
$\tilde u_{\tilde t}=\varepsilon b_2u_t$,
$\tilde u_{\tilde y}=\varepsilon u_y$,
$\tilde u_{\tilde x\tilde x}=\varepsilon b_1b_2^{-1}u_{xx}$,
$\tilde\alpha_{\tilde x}=\alpha_x+b_2^{-1}\big(A^0_\omega(1-\alpha_x)+A^0_\beta\beta_x\big)$ and
$\tilde\beta_{\tilde x}=b_2^{-1}\big(B_\omega(1-\alpha_x)+B_\beta\beta_x\big)$.
Then substitute the obtained expressions into the transformed counterpart of~\eqref{eq:FPsubclass}.
The resulted equation should be satisfied identically for all the solutions of~\eqref{eq:FPsubclass},
which results in the constraints $b_1=1$, $b_2^2=1$, $A^0=0$ and $B=\beta$.
Therefore, the class $\mathcal F$ possesses only two independent discrete point equivalence transformations,
which are the involutions~$\mathscr I_u$ and~$\mathscr I_{\rm s}$
associated with the values $\varepsilon=-1$ and $b_2=-1$, respectively,
where the other parameters take the values as for the identity transformation,
see Corollary~\ref{cor:DisceteEquivTransOfF}.
\end{proof}

\subsection{The gauged class}\label{sec:AlgMetClassF'}

Similarly to what has been done in Section~\ref{sec:AlgMetClassF},
to find the automorphism group ${\rm Aut}(\hat{\mathfrak g})$ of the algebra $\hat{\mathfrak g}:=\mathfrak g^\sim_{\mathcal F'}$,
we start with constructing megaideals of the algebra~$\hat{\mathfrak g}$.
Since the algebra is not abelian, we straightforwardly obtain the megaideal
\begin{gather*}
\hat{\mathfrak m}_2:=\hat{\mathfrak g}'=\langle\p_u,x\p_u,(tx-y)\p_u\rangle.
\end{gather*}
Recall that a nilradical of a Lie algebra and its derivatives are megaideals of the algebra.
The nilradical of the algebra~$\hat{\mathfrak g}$ and its derivative respectively are
\begin{gather*}
\hat{\mathfrak n}=\langle\p_u,x\p_u,(tx-y)\p_u,\p_t,\p_y\rangle\quad\mbox{and}\quad
\hat{\mathfrak m}_1:=\hat{\mathfrak n}'=\langle\p_u,x\p_u\rangle.
\end{gather*}
Applying~\cite[Proposition~1]{card2013a} to the triple
$(\mathfrak i_0,\mathfrak i_1,\mathfrak i_2)=(\hat{\mathfrak n},\hat{\mathfrak g},\hat{\mathfrak m}_1)$,
we obtain the megaideal
\begin{gather*}
\hat{\mathfrak m}_4=\langle\p_u,x\p_u,\p_t,\p_y\rangle.
\end{gather*}

\begin{lemma}
The span $\hat{\mathfrak m}_3=\langle\p_u,x\p_u,(tx-y)\p_u,u\p_u\rangle$
is a megaideal of the algebra~$\hat{\mathfrak g}$.
\end{lemma}

\begin{proof}
Since $u\p_u\in\hat{\mathfrak g}\setminus\hat{\mathfrak n}$ and $(tx-y)\p_u\in\hat{\mathfrak m}_2\setminus\hat{\mathfrak m}_1$,
for any $\Phi\in{\rm Aut}(\hat{\mathfrak g})$ we have
\begin{gather*}
\Phi(u\p_u)=c_1\p_u+c_2x\p_u+c_3(tx-y)\p_u+c_4u\p_u+c_5\p_y+c_6\p_t\quad\mbox{and}
\\
\Phi((tx-y)\p_u)=\hat c_1\p_u+\hat c_2x\p_u+\hat c_3(tx-y)\p_u,
\end{gather*}
where $c_1$,\dots , $c_6$, $\hat c_1$,$\hat c_2$ and $\hat c_3$ are real constants with $\hat c_3c_4\ne0$.
Therefore,
\begin{gather*}
\Phi([(tx-y)\p_u,u\p_u])=\Phi((tx-y)\p_u)=\hat c_1\p_u+\hat c_2x\p_u+\hat c_3(tx-y)\p_u
\\
\qquad=\big[\Phi((tx-y)\p_u),\Phi(u\p_u)\big]=c_4(\hat c_1\p_u+\hat c_2x\p_u+\hat c_3(tx-y)\p_u)+\hat c_3c_5\p_u-\hat c_3c_6x\p_u,
\end{gather*}
which results in the equations  $\hat c_1=\hat c_1c_4+\hat c_3c_5$, $\hat c_2=\hat c_2c_4-\hat c_3c_6$ and $\hat c_3=\hat c_3c_4$,
i.e., $c_4=1$ and $c_5=c_6=0$.
Hence, $\Phi(u\p_u)\in\langle u\p_u\rangle+\hat{\mathfrak m}_2$.
Thus, we get that $\hat{\mathfrak m}_4$ is stable under the action of the group ${\rm Aut}(\hat{\mathfrak g})$.
\end{proof}

Fixing the basis $(\p_u,x\p_u,(tx-y)\p_u,u\p_u,\p_t,\p_y)$ of the algebra~$\hat{\mathfrak g}$
and using the knowledge of the proper megaideals $\hat{\mathfrak m}_1$, \dots, $\hat{\mathfrak m}_4$,
we can impose the following constraints on the matrix $(a_{ij})_{i,j=1,\dots,6}$
of an arbitrary automorphism of~$\hat{\mathfrak g}$ in this basis:
\begin{gather*}
a_{i1}=a_{i2}=a_{j3}=a_{j4}=a_{k5}=a_{k6}=0,\ \ i=3,4,5,6,\ \ j=5,6,\ \ k=3,4,\quad
a_{65}=0.
\end{gather*}
Under the derived constraints, we complete the computation of~${\rm Aut}(\hat{\mathfrak g})$ in {\sf Maple}.
We obtain that the group ${\rm Aut}(\hat{\mathfrak g})$ consists of the  matrices
\begin{gather*}
\begin{pmatrix}
 a_{44}a_{55} & -a_{43}a_{55}  & a_{43}a_{56} & a_{44}a_{56} &  a_{15}  & a_{16}  \\
-a_{34}a_{55} &  a_{33}a_{55}  &-a_{33}a_{56} &-a_{34}a_{56} &  a_{25}  & a_{26}  \\
     0        &        0       &  a_{33}      &   a_{34}     &   0      &   0     \\
     0        &        0       &  a_{43}      &   a_{44}     &   0      &   0     \\
     0        &        0       &    0         &   0          &  a_{55}  &  a_{56} \\
     0        &        0       &    0         &   0          &   0      &   1     \\
\end{pmatrix},
\end{gather*}
where the remaining parameters~$a_{ij}$ are arbitrary real constants with $a_{33}(a_{55}a_{66}-a_{56}a_{65})\neq0$.
It is then clear that the complete list of the essential megaideals of~$\hat{\mathfrak g}$,
which are not the sums of other megaideals, is exhausted by
$\hat{\mathfrak m}_1$,
$\hat{\mathfrak m}_2$,
$\hat{\mathfrak m}_3$ and
$\hat{\mathfrak m}_4$.
The only other nonzero megaideals are
$\hat{\mathfrak n}=\hat{\mathfrak m}_2+\hat{\mathfrak m}_4$ and
$\hat{\mathfrak g}=\hat{\mathfrak m}_3+\hat{\mathfrak m}_4$.
This is why the nilradical~$\hat{\mathfrak n}$ and the entire algebra~$\hat{\mathfrak g}$
are not essential in the course of computing
the automorphisms of the algebra~$\hat{\mathfrak g}$
and the equivalence group~$G^\sim_{\mathcal F'}$ by the algebraic method.

The inner automorphism group ${\rm Inn}(\hat{\mathfrak g})$ is constituted by the matrices
\begin{gather*}
\begin{pmatrix}
{\rm e}^{\delta_6} &          0           &     0       & \delta_5 & -\delta_4{\rm e}^{\delta_6} &-\delta_4\delta_5+\delta_1 \\
        0           & {\rm e}^{\delta_6}  &  -\delta_5  &     0    &  \delta_3{\rm e}^{\delta_6} & \delta_3\delta_5+\delta_2 \\
        0           &          0          &     1       &     0    &               0             &             0             \\
        0           &          0          &     0       &     1    &               0             &             0             \\
        0           &          0          &     0       &     0    &      {\rm e}^{\delta_6}     &         \delta_5          \\
        0           &          0          &     0       &     0    &               0             &             1             \\
\end{pmatrix},
\end{gather*}
where the parameters $\delta_1$, \dots, $\delta_6$ are arbitrary real constants.
The group ${\rm Aut}(\hat{\mathfrak g})$ splits over its normal subgroup ${\rm Inn}(\hat{\mathfrak g})$,
${\rm Aut}(\hat{\mathfrak g})=H\ltimes{\rm Inn}(\hat{\mathfrak g})$.
The subgroup $H\subset{\rm Aut}(\hat{\mathfrak g})$ consists of the matrices
\begin{gather*}
\begin{pmatrix}
 \varepsilon c_{22}&-\varepsilon c_{21} \\
-\varepsilon c_{12}& \varepsilon c_{11}
\end{pmatrix}
\oplus
\begin{pmatrix}
c_{11} & c_{12} \\
c_{21} & c_{22}
\end{pmatrix}\oplus{\rm diag}(\varepsilon,1),
\noprint{
\begin{pmatrix}
 \varepsilon c_{22}&-\varepsilon c_{21} &   0    &   0    &      0      & 0 \\
-\varepsilon c_{12}& \varepsilon c_{11} &   0    &   0    &      0      & 0 \\
      0            &        0           & c_{11} & c_{12} &      0      & 0 \\
      0            &        0           & c_{21} & c_{22} &      0      & 0 \\
      0            &        0           &   0    &   0    & \varepsilon & 0 \\
      0            &        0           &   0    &   0    &      0      & 1
\end{pmatrix},
}
\end{gather*}
where the parameters $c_{kl}$, $k,l=1,2$, are arbitrary real constants with $c_{11}c_{22}-c_{21}c_{12}\ne0$,
and $\varepsilon=\pm1$.
The general form of a point fiber-preserving transformation~$\mathcal T$ in the foliated space $\mathbb R^4_{t,x,y,u}\times\mathbb R_\beta$~is
\begin{gather*}
\mathcal T\colon\quad
(\tilde t,\tilde x,\tilde y,\tilde u)=(T,X,Y,U)(t,x,y,u),\quad
\tilde\beta=B(t,x,y,u,\beta).
\end{gather*}
If~$\mathcal T$ is an equivalence transformation of the class~$\mathcal F'$,
then the pushforward~$\mathcal T_*$ of vector fields by~$\mathcal T$ is an automorphism of the algebra~$\hat{\mathfrak g}$.
Moreover, up to composing~$\mathcal T$ with elements of the identity component of~$G^\sim_{\mathcal F'}$,
which is known if $\hat{\mathfrak g}$ is,
we can assume that $\mathcal T_*\in H$.
Therefore,
\begin{gather*}
\mathcal T_*(\p_u)      =\varepsilon c_{22}\p_u-\varepsilon c_{12}x\p_u,\quad
\mathcal T_*(x\p_u)     =-\varepsilon c_{21}\p_u+\varepsilon c_{11}x\p_u,
\\
\mathcal T_*(\p_t) =c_{11}\p_t+c_{21}\p_y,\quad
\mathcal T_*(\p_y) =c_{12}\p_t+c_{22}\p_y,
\\
\mathcal T_*((tx-y)\p_u)=\varepsilon(tx-y)\p_u,\quad
\mathcal T_*(u\p_u)=u\p_u.
\end{gather*}
This leads to a system of linear differential equations on the components of~$\mathcal T$,
\begin{gather*}
X_t=U_t=B_t=0,\quad
T_t=c_{11},\quad
Y_t=c_{21},
\\
X_y=U_y=B_y=0,\quad
T_y=c_{12},\quad
Y_y=c_{22},
\\
T_u=X_u=Y_u=B_u=0,\quad
U_u=\varepsilon(c_{22}-c_{12}X),
\\
xU_u=\varepsilon(c_{11}X-c_{21}),\quad
uU_u=U,\quad
(tx-y)U_u=\varepsilon(TX-Y),
\end{gather*}
which integrates to
\begin{gather*}
T=c_{11}t+c_{12}y,\quad
X=\frac{c_{22}x+c_{21}}{c_{12}x+c_{11}},\quad
Y=c_{21}t+c_{22}y,\quad
U=\varepsilon\frac{c_{11}c_{22}-c_{12}c_{21}}{c_{12}x+c_{11}}u,
\\
B=B(x,\beta).
\end{gather*}
This is the most restricted form of candidates for equivalence transformations of~$\mathcal F'$
that can be obtained within the framework of the algebraic method.
We complete the computation of~$G^\sim_{\mathcal F'}$, applying the direct method.
Using the chain rule, we find the transformation components for involved derivatives,
\begin{gather*}
\tilde u_{\tilde t}=\varepsilon\frac{c_{22}u_t-c_{21}u_y}{c_{12}x+c_{11}},\quad
\tilde u_{\tilde y}=\varepsilon\frac{c_{11}u_y-c_{12}u_t}{c_{12}x+c_{11}},\quad
\tilde u_{\tilde x\tilde x}=\frac\varepsilon\Delta(c_{12}x+c_{11})^3u_{xx},\\
\tilde \beta_{\tilde x}=\frac{(c_{12}x+c_{11})^2}\Delta(B_x+B_\beta\beta_x),
\end{gather*}
where $\Delta:=c_{11}c_{22}-c_{12}c_{21}$.
The transformed counterpart $\tilde u_{\tilde t}+\tilde x\tilde u_{\tilde y}=|\tilde x|^{\tilde\beta}\tilde u_{\tilde x\tilde x}$
of the equation~\eqref{eq:FPGaugedClass}
in the initial coordinates takes the form
\begin{gather*}
u_t+xu_y=\left|\frac{c_{22}x+c_{21}}{c_{12}x+c_{11}}\right|^{\tilde\beta}\frac{(c_{12}x+c_{11})^5}{\Delta^2}u_{xx},
\end{gather*}
which should coincide with~\eqref{eq:FPGaugedClass}.
This results in the equation
\[
|c_{22}x+c_{21}|^{\tilde\beta}=\Delta^2\sgn(c_{12}x+c_{11})|x|^\beta|c_{12}x+c_{11}|^{\tilde\beta-5}.
\]
Hence $\sgn(c_{12}x+c_{11})=1$.
In view of Lemma~\ref{lem:EqualityForArbitraryX} and the inequality $c_{11}c_{22}-c_{21}c_{12}\ne0$,
in the case $c_{22}=0$, we successively derive
$c_{12}c_{21}\ne0$, $c_{11}=0$, $c_{12}=\sgn x$, $c_{21}=\pm\sgn x$ and $\beta=5-\tilde\beta$.
Analogously, for $c_{22}\ne0$ we obtain $c_{12}=0$, $c_{11}\ne0$, $c_{21}=0$, $c_{11}=1$, $\tilde\beta=\beta$,
$|c_{22}|^{\beta-2}=1$ and thus $c_{22}=\pm1$.
Hence, there are only three independent discrete equivalence transformations in~$G^\sim_{\mathcal F'}$,
$\mathscr I'_u$, $\mathscr I'_{\rm s}$ and~$\mathscr J'$,
which are respectively associated with the values $\varepsilon=-1$, $(\varepsilon,c_{22})=(-1,1)$
(if the other parameters in these two cases take the values as for the identity transformation)
and $(c_{11},c_{12},c_{21},c_{22},\varepsilon)=(0,\sgn x,\sgn x,0,-\sgn x)$,
see Corollary~\ref{cor:DisceteEquivTransOfF'}.

\section{Darboux transformations\\ of linear heat equations with potentials}\label{sec:DarbouxTrans}

Consider a (1+1)-dimensional linear heat equations with an arbitrary potential $V=V(t,x)$,
\begin{gather}\label{eq:HeatWithPot}
u_t-u_{xx}+V(t,x)u=0.
\end{gather}
By $\mathrm W(f^1,\dots,f^k)$ we denote the Wronskian of sufficiently smooth functions~$f^1$, \dots, $f^k$ of $(t,x)$
with respect to~$x$, $\mathrm W(f^1,\dots,f^k)=\det(\p_x^{l'-1}f^l)_{l,l'=1,\dots,k}$.
The Darboux transformation
\[
\tilde u=\mathrm{DT}[f^1,\dots,f^k]u:=\frac{\mathrm W(f^1,\dots,f^k,u)}{\mathrm W(f^1,\dots,f^k)},
\]
with a fixed tuple $(f^1,\dots,f^k)$ of solutions of~\eqref{eq:HeatWithPot}
maps the solution set of~\eqref{eq:HeatWithPot} onto the solution set of the equation
$\tilde u_t-\tilde u_{xx}+\tilde V(t,x)\tilde u=0$, where
\[
\tilde V=V-2\left(\frac{\big(\mathrm W(f^1,\dots,f^k)\big)_x}{\mathrm W(f^1,\dots,f^k)}\right)_x,
\]
see \cite{matv1991A,popo2008a}.
An evident relation is established by Darboux transformations
between (1+1)-dimensional linear heat equations with inverse square potentials,
$V=-\mu x^{-2}$, where $\mu$ is an arbitrary constant with $\mu\leqslant\frac14$.
The function $u=|x|^\alpha$ with $\alpha=\frac12\pm(\frac14-\mu)^{1/2}$
is a solution of the equation
\begin{gather}\label{eq:LinHeatWithInverseSquarePot}
u_t=u_{xx}+\mu x^{-2}u,
\end{gather}
and the corresponding Darboux transformation $\mathrm{DT}[|x|^\alpha]$
maps this equation to the equation $\tilde u_t=\tilde u_{xx}+\tilde\mu x^{-2}\tilde u$
with $\tilde\mu=\mu-2\alpha$.

The most interesting is the case when solutions of (1+1)-dimensional linear heat equations with inverse square potentials
are expressed in terms of the general solution of the (1+1)-dimensional linear heat equation.
Since the Darboux transformations $\mathrm{DT}[x^n]$ maps
the equation of the form~\eqref{eq:LinHeatWithInverseSquarePot} with $\mu=n(n-1)$
to the equation of the same form with $\mu=n(n+1)$.
As a result, beginning from $n=1$ we derive the representation
\begin{gather}\label{eq:LinHeatWithInverseSquarePotN(N-1)}
u=\bigg(\p_x-\frac nx\bigg)\bigg(\p_x-\frac{n-1}x\bigg)\cdots\bigg(\p_x-\frac1x\bigg)\vartheta^0(t,x)\
=\mathrm{DT}[P_1,P_3,\dots,P_{2n-1}]\vartheta^0(t,x)
\end{gather}
for an arbitrary solution of the equation of the form~\eqref{eq:LinHeatWithInverseSquarePot} with $\mu=n(n+1)$,
where $\vartheta^0$ is an arbitrary solution of the (1+1)-dimensional linear heat equation, for which $\mu=0$.
Here $P_k$ is the canonical heat polynomial of degree~$k\in\mathbb N\cup\{0\}$,
which is, for odd~$k=2m-1$ with $m\in\mathbb N$, of the form
\begin{gather*}
P_{2m-1}(t,x)=\frac{x^{2m-1}}{(2m-1)!}+\frac{t}{1!}\frac{x^{2m-3}}{(2m-3)!}
+\frac{t^2}{2!}\frac{x^{2m-5}}{(2m-5)!}+ \cdots +
\frac{t^{m-2}}{(m-2)!}\frac{x^3}{3!}+\frac{t^{m-1}}{(m-1)!}\frac{x}{1!}.
\end{gather*}
To check the second equality in~\eqref{eq:LinHeatWithInverseSquarePotN(N-1)},
it suffices to prove that
\[
R_mR_{m-1}\cdots R_1P_{2m-1}=0,
\quad\mbox{where}\quad R_l:=\p_x-\frac lx,\quad l\in\mathbb N,
\]
which we carry out by induction.
The induction base $m=1$ is clear since $P_1(t,x)=x$ and $(\p_x-x^{-1})P_1=0$.
The induction step follows from the following facts:
\begin{gather*}
P_k(t,x)=\frac{2^k}{k!}\mathrm G^k1,\quad\mbox{where}\quad \mathrm G:=t\p_x+\frac x2,\quad
\p_xP_k=P_{k-1},
\\
R_mR_{m-1}\cdots R_1\mathrm G
=\bigg(tR_m+\frac x2\bigg)R_{m-1}\cdots R_1\p_x.
\end{gather*}

\section*{Acknowledgments}
The authors are grateful to Alexander Bihlo for valuable discussions.
This research was undertaken thanks to funding from the Canada Research Chairs program and the NSERC Discovery Grant program.
It was also supported in part by the Ministry of Education, Youth and Sports of the Czech Republic (M\v SMT \v CR)
under RVO funding for I\v C47813059
and by the National Academy of Science of Ukraine under the project 0116U003059.
ROP expresses his gratitude for the hospitality shown by the University of Vienna during his long stay at the university.



\footnotesize
\begin{thebibliography}{10}\itemsep=0.3ex

\bibitem{anco1997b}
Anco S.C. and Bluman G.,
Nonlocal symmetries and nonlocal conservation laws of Maxwell's equations,
{\it J.~Math. Phys.} {\bf 38} (1997), 3508--3532.

\bibitem{baru2001a}
Barucci E., Polidoro S. and Vespri V.,
Some results on partial differential equations and Asian options,
{\it Math. Models Methods Appl. Sci.} {\bf 11} (2001) 475--497. 

\bibitem{berg1978a}
Bergman G.M.,
The diamond lemma for ring theory,
{\it Adv. in Math.} {\bf 29} (1978), 178--218.

\bibitem{bihl2012b}
Bihlo A., Dos Santos Cardoso-Bihlo E. and Popovych R.O.,
Complete group classification of a class of nonlinear wave equations,
{\it J.~Math. Phys.} {\bf 53} (2012), 123515, arXiv:1106.4801.

\bibitem{bihl2015a}
Bihlo A., Dos Santos Cardoso-Bihlo E. and Popovych R.O.,
Algebraic method for finding equivalence groups,
{\it J.~Phys. Conf. Ser.} {\bf 621} (2015) 012001, arXiv:1503.06487.

\bibitem{bihl2011b}
Bihlo A. and Popovych R.O.,
Point symmetry group of the barotropic vorticity equation,
in {\it Proceedings of 5th Workshop ``Group Analysis of Differential Equations \& Integrable Systems''
(June 6--10, 2010, Protaras, Cyprus)}, University of Cyprus, Nicosia, 2011, pp. 15--27, arXiv:1009.1523.

\bibitem{bihlo2016a}
Bihlo A. and Popovych R.O.,
Group classification of linear evolution equations,
{\it J.~Math. Anal. Appl.} {\bf 448} (2017), 982--2015, arXiv:1605.09251.

\bibitem{blum2010A}
Bluman G.W., Cheviakov A.F. and Anco S.C.,
{\it Applications of symmetry methods to partial differential equations},
Springer, New York, 2010. 

\bibitem{blum1989A}
Bluman G.W. and Kumei S.,
{\it Symmetries and differential equations},
Springer, New York, 1989.

\bibitem{boyk2016a}
Boyko V.M., Kunzinger M. and Popovych R.O.,
Singular reduction modules of differential equations,
{\it J.~Math. Phys.} {\bf 57} (2016), 101503, arXiv:1201.3223.

\bibitem{boyk2022a}
Boyko V.M., Popovych R.O. and Vinnichenko O.O.,
Point- and contact-symmetry pseudogroups of dispersionless Nizhnik equation,
{\it Commun. Nonlinear Sci. Numer. Simul.} {\bf 132} (2024), 107915, arXiv:2211.09759. 

\bibitem{chri2003a}
Christodoulakis T., Papadopoulos G.O. and Dimakis A.,
Automorphisms of real four-dimensional Lie algebras and the invariant characterization of homogeneous 4-spaces,
{\it J.~Phys.~A} {\bf 36} (2003),  427--441; Corrigendum, {\it J.~Phys.~A} {\bf 36} (2003), 2379.

\bibitem{davy2015a}
Davydovych V.,
Preliminary group classification of $(2+1)$-dimensional linear ultraparabolic Kolmogorov--Fokker--Planck equations,
{\it Ufimsk. Mat. Zh.} {\bf 7} (2015), 38--46.  

\bibitem{card2011a}
Dos Santos Cardoso-Bihlo E., Bihlo A. and Popovych R.O.,
Enhanced preliminary group classification of a class of generalized diffusion equations,
{\it Commun. Nonlinear Sci. Numer. Simulat.} {\bf 16} (2011), 3622--3638, arXiv:1012.0297.

\bibitem{card2013a}
Dos Santos Cardoso-Bihlo E. and Popovych R.O.,
Complete point symmetry group of the barotropic vorticity equation on a rotating sphere,
{\it J.~Engrg. Math.} {\bf 82} (2013), 31--38, arXiv:1206.6919.

\bibitem{gonz1992b}
Gonz\'{a}lez-L\'{o}pez A., Kamran N. and Olver P.J.,
Lie algebras of vector fields in the real plane,
{\it Proc. London Math. Soc. (3)} {\bf 64} (1992), 339--368. 

\bibitem{gung2018a}
G{\"u}ng{\"o}r F.,
Equivalence and symmetries for variable coefficient linear heat type equations. I,
{\it J.~Math. Phys.} \textbf{59} (2018), 051507, arXiv:1501.01481.

\bibitem{gung2018b}
G\"{u}ng\"{o}r F.,
Equivalence and symmetries for variable coefficient linear heat type equations. II. Fundamental solutions,
{\it J.~Math. Phys.} {\bf 59} (2018), 061507. 

\bibitem{hilg2011A}
Hilgert J. and Neeb K.H.,
{\it Structure and geometry of Lie groups},
Springer, New York, 2012.

\bibitem{hydo1998a}
Hydon P.E.,
Discrete point symmetries of ordinary differential equations,
{\it Proc. R. Soc. Lond. Ser. A Math. Phys. Eng. Sci.} {\bf 454} (1998), 1961--1972. 

\bibitem{hydo1998b}
Hydon P.E.,
How to find discrete contact symmetries,
{\it J.~Nonlinear Math. Phys.} {\bf 5} (1998), 405--416. 

\bibitem{hydo2000b}
Hydon P.E.,
How to construct the discrete symmetries of partial differential equations,
{\it Eur. J. Appl. Math.} {\bf 11} (2000), 515--527. 

\bibitem{hydo2000A}
Hydon P.E.,
{\it Symmetry methods for differential equations. A beginner's guide},
Cambridge University Press, Cambridge, 2000.

\bibitem{kamk1977A}
Kamke E.,
{\it Differentialgleichungen. L\"osungsmethoden und L\"osungen. I: Gew\"ohnliche Differentialgleichungen},
B.\,G. Teubner, Stuttgart, 1977.

\bibitem{king1998a}
Kingston J.G. and Sophocleous C.,
On form-preserving point transformations of partial differential equations,
{\it J.~Phys.~A} {\bf 31} (1998), 1597--1619.

\bibitem{kolm1931a}
Kolmogoroff A.,
\"Uber die analytischen Methoden in der Wahrscheinlichkeitsrechnung,
{\it Math. Ann.} {\bf 104} (1931), 415--458.

\bibitem{kova2023a}
Koval S.D., Bihlo A. and Popovych R.O.,
Extended symmetry analysis of remarkable (1+2)-dimensional Fokker--Planck equation,
{\it European~J. Appl. Math.} {\bf 34} (2023), 1067--1098, arXiv:2205.13526.

\bibitem{kova2023b}
Koval S.D. and Popovych R.O.,
Point and generalized symmetries of the heat equation revisited,
{\it J.~Math. Anal. Appl.} {\bf 527} (2023), 127430, arXiv:2208.11073.

\bibitem{kova2024a}
Koval S.D. and Popovych R.O.,
Extended symmetry analysis of (1+2)-dimensional fine Kolmogorov backward equation,
{\it Stud. Appl. Math.} {\bf 153} (2024), e12695, arXiv:2402.08822.

\bibitem{kova2024b}
Koval S.D. and Popovych R.O.,
Linear generalized symmetries of linear systems of differential equations, in preparation.

\bibitem{kova2013a}
Kovalenko S.S., Kopas I.M. and Stogniy V.I.,
Preliminary group classification of a class of generalized linear Kolmogorov equations,
{\it Res. Bull. Natl. Tech. Univ. Ukraine ``Kyiv Polytechnic Institute''} {\bf 4} (2013), 67--72 (in Ukrainian).

\bibitem{lie1881a}
Lie S.,
\"Uber die Integration durch bestimmte Integrale von einer Klasse linear partieller Differentialgleichungen,
{\it Arch. for Math.} {\bf 6} (1881), 328--368;
translation by N.H. Ibragimov:
Lie S.,
On integration of a class of linear partial differential equations by means of definite integrals,
in {\it CRC Handbook of Lie group analysis of differential equations}, vol. 2, CRC Press, Boca Raton, FL, 1995, pp.~473--508.

\bibitem{lie1893A}
Lie S.,
{\it Theorie der Transformationsgruppen}, vol.~3, B.G. Teubner, Leipzig, 1893.

\bibitem{malt2021a}
Maltseva D.S. and Popovych R.O.,
Complete point-symmetry group, Lie reductions and exact solutions of Boiti--Leon--Pempinelli system,
{\it Phys.~D} {\bf 460} (2024), 134081, arXiv:2103.08734.

\bibitem{matv1991A}
Matveev V.B. and Salle M.A.,
{\it Darboux transformations and solitons},
Springer, Berlin, 1991.

\bibitem{moro1958a}
Morozov V.V.,
Classification of nilpotent Lie algebras of sixth order,
{\it Izv. Vys\v s. U\v cebn. Zaved. Matematika} (1958), no.~4, 161--171.

\bibitem{muba1963b}
Mubarakzjanov G.M.,
On solvable Lie algebras,
{\it Izv. Vys\v s. U\v cebn. Zaved. Matematika} (1963), no.~1, 114--123.

\bibitem{muba1963c}
Mubarakzjanov G.M.,
Classification of real structures of Lie algebras of fifth order,
{\it Izv. Vys\v s. U\v cebn. Zaved. Matematika} (1963), no.~3, 99--106.

\bibitem{muba1963a}
Mubarakzjanov G.M.,
Classification of solvable Lie algebras of sixth order with a non-nilpotent basis element,
{\it Izv. Vys\v s. U\v cebn. Zaved. Matematika} (1963), no.~4, 104--116.

\bibitem{olve1993A}
Olver P.J.,
{\it Application of Lie groups to differential equations},
Springer, New York, 2000.

\bibitem{opan2020a}
Opanasenko S., Bihlo A., Popovych R.O. and Sergyeyev A.,
Extended symmetry analysis of isothermal no-slip drift flux model,
{\it Phys. D} {\bf 402} (2020), 132188, arXiv:1705.09277. 

\bibitem{opan2022a}
Opanasenko S. and Popovych R.O.,
Mapping method of group classification,
{\it J. Math. Anal. Appl.} {\bf 513} (2022), 126209, arXiv:2109.11490. 

\bibitem{pate1977a}
Patera J. and Winternitz P.,
Subalgebras of real three and four-dimensional Lie algebras,
{\it J.~Math. Phys.} {\bf 18} (1977), 1449--1455.

\bibitem{popo2024b}
Popovych D.R., Koval S.D. and Popovych R.O.,
Generalized symmetries of remarkable (1+2)-dimensional Fokker-Planck equation,
{\it European J. Appl. Math.} (2025),
doi:\href{https://doi.org/10.1017/S0956792525100107}{10.1017/S0956792525100107},
arXiv:2409.10348.

\bibitem{popo2008b}
Popovych R.O.,
Reduction operators of linear second-order parabolic equations,
{\it J.~Phys.~A} {\bf 41} (2008), 185202, arXiv:0712.2764.

\bibitem{popo2003a}
Popovych R.O., Boyko V.M., Nesterenko M.O. and Lutfullin M.W.,
Realizations of real low-dimensional Lie algebras,
arXiv:math-ph/0301029v7 (2005) (extended and revised version of paper {\it J.~Phys.~A} {\bf 36} (2003), 7337--7360).

\bibitem{popo2010a}
Popovych R.O., Kunzinger M. and Eshraghi H.,
Admissible transformations and normalized classes of nonlinear Schr\"odinger equations,
{\it Acta Appl. Math.} {\bf 109} (2010), 315--359, arXiv:math-ph/0611061.

\bibitem{popo2008a}
Popovych R.O., Kunzinger M. and Ivanova N.M.
Conservation laws and potential symmetries of linear parabolic equations,
\emph{Acta Appl. Math.} {\bf 2} (2008), 113--185, arXiv:math-ph/0706.0443.

\bibitem{risk1989A}
Risken H.,
{\it The Fokker--Planck equation. Methods of solution and applications.}
Springer, Berlin, 1989.

\bibitem{rote83a}
Rotenberg M., Transport theory for growing cell population,
{\it J.~Theor. Biol.} {\bf 103} (1983), 181--199.

\bibitem{spic1997a}
Spichak S. and Stogny V.,
Symmetry analysis of the Kramers equation,
{\it Rep. Math. Phys.} {\bf 40} (1997), 125--130.

\bibitem{spic1999a}
Spichak S. and Stognii V.,
Symmetry classification and exact solutions of the one-dimensional Fokker--Planck equation with arbitrary coefficients of drift and diffusion,
{\it J. Phys. A} {\bf 32} (1999), 8341--8353.

\bibitem{spic2011a}
Spichak S., Stogniy V. and Kopas I.,
Symmetry properties and exact solutions of the linear Kolmogorov equation,
{\it Res. Bull. Natl. Tech. Univ. Ukraine ``Kyiv Polytechnic Institute''} {\bf 4} (2011), 93--97 (in Ukrainian).

\bibitem{vane2020a}
Vaneeva O.O., Bihlo A. and Popovych R.O.,
Generalization of the algebraic method of group classification with application to nonlinear wave and elliptic equations,
{\it Commun. Nonlinear Sci. Numer. Simul.} {\bf91} (2020), 105419,
arXiv:2002.08939.

\bibitem{vane2021a}
Vaneeva O.O., Popovych R.O. and Sophocleous C.,
Extended symmetry analysis of two-dimensional degenerate Burgers equation,
{\it J. Geom. Phys.} {\bf 169} (2021), 104336, arXiv:1908.01877. 

\bibitem{zhan2020a}
Zhang Z.-Y., Zheng J., Guo L.-L. and Wu H.-F.,
Lie symmetries and conservation laws of the Fokker--Planck equation with power diffusion,
{\it Math. Methods Appl. Sci.} {\bf 43} (2020), 8894--8905.

\end{thebibliography}
\end{document}